\documentclass[american,aps,pra,reprint,superscriptaddress,longbibliography]{revtex4-2}

\usepackage{amsmath,amssymb}

\usepackage{color}
\usepackage{graphicx}
\usepackage[american]{babel}
\usepackage[utf8]{inputenc}
\usepackage{times}
\usepackage{braket} 				
\usepackage{amsthm}
\usepackage{amsfonts}
\usepackage{mathrsfs}  
\usepackage{booktabs}
\usepackage{pgfplots}
\pgfplotsset{compat=1.17}
\usepackage{mathtools}
\usepackage{bbold}
\usepackage{titlesec}
\usepackage{bbm}

\newcommand{\ketbra}[2]{|#1\rangle\!\langle#2|}
\newcommand*\diff{\mathop{}\!\mathrm{d}}
\DeclareMathOperator*{\argmin}{arg\,min} 

\definecolor{mygrey}{gray}{0.35}
\definecolor{myblue}{rgb}{0.2,0.2,0.8}
\definecolor{myzard}{cmyk}{0,0,0.05,0}
\definecolor{mywhite}{rgb}{1,1,1}
\definecolor{myred}{rgb}{0.9,0.1,0.}
\usepackage[colorlinks=true,citecolor=myblue,linkcolor=myblue,urlcolor=myblue]{hyperref}

\newtheoremstyle{customStyle1}  
{0pt}       
{0pt}       
{\normalfont}   
{\parindent}        
{\em}  
{. --}   	 
{.5em}       
{\thmname{#1}\thmnumber{ #2}\thmnote{ (#3)}}  

\newcounter{theorems}
\newtheorem{theorem}[theorems]{Theorem}
\newtheorem{proposition}[theorems]{Proposition}
\newtheorem{corollary}[theorems]{Corollary}
\newtheorem{definition}[theorems]{Definition} 
\newtheorem{lemma}[theorems]{Lemma}

\def\pure{{\rm Pure}}
\def\pos{{\rm Pos}}
\def\cp{{\rm CP}}
\def\cptp{{\rm CPTP}}

\def\distill{{\rm Distill}}
\def\cost{{\rm Cost}}

\newcommand{\bea}{\begin{eqnarray}}
\newcommand{\eea}{\end{eqnarray}}
\newcommand{\be}{\begin{equation}}
\newcommand{\ee}{\end{equation}}
\newcommand{\ba}{\begin{equation}\begin{aligned}}
\newcommand{\ea}{\end{aligned}\end{equation}}

\newcommand{\epm}{\end{pmatrix}}
\newcommand{\bpm}{\begin{pmatrix}}

\newcommand{\ebm}{\end{bmatrix}}
\newcommand{\bbm}{\begin{bmatrix}}

\newcommand{\bex}{\begin{exercise}}
\newcommand{\eex}{\end{exercise}}

\newcommand{\ben}{\begin{enumerate}}
\newcommand{\een}{\end{enumerate}}

\def\be{\begin{equation}}
\def\ee{\end{equation}}

\newcommand{\rank}{{\rm Rank}}

\newcommand{\mb}{\mathfrak{B}}
\newcommand{\md}{\mathfrak{D}}

\newcommand{\mM}{\mathcal{M}}

\newcommand{\mW}{\mathcal{W}}

\newcommand{\la}{\langle}
\newcommand{\ra}{\rangle}

\newcommand{\da}{\downarrow}

\newcommand{\eps}{\varepsilon}

\newcommand{\mbb}[1]{\mathbb{#1}}

\newcommand{\eqdef}{\coloneqq}

\def\r{\mathbf{r}}

\def\p{\mathbf{p}}
\def\q{\mathbf{q}}

\def\e{\mathbf{e}}

\def\u{\mathbf{u}}

\def\0{\mathbf{0}}

\def\id{\mathsf{id}}

\def\mE{\mathcal{E}}

\def\mF{\mathcal{F}}
\def\mN{\mathcal{N}}

\def\mL{\mathcal{L}}

\def\mS{\mathcal{S}}

\def\mV{\mathcal{V}}
\def\mG{\mathcal{G}}

\def\tA{\tilde{A}}

\def\prob{{\rm Prob}}

\newcommand{\GG}[1]{\rm \textcolor{red}{ #1}}
\newcommand{\Gg}[1]{\textcolor{red}{ #1}}

\DeclareMathOperator{\tr}{Tr}

\DeclareMathOperator{\sr}{SR}
\newcommand{\norm}[1]{\left\lVert#1\right\rVert}
\DeclareMathOperator{\locc}{LOCC}

\usepackage[ruled]{algorithm2e}
\SetAlgoSkip{bigskip} 

\begin{document}
\title{Single-shot entanglement manipulation of states and channels revisited}
\author{Thomas Theurer}
\email{thomas.theurer@ucalgary.ca}
\affiliation{Department of Mathematics and Statistics, University of Calgary, Calgary, AB T2N 1N4, Canada}
\affiliation{Institute for Quantum Science and Technology, University of Calgary, Calgary, AB T2N 1N4, Canada}
\author{Kun Fang}
\affiliation{Institute for Quantum Computing, Baidu Research, Beijing 100193, China}
\author{Gilad Gour}
\affiliation{Department of Mathematics and Statistics, University of Calgary, Calgary, AB T2N 1N4, Canada}
\affiliation{Institute for Quantum Science and Technology, University of Calgary, Calgary, AB T2N 1N4, Canada}

\date{\today}

\begin{abstract}
    We study entanglement distillation and dilution of states and channels in the single-shot regime. With the help of a recently introduced conversion distance, we provide compact closed-form expressions for the dilution and distillation of pure states and show how this can be used to efficiently calculate these quantities on multiple copies of pure states. These closed-form expressions also allow us to obtain second-order asymptotics. We then prove that the $\eps$-single-shot entanglement cost of mixed states is given exactly in terms of an expression containing a suitably smoothed version of the conditional max-entropy. For pure states, this expression reduces to the smoothed max-entropy of the reduced state. Based on these results, we bound the single-shot entanglement cost of channels. We then turn to the one-way entanglement distillation of states and channels and provide bounds in terms of a quantity we denote coherent information of entanglement.
\end{abstract}
\date{\today}
\maketitle

\section{Introduction}
Quantum entanglement~\cite{Plenio2007,Horodecki2009} plays a fundamental role in many technological applications that involve two (or more) spatially separated parties such as quantum teleportation~\cite{Bennett1993}, superdense coding~\cite{Bennett1992}, and secure quantum communication~\cite{Ekert1991}. If these parties are restricted to local operations and classical communication (LOCC)~\cite{Bennett1996,Horodecki2009}, they can manipulate and consume entanglement, but they cannot create it. Entanglement is thus a valuable quantum resource~\cite{Chitambar2019}. The fact that the previously mentioned protocols typically require specific entangled states, most often in the form of pure, maximally entangled states, makes the interconversion between entangled states via LOCC an important primitive. The interconversion between entangled states is traditionally studied in two different limits~\cite{Horodecki2009}, either in the limit where one has access to (unboundedly) many identical and uncorrelated copies of a given initial state and tries to convert them to as many target states as possible or in the so-called single-shot regime: Here, one asks how well one can approximate a given target state with a single copy of an initial state and LOCC. Whilst the asymptotic regime provides ultimate bounds on the usefulness of a quantum state, the single-shot regime is closer to what is relevant from an experimental perspective, where one has only access to finitely many copies. 

The maximally entangled states of dimension $m$ play a prominent role, not only because they are required in many protocols, but also because they can be converted to all other states of lower or equal dimension~\cite{Nielsen1999}. The special cases of (approximate) entanglement interconversion where either the target or the initial state are maximally entangled are thus of special interest and studied under the name entanglement distillation and dilution~\cite{Bennett1996a, Bennett1996b, Rains1999, Hayden2001, fang2019non, Regula2019}. More precisely, single-shot entanglement dilution describes the task where two distant parties try to convert, up to a fixed error $\eps$, a maximally entangled state of dimension $m$ to a given target state via LOCC. The minimal $m$ such that this is possible is then identified with the $\eps$-single-shot entanglement cost of the target state. Conversely, the $\eps$-single-shot distillable entanglement of an initial state $\rho$ is determined by the maximal dimension of a maximally entangled state to which $\rho$ can be converted (again up to an error $\eps$ and via LOCC). Importantly, there are many different (but topologically equivalent) ways to define the error $\eps$, e.g., via the trace norm or the fidelity. In the following, we will consider two choices, one based on the star conversion distance recently introduced in Ref.~\cite{Zanoni2023}, and the other based on the fidelity. After an introduction of the necessary notation in Sec.~\ref{sec:Notation}, in Sec.~\ref{subsec:StarSingleShot}, we will provide closed-form expressions for the $\eps$-single-shot distillable entanglement of pure states with respect to the star conversion distance. Moreover, we show that these quantities can be efficiently computed on multiple copies of a given state and provide explicit algorithms to do so. This allows us to derive analytical expressions for the second-order correction terms of the asymptotic expressions~\cite{Kumagai2017}. 

We then move on to the $\eps$-single-shot entanglement cost of mixed states, where we restrict the error using the square of the purified distance~\cite{Buscemi2011}. As our main result, we prove that this entanglement cost can be expressed \textit{exactly} in terms of a smoothed version of the conditional max-entropy, strengthening a result by Buscemi and Datta~\cite{Buscemi2011}. We then show that on pure states, the two definitions of the $\eps$-single-shot entanglement cost coincide, are equal to the smoothed max-entropy as the reduced state, and compare our findings to previously known results.

Historically, entanglement was primarily studied in the framework of so-called static resource theories~\cite{Chitambar2019} which focus on the value of quantum states~\cite{Plenio2007,Horodecki2009}. However, in typical applications in which we hope for quantum advantages, we are interested in performing a task (such as sending a secure message) that is done with the help of a quantum channel (called a dynamical resource). As demonstrated by quantum teleportation, with LOCC, we can convert static entanglement present in quantum states into channels outside of LOCC and thus indirectly quantify the value of channels via the value of states. From a conceptual point of view, it is however more natural to quantify the value of operations directly~\cite{Theurer2019}. Since quantum states can be seen as a special case of quantum operations with no input and a fixed output, quantifying the value of operations is a unifying concept that can also be used to quantify properties of operations that cannot be reduced to static resources~\cite{Theurer2019}. These observations have recently led to the development of dynamical resource theories~\cite{Zhuang2018, Theurer2019, Bauml2019, fang2019quantum, Liu2020, Wang2019a, Liu2019, Gour2019a, Takagi2019, Gour2020b,Gour2021, Fang2022} and in particular dynamical resource theories of entanglement~\cite{Gour2020b,Gour2021,Bauml2019,Li2021, Zhou2022}. The distillation and dilution of the entanglement of channels is an important primitive in such theories for the same reasons as in the static case and has been studied under relaxations of LOCC such as complete PPT-preservation~\cite{Rains1999, Rains2001, Audenaert2003} in Refs.~\cite{Das2020, Bauml2019, Wang2023} and under separability preservation in Ref.~\cite{Kim2021,Kim2021b}. See also Ref.~\cite{Fang2022,Regula2021} for fundamental limitations on the distillation of channel resources and Ref.~\cite{Takagi2022} for yield-cost relations in general resource theories.  In Sec.~\ref{subsec:channelCost}, we provide bounds on the $\eps$-single-shot entanglement cost of quantum channels under LOCC which coincide in the zero-error limit. 

An important subclass of LOCC is one-way LOCC in which classical communication is only allowed in one direction. In Sec.~\ref{sec:OneWay}, we introduce the coherent information of entanglement, which is monotonic under one-way LOCC, and use it to bound the one-way $\eps$-single shot distillable entanglement of both states and channels. We conclude with a discussion and outlook.

\section{Notation and preliminaries}\label{sec:Notation}

\begin{figure*}[t]\centering    \includegraphics[width=0.7\textwidth]{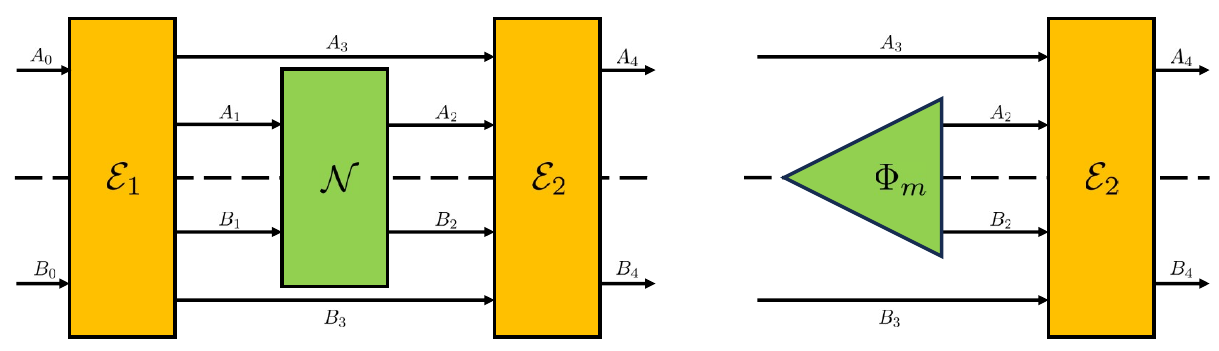}
	\caption{Left: LOCC superchannel converting a channel $\mN$ 	 to a channel $\mM=\mE_2\circ\mN\circ\mE_1$, where both $\mE_1$ and $\mE_2$ are in $\locc$. Right: If $\mN$ is replaced by a state, the superchannel simplifies. Solid lines represent quantum systems and the dashed line the spacial separation between Alice and Bob.}
	\label{fig:LOCCsuperChannel}
\end{figure*} 
Unless stated otherwise, proofs are provided in App.~\ref{sec:proofs}. In this paper, we only consider finite-dimensional Hilbert spaces, which we denote with capital Latin letters such as $C$. The dimension of a Hilbert space $C$ is denoted by $|C|$, and the set of density matrices acting on it by $\md(C)$, with $\pure(C)$ denoting the subset of pure states. Density matrices are represented by small Greek letters such as $\rho^C$, where the superscript indicates that $\rho$ acts on $C$. For, e.g., a state $\rho^{AB}\in\md(AB)$ we will also use the convention that $\rho^A=\tr_B\left[\rho^{AB}\right]$ denotes the marginal on system $A$. Whenever we consider a copy of a Hilbert space $A$, we will denote it as $\tA$ and we reserve $X$ to denote a classical system or register.

In the following, we will mainly be concerned with two spatially separated parties, Alice and Bob. To make clear which system is under the control of whom, we will use $A$ and $A'$ to denote Alice's systems and $B$ and $B'$ for Bob's. Quantum channels will be denoted by calligraphic large Latin letters such as $\mN$ (with the exemption of the identity channel, which we denote by $\id$) and the set of all quantum channels from $A$ to $B$ by CPTP$(A\rightarrow B)$, which stands for completely positive and trace-preserving. Completely positive linear maps will be denoted by CP and a collection of CP maps $\{\mN_x\}_{x=1}^n$ for which $\sum_{x = 1}^n \mN_x$ is a quantum channel will be called an instrument. Since we can always store the classical outcome $x$ of any instrument in a classical system $X$, we will interchangeably also write $\sum_x \mN_x \otimes \ketbra{x}{x}^X$ for the instrument.

The set of quantum channels that is implementable with local operations and classical communication will be denoted by $\locc$. Since with $\locc$, Alice and Bob can always attach and remove local auxiliary systems, for bipartite systems $AB$, we will assume in the following w.l.o.g. that $|A| = |B|$. Whenever there exists an $\mN\in \locc$ such that $\mN(\rho)=\sigma$, we write $\rho \xrightarrow\locc \sigma$. The set of channels that can be implemented with local operations and one-way classical communication \textit{from Alice to Bob} will be denoted by $\locc_1$. Moreover, we will denote super-channels~\cite{Chiribella2008}, i.e., linear maps between quantum channels that can be implemented by concatenating the channel they act on with two other channels implementing a pre- and post-processing, with capital Greek letters such as $\Theta$. We will be particularly concerned with LOCC super-channels, i.e., super-channels that can be implemented by a pre- and post-processing that are both in LOCC, see Fig.~\ref{fig:LOCCsuperChannel}.

We utilize bold small Latin letters such as $\p$ for probability vectors, with $p_x$ the $x$-th component of $\p$, and let $\prob(d)$ be the set of probability vectors of length $d$. The subset that contains the $d$-dimensional probability vectors with non-increasing entries will be denoted as $\prob^\da(d)$ and for $\p\in \prob(d)$, $\p^\da\in\prob^\da(d)$ represents the vector that we obtain by reordering the elements of $\p$. For $k\in [d]$, where $[d]$ is a short-hand notation for $\{1,...,d\}$, the $k$-th Ky-Fan norm of $\p\in \prob(d)$ is defined as 
\begin{align}
    \norm{\p}_{(k)}=\sum_{x=1}^k p^\da_x.
\end{align}
This definition can be extended to positive semidefinite operators $A$ by denoting with $\left\|A\right\|_{(k)}$ the sum of the $k$ largest eigenvalues of $A$ (i.e., the Ky-Fan norm of its eigenvalues). For $\p,\q \in\prob(d)$, we write $\p \succ \q$ if $\p$ majorizes $\q$~\cite{Marshall2011}, i.e., if
\begin{align}
	\norm{\p}_{(k)} \ge \norm{\q}_{(k)} \forall k\in [d].
\end{align}
If $\p,\q$ have different dimensions, we extend the definition by padding the shorter vector with zeros.
By fixing an arbitrary orthonormal basis $\Ket{x}^C$ for every Hilbert space $C$, any $\psi \in \pure(AB)$ is then $\locc$-equivalent to its standard form $\sum_x \sqrt{p_x} \Ket{xx}^{AB}$, where $\p \in \prob^\downarrow(|A|)$ contains the Schmidt coefficients of $\psi$, which, w.l.o.g., we always assume to be ordered non-increasingly. The number of non-zero Schmidt coefficients of $\psi$, i.e., its Schmidt rank, is denoted by $\sr(\psi)$ and $\Phi_m = \frac{1}{\sqrt{m}}\sum_{x = 1}^m \Ket{xx}^{AB}$ represents the the maximally entangled state of dimension $m$, where w.l.o.g., we will always assume that $m=|A|=|B|$.

We use $H(\rho)\eqdef-\tr(\rho\log \rho)$ to denote the von-Neumann entropy, and for $\rho\in\md(ABE)$, let $H(A|B)_\rho\eqdef H(\rho^{AB})-H(\rho^B)$ be the conditional von-Neumann entropy. For $\rho,\sigma \in\md(C)$ and $\norm{\cdot}_1$ the trace norm, the trace distance between $\rho$ and $\sigma$ is defined as $\frac12\norm{\rho-\sigma}_1$ and the relative entropy as $D(\rho\|\sigma)\eqdef\tr(\rho \log \rho)-\tr(\rho \log  \sigma$). The (trace) conversion distance under $\locc$ is commonly defined as  
\begin{align}\label{eq:convdist}
	&T\left(\rho^{AB} \xrightarrow{\locc} \sigma^{A'B'}\right) \\
    &\eqdef \min_{\tau\in\md(A'B')}\!\left\{ \frac12 \norm{\tau^{A'B'}-\sigma^{A'B'}}_1  :\rho^{AB}\xrightarrow{\locc} \tau^{A'B'}\right\}. \nonumber
\end{align}
On pure states $\psi$, $\phi \in \pure(AB)$, one can analogously define the star conversion distance~\cite{Zanoni2023} (see also Ref.~\cite{Kumagai2017}) as
\begin{align} \label{eq:starconvdist}
	&T_\star\left(\psi^{AB} \xrightarrow{\locc} \phi^{AB}\right) \nonumber\\
	&\eqdef \min_{\r \in \prob(|A|)}\left\{\frac{1}{2}\norm{\r-\q}_1\; :\; \r \succ \p \right\},
\end{align}
where $\p$, $\q \in \prob^\downarrow(|A|)$ are the Schmidt coefficients of $\psi$ and $\phi$, respectively, and $\norm{\p-\q}_1=\sum_x |p_x-q_x|$. 
This definition can easily be extended to pure bipartite states that do not belong to the same Hilbert space~\cite{Zanoni2023}: Using LOCC, one can always add and remove separable auxiliary states such that the Hilbert spaces match. For $\psi \in \pure(AB)$, $\phi \in \pure(A'B')$, and $d = |A| = |B|$, $d' = |A'| = |B'|$, we thus define
\begin{align}\label{eq:TstarGen}
    &T_\star\left(\psi^{AB} \xrightarrow{\locc} \phi^{A'B'}\right)  \\
    &\eqdef T_\star\left(\psi^{AB} \otimes \ketbra{11}{11}^{A'B'} \xrightarrow{\locc} \ketbra{11}{11}^{AB}\otimes\phi^{A'B'}  \right). \nonumber
\end{align}

Importantly, the conversion distances discussed above are topologically equivalent in the sense that~\cite{Zanoni2023}
\begin{align}\label{eq:equiv}
	&\frac{1}{2}T_\star^2 \left(\psi^{AB} \xrightarrow{\locc} \phi^{A'B'}\right) \leq T\left(\psi^{AB} \xrightarrow{\locc} \phi^{A'B'}\right) \nonumber \\
    &\leq \sqrt{2T_\star\left(\psi^{AB} \xrightarrow{\locc} \phi^{A'B'}\right)}
\end{align}
and according to Ref.~\cite[Thm.~3]{Zanoni2023}, $T_\star$ exhibits the following closed-form expression that is derived via the concept of approximate majorization~\cite{Horodecki2018}:

\begin{theorem}\label{thm:StarDist}
Let $\psi\in \pure(AB), \phi\in \pure(A'B')$, and $\p\in\prob^\downarrow(|A|), \q\in\prob^\downarrow(|A'|)$ be their corresponding Schmidt coefficients. Then,
\ba\label{eq:dlocc}
T_{\star}\!\left(\!\psi^{AB}\!\xrightarrow{\locc}\!\phi^{A'B'}\right)=&\max_{k\in[\sr(\psi^{AB})]}\big\{\|\p\|_{(k)}\!-\!\|\q\|_{(k)}\big\}.\nonumber
\ea
\end{theorem}

The purified distance between two states $\rho$, $\sigma \in \md\left(AB\right)$ is defined as~\cite{Tomamichel2010,Gour2020}
\begin{equation}\label{eq:purifieddistance}
    P\left(\rho, \sigma\right) \eqdef \sqrt{1 - F^2\left(\rho, \sigma\right)},
\end{equation}
where $F\left(\rho, \sigma\right) \eqdef \norm{\sqrt{\rho}\sqrt{\sigma}}_1$ is the fidelity. 
Analogously to what was done above with the trace distance, one can define two conversion distances based on the purified distance~\cite{Zanoni2023}: For $\rho\in \md(AB), \sigma\in \md(A'B')$, the purified conversion distance is defined as
\begin{equation}\label{eq:PurConvDist}
    P\left(\rho \xrightarrow{\locc} \sigma\right) \eqdef \min_{\tau\in\md(A'B')}\left\{ P\left(\tau, \sigma\right) : \rho\xrightarrow{\locc}\tau \right\}
\end{equation}
and for $\psi\in \pure(AB), \phi \in \pure(A'B')$, the purified star conversion distance is defined as~\cite{Zanoni2023}
\begin{equation}
    P_\star\left(\psi \xrightarrow{\locc} \phi\right) \eqdef \min_{\r\in\prob(|A|)}\Big\{P(\r,\q)\;:\;\r\succ\p\Big\},
\end{equation}
where $\p,\q \in \prob^\downarrow(|A|)$ are the Schmidt coefficients of $\psi$ and $\phi$, respectively,  $P\left(\r, \q\right) \eqdef \sqrt{1 - F^2(\r, \q)}$ is the classical version of the purified distance, and $F(\r, \q)\eqdef \sum_x \sqrt{r_x q_x}$. The extension to states of different systems is done as in Eq.~\eqref{eq:TstarGen}.
 According to Ref.~\cite{Zanoni2023}, the following Theorem holds (compare also Ref.~\cite[Lem.~2]{Vidal2000} and Ref.~\cite[Eq.~(131)]{Kumagai2017}).
\begin{theorem}\label{thm:EquivPurifDist}
	Let $\psi\in\pure(AB),\phi\in\pure(A'B')$. Then,
	\be\label{eq12p70}
	P\left(\psi^{AB}\xrightarrow{\locc}\phi^{A'B'}\right)=P_{\star}\left(\psi^{AB} \xrightarrow{\locc} \phi^{A'B'}\right).
	\ee
\end{theorem}

\section{Single-shot entanglement manipulation with LOCC}
\subsection{Pure state entanglement manipulation and the star conversion distance}\label{subsec:StarSingleShot}
The star conversion distance can be used to define and calculate an $\eps$-single-shot distillable entanglement, which we will do in the following.
For any $\eps\in[0,1]$ and $\psi\in\pure(AB)$, the $\eps$-single-shot distillable entanglement is defined as
\begin{align}\label{eq:PureSingleDistDef}
	&\distill^\eps\left(\psi^{AB}\right)\\
	&\eqdef\max_{m\in\mbb{N}}\Big\{\log m\;:\;T_\star\left(\psi^{AB}\xrightarrow{\locc} \Phi_m\right)\leq\eps\Big\}. \nonumber
\end{align}
From the closed formula for $T_\star$, we get the following result.

\begin{theorem}\label{thm:PureSingleDist}
Let $\eps\in[0,1)$, $\psi\in\pure(AB)$, $d\eqdef\sr(\psi^{AB})$, and $\p\in\prob^\downarrow(|A|)$ be the Schmidt coefficients of $\psi^{AB}$. The $\eps$-single-shot distillable entanglement of $\psi^{AB}$ is then given by
\be\label{1236}
\distill^\eps\left(\psi^{AB}\right)=\min_{k\in\{\ell,\ldots,d\}} \log\left\lfloor\frac{k}{\|\p\|_{(k)}-\eps}\right\rfloor\;,
\ee
where $\ell\in[d]$ is the integer satisfying 
$\|\p\|_{(\ell-1)}\leq\eps<\|\p\|_{(\ell)}$. 
\end{theorem}

For any $\eps\in[0,1]$ and $\psi\in\pure(AB)$, the $\eps$-single-shot entanglement cost is defined as
\begin{align}\label{eq:DefCostPure}
&\cost^\eps\left(\psi^{AB}\right)\\
&\eqdef\min_{m\in\mbb{N}}\Big\{\log m\;:\;T_\star\left(\Phi_m \xrightarrow{\locc} \psi^{AB}\right)\leq\eps\Big\}. \nonumber
\end{align}
From the closed formula for $T_\star$, we get the following result.
\begin{theorem}\label{thm:PureSingleCost}
Let $\eps\in[0,1)$, $\psi\in\pure(AB)$, $d\eqdef\sr(\psi^{AB})$, and $\p\in\prob^\downarrow(|A|)$ be the Schmidt coefficients of $\psi^{AB}$.
The $\eps$-single-shot entanglement cost of $\psi^{AB}$ is then given by
\be\label{1242}
\cost^\eps\left(\psi^{AB}\right)=\log m,
\ee
where $m\in[d]$ is the integer satisfying 
$
\|\p\|_{(m-1)}<1-\eps\leq\|\p\|_{(m)}
$.
\end{theorem} 

As explained in detail in App.~\ref{app:computability}, the above Theorems allow us to efficiently compute both the $\eps$-single-shot distillable entanglement  and the $\eps$-single-shot entanglement cost for multiple copies of a given pure state.

\begin{theorem}\label{thm:effComp}
    Let $n\in\mbb{N}$, $\eps\in[0,1)$, and $\psi\in\pure(AB)$. This implies that both $\distill^\eps\left(\psi^{\otimes n}\right)$ and $\cost^\eps\left(\psi^{\otimes n}\right)$ can be computed efficiently.   
\end{theorem}

The main idea behind the proof is to use that whilst $\p^{\otimes n}$ has a number of entries that is exponential in $n$, it only has a polynomial number of distinct entries, which allows to determine $\|\p^{\otimes n}\|_{(k)}$ and therefore $\distill^\eps\left(\psi^{\otimes n}\right)$ and $\cost^\eps\left(\psi^{\otimes n}\right)$. In App.~\ref{app:computability}, we provide explicit algorithms to determine these quantities. 

Moreover, Thms.~\ref{thm:PureSingleDist} and \ref{thm:PureSingleCost} also allow us to obtain the second-order asymptotics of the entanglement cost and the distillable entanglement: Let the Gaussian cumulative distribution function with mean value $\mu$ and variance $\nu$ be denoted by
\begin{align}
	\Phi_{\mu,\nu}(x) \eqdef \frac{1}{\sqrt{2\pi \nu}} \int_{-\infty}^{x} e^{-\frac{(t-\mu)^2}{2\nu}} \diff t
\end{align}
and the entropy variance $V(\p)$ of a probability distribution $\p$ by
\begin{align}
	V(\p)\eqdef \sum_i p_i (-\log p_i - H(\p))^2,
\end{align}
where $H(\p)\eqdef-\sum_i p_i \log p_i $ is the Shannon entropy. As a shorthand notation, we will also use $\Phi$ to denote $\Phi_{0,1}$. 
We begin by stating the following Lemma, which is a consequence of Ref.~\cite[Lem.~12]{Kumagai2017} (see also Ref.~\cite[Lem.~16]{chubb2018beyond}).
\begin{lemma}\label{lem: second-order 1}
	For any distribution $\p$ such that $V(\p) > 0$, any natural number $n$, and $\eps \in [0,1)$, let
	\begin{align}
		f_{n,\eps}(\p):=\min \left\{k: \|\p^{\otimes n}\|_{(k)} > \eps\right\}, \nonumber \\
		f'_{n,\eps}(\p):=\min \left\{k: \|\p^{\otimes n}\|_{(k)} \geq \eps\right\}.
	\end{align}
	Then we have that
	\begin{align}
		&\lim_{n\to \infty} \frac{\log {f'_{n,\eps}(\p)} - n H(\p)}{\sqrt{n V(\p)}}\nonumber\\
		&=\lim_{n\to \infty} \frac{\log {f_{n,\eps}(\p)} - n H(\p)}{\sqrt{n V(\p)}} \nonumber\\
		&= \Phi^{-1}(\eps).
	\end{align}    
\end{lemma}
By combining this Lemma with Thm.~\ref{thm:PureSingleCost}, one immediately obtains the following Proposition. 

\begin{proposition}\label{prop:SecondOrderPureCost}
	For any pure state $\psi\in\pure(AB)$ with Schmidt vector $\p$, $V(\p) > 0$, and $\eps \in [0,1)$, it holds that
	\begin{align}
		\cost^\eps(\psi^{\otimes n}) = n H(\p) - \Phi^{-1}(\eps) \sqrt{n V(\p)} + o(\sqrt{n}).
	\end{align}   
\end{proposition}

With slightly more work, it is also possible to obtain the second-order asymptotics of the distillable entanglement.
\begin{proposition}\label{prop:SecondOrderPureDist}
	For any pure state $\psi\in\pure(AB)$ with Schmidt vector $\p$, $V(\p) > 0$, and $\eps \in [0,1)$, it holds that
	\begin{align}
		\distill^\eps(\psi^{\otimes n}) = n H(\p) + \Phi^{-1}(\eps) \sqrt{n V(\p)} + o(\sqrt{n}).
	\end{align}
\end{proposition}

The above two Propositions are the analogues of Ref.~\cite[Eqs.~(136) and (137)]{Kumagai2017}, where the error $\eps$ was measured with $P$ instead of $T_\star$. Since the considered distance measures are topologically equivalent, one immediately recovers the asymptotic entanglement cost and distillable entanglement of pure entangled states~\cite{Horodecki2009}.

\subsection{Single-shot entanglement cost of mixed states}
In this section, we compute an $\eps$-single-shot entanglement cost of bipartite mixed states. To this end, we review and study a convenient conversion distance first: On bipartite states $\psi\in \pure(AB)$, and for any $k\in [|A|]$, let
\be\label{eq:vidal}
E_{(k)}(\psi^{AB})\eqdef 1-\|\p\|_{(k)}\;,
\ee 
where $\p$ contains the Schmidt coefficients of $\psi^{AB}$. Applying a convex-roof extension, for mixed states $\rho\in\md(AB)$, define
\begin{align}\label{kyfan}
E_{(k)}\left(\rho^{AB}\right)\eqdef&\inf\sum_xp_xE_{(k)}\left(\psi^{AB}_x\right)\nonumber \\
=& \inf \left(1-\sum_x p_x\left\|\tr_A\left[\psi^{AB}_x\right]\right\|_{(m)}\right),
\end{align}
where the infimums are over all pure-state decompositions $\rho^{AB}=\sum_xp_x\psi^{AB}_x$~\cite{Terhal2000,Hayashi2006, Buscemi2011, Yue2019}. In the proof of Thm.~\ref{lem1291}, we will see that these infimums are attained and that it is sufficient to consider decompositions into at most $|AB|^2$ pure states.

According to Refs.~\cite{Vidal1999, Nielsen1999}, for $\psi,\phi \in \pure(AB)$, $\psi^{AB}\xrightarrow{\locc}\phi^{AB}$ iff
\be
 E_{(k)}(\psi^{AB})\geq E_{(k)}(\phi^{AB})\;\;\forall\;k\in[|A|]\;.
\ee 
As we will see now, this operational interpretation can be partially extended to mixed states.
\begin{theorem}\label{lem1291} 
Let $\rho\in\md(AB)$ and $m\in\mbb{N}$. Then,
\be
P^2\left(\Phi_m\xrightarrow{\locc} \rho^{AB} \right)=E_{(m)}\left(\rho^{AB}\right),
\ee
where $E_{(m)}\left(\rho^{AB}\right)$ is defined in Eq.~\eqref{kyfan} with $k=m$.
\end{theorem}
Note that a similar characterization of $$P^2\left(\Phi_m\xrightarrow{\locc} \rho^{AB} \right)$$ previously appeared in Ref.~\cite[Lem.~1]{Buscemi2011} where Refs.~\cite{Bowen2007,Hayashi2006} were credited. However, it was not expressed in terms of $E_{(m)}$. In the proof of the Theorem, we will show that this is possible due to certain properties of the Schmidt number $\sr(\rho^{AB})$ defined as~\cite{Terhal2000}
\begin{align}\label{eq:SchmidtNumberTerhal}
    \sr(\rho^{AB})=\inf_{\{p_x, \psi_x\}} \max_{x\in[k]:p_x\ne0} \sr(\psi^{AB}_x),
\end{align}
where the infimum is over all pure-state decompositions of $\rho^{AB}=\sum_{x\in[k]}p_x\psi_x^{AB}$. 

Whilst this result is interesting in itself, it allows us to characterize the $\eps$-single-shot entanglement cost: 
According to Ref.~\cite[Def.~1]{Buscemi2011}, for any $\eps\in[0,1]$, the $\eps$-single-shot entanglement cost is defined as 
\begin{align}\label{econe}
&\cost^{\eps}\left(\rho^{AB}\right)\\
&\eqdef\min\Big\{\log m\;:\;P^2\left(\Phi_m\xrightarrow{\locc} \rho^{AB}\right)\leq \eps\Big\}.\nonumber
\end{align}
Note that at first glance, for pure states, this definition conflicts with Eq.~\eqref{eq:DefCostPure}. However, as we will show later, the two definitions coincide. In a bit of abuse of notation, we will thus not differentiate between the two, e.g., by adding a subscript to denote the distance according to which we determine the allowed error in the transformation. 

To provide formulas for the $\eps$-single-shot entanglement cost, it is convenient to utilize the concept of \textit{classical extensions} of a bipartite density matrix $\rho\in\md(AB)$: With any decomposition $\{p_x,\rho^{AB}_x\}_{x\in[k]}$ of $\rho^{AB}$, i.e., $\rho^{AB}=\sum_{x\in[k]}p_x\rho_x^{AB}$, one associates a classical-quantum-state (see Ref.~\cite[Eq.~(3)]{Buscemi2011})
\be
\rho^{XAB}\eqdef\sum_{x\in[k]} p_x\ketbra{x}{x}^X\otimes\rho_{x}^{AB},
\ee
where $X$ is a classical system of dimension $k$. Such an extension will be called a \textit{regular extension} of $\rho^{AB}$ if all $\rho_x^{AB}$ are \textit{pure states}.

We further denote by $H_{\max}$ the conditional max-entropy, i.e., 
\begin{align}\label{eq:condMaxEnt}
    H_{\max}(A|B)_{\rho}\eqdef \max_{\tau\in\md(B)} \log\tr\left[\Pi_\rho^{AB} \left(I^A\otimes \tau^B \right) \right],
\end{align}
where $\Pi_\rho^{AB}$ is the projection onto the support of $\rho^{AB}=\tr_C\left[\rho^{ABC}\right]$, and by $H_{\max}^{\eps}$ its smoothed version defined as
\begin{align}\label{eq:smoothCondMaxEnt}
    H_{\max}^{\eps}(A|B)_{\rho}=\min_{\omega \in \mb_{\eps}(\rho^{AB})} H_{\max}(A|B)_{\omega},
\end{align} 
where 
\begin{align}\label{eq:epsBall}
    \mb_\eps\left(\rho^A\right)=\left\{\tau\in\md(A): \frac{1}{2}\norm{\tau^A-\rho^A}_1\le\eps\right\}.
\end{align} 
For a classical extension 
\begin{align}
    \rho^{XAB}=\sum_{x\in [k]} p_x\ketbra{x}{x}^X\otimes\rho_x^{AB}
\end{align} 
of a bipartite state $\rho\in\md(AB)$, this implies that
\begin{align}\label{1260}
H_{\max}(A|X)_{\rho} &=\max_{x\in [k]:p_x\ne0}\log\tr\left[\Pi_{\rho_x}^A\right],
\end{align}
where $\rho_x^A\eqdef\tr_B\left[\rho_x^{AB}\right]$, and
\be\label{omex}
H_{\max}^{\eps}(A|X)_\rho=\min_{\omega\in\mb_{\eps}\left(\rho^{XA}\right)}\max_{x\in [k]:q_x\ne0}\log\tr\left[\Pi_{\omega_x}^A\right]\;,
\ee
with
\begin{align}
    \omega^{XA}=\sum_{x\in [k]} q_x\ketbra{x}{x}^X\otimes\omega_x^A.
\end{align}
It was pointed out in Ref.~\cite{Buscemi2011} (see also the proof of Thm.~\ref{lem1291}) that 
\begin{align}\label{eq:costSchmidtNumber}
    &\cost^{\eps}\left(\rho^{AB}\right)\\
    &=\min_{\omega\in\md(AB)}\Big\{\log \sr\left(\omega^{AB}\right):\;P^2\left(\omega^{AB},\rho^{AB}\right)\leq \eps\Big\}\nonumber \\
    &=\inf_{\omega^{XAB}}\Big\{H_{\max}(A|X)_{\omega}:\;P^2\left(\omega^{AB},\rho^{AB}\right)\leq \eps\Big\},\nonumber
\end{align}
where the infimum is over all regular extensions of $\rho^{AB}$.

This means that the $\eps$-single-shot entanglement cost can be expressed in terms of a smoothed version of the Schmidt number/conditional max-entropy. As the main result of Ref.~\cite{Buscemi2011}, it was further shown in their Thm.~1 that the $\eps$-single-shot entanglement cost can be lower and upper bounded by a slight modification of $H_{\max}^{\eps}$  (take Eq.~\eqref{eq:smoothCondMaxEnt}, but relax the requirement in Eq.~\eqref{eq:epsBall} that $\tau$ is a density operator to the requirement that it is positive semidefinite). In the following, we will show that (with our slight change in the smoothing), the $\eps$-single-shot entanglement cost can be expressed \textit{exactly} in terms of the conditional max-entropy. To this end, we introduce some notation and a Lemma first.

For every classical-quantum-state 
\begin{align}\label{eq:cqState}
	\rho^{XA}=\sum_{x\in [k]} p_x\ketbra{x}{x}^X\otimes\rho_x^{A},
\end{align} 
define 
\begin{align}\label{eq:mPrune}
	\rho^{(m)}=\sum_{x\in [k]} p_x\ketbra{x}{x}^X\otimes\rho_x^{(m)},
\end{align} 
with $\sigma^{(m)}$ an $m$-pruned version of $\sigma^A$ defined as 
\begin{align}
	\sigma^{(m)}\eqdef\frac{\Pi_m^A\sigma^A\Pi_m^A}{\tr\left[\sigma^A\Pi_m^A\right]},
\end{align}
where $\Pi_m^A$ is a projection to a subspace spanned by $m$ orthogonal eigenvectors corresponding to the $m$ largest eigenvalues of $\sigma^A$~\footnote{Note that due to the possibility of degenerate eigenvalues, $\Pi_m$ and thus $\sigma^{(m)}$ are not necessarily unique. However, for our purpose, this is irrelevant.}.

\begin{lemma}\label{lem:HMaxCQ}
	Let $\rho\in\md(XA)$ be a classical-quantum-state as in Eq.~\eqref{eq:cqState}, and for any $m\in[|A|]$, let $\rho^{(m)}$ be as defined in Eq.~\eqref{eq:mPrune}. Then, for any $\eps\in[0,1]$, 
	\begin{align}
		H&_{\max}^{\eps}(A|X)_\rho \\=&\min_{m\in[|A|]}\left\{\log m: \frac{1}{2}\norm{\rho^{(m)}-\rho^{XA}}_1\le\eps\right\}\nonumber\\
		=&\min_{m\in[|A|]}\left\{\log m: \sum_{x\in[k]}p_x \left\|\rho^A_x\right\|_{(m)}\ge 1-\eps\right\}.\nonumber
	\end{align}
\end{lemma}
This Lemma shows that $H_{\max}^{\eps}(A|X)_\rho$ can be directly evaluated by calculating $|A|$ trace distances, whilst a priori, it is defined as an optimization problem over an $\eps$-ball. Whilst this might be of independent interest, it allows us to prove the promised Theorem.   

\begin{theorem}\label{cmaint}
For $\rho\in\md(AB)$, the $\eps$-single-shot entanglement cost is given by
\be\label{eq:mainThm}
\cost^\eps\left(\rho^{AB}\right)= \inf_{\rho^{XAB}} H_{\max}^{\eps}(A|X)_\rho\;,
\ee
where the infimum is over all classical systems $X$ and all classical extensions $\rho^{XAB}$ of $\rho^{AB}$. Moreover, the infimum is attained for a classical extension with $|X|=|AB|^2$ and can also be taken over all regular extensions of $\rho^{AB}$. 
\end{theorem}
It is worth noticing that both in the above Theorem as well as in Ref.~\cite[Thm.~1]{Buscemi2011}, the $\eps$ in $H_{\max}^{\eps}$ stands for smoothing with respect to the trace distance whilst the $\eps$ in $\cost^\eps$ corresponds to an error in conversion measured by the square of the purified distance. This is in contrast to Eq.~\eqref{eq:costSchmidtNumber} where, e.g., also the Schmidt number is smoothed with respect to the square of the purified distance. 

For a pure state $\rho^{AB}=\psi^{AB}$, the optimization over classical extensions $\rho^{XAB}$ in Thm.~\ref{cmaint} is trivial, since all of them are of the form $\rho^{XAB}=\sigma^X\otimes\psi^{AB}$. This allows us to simplify the expression for $\cost^\eps\left(\psi^{AB}\right)$:  Let $$H_{\max}(\rho) = \log \tr[\Pi_\rho]$$ be the max-entropy~\cite{Datta2009} and $$H_{\max}^{\eps}(\rho^A)\eqdef \min_{\omega \in \mb_{\eps}(\rho^{A})} H_{\max}(\omega^A)$$ its $\eps$-smoothed version. Via Eqs.~\eqref{eq:smoothCondMaxEnt} and~\eqref{1260}, we thus obtain the following Corollary. 
\begin{corollary}\label{cor:pureCostHmax}
	Let $\psi\in\pure(AB)$ and $\rho^A=\tr_B\left[\psi^{AB}\right]$. It then holds that
	\begin{align}
		\cost^\eps\left(\psi^{AB}\right)= H_{\max}^{\eps}(\rho^A).
	\end{align}
\end{corollary}

This equips the smoothed max-entropy with an operational interpretation in terms of the $\eps$-single-shot entanglement cost of pure states. Moreover, by applying Lem.~\ref{lem:HMaxCQ} (with a trivial classical system $X$), we find that 
\begin{align}\label{eq:smoothHMax}
	H_{\max}^{\eps}(\rho^A)=\min_{m\in[|A|]}\left\{\log m:  \left\|\rho^A\right\|_{(m)}\ge 1-\eps\right\}.
\end{align} 
Now notice that $\left\|\rho^A\right\|_{(m)}=\norm{\p}_{(m)}$,  where $\p\in\prob^\downarrow(|A|)$ are the Schmidt coefficients of $\psi^{AB}$.  As claimed earlier, this also shows that the two definitions of $\cost^\eps$ provided in Eqs.~\eqref{eq:DefCostPure} and~\eqref{econe} indeed coincide on pure states (see Thm.~\ref{thm:PureSingleCost}), and Prop.~\ref{prop:SecondOrderPureCost} thus applies to both. This is interesting, since in Eq.~\eqref{eq:DefCostPure}, $\eps$ was bounding the star conversion distance $T_\star$, whilst, in Eq.~\eqref{econe}, it was bounding the square of the purified distance $P^2$. Moreover, we note that the previous discussion also implies that the quantity $f'_{n,\eps}(\p)$ in Lem.~\ref{lem: second-order 1} has an operational interpretation in terms of the smoothed max-entropy.

At this point, we want to mention that Ref.~\cite{Regula2019} investigated a variant of the $\eps$-single-shot distillable entanglement too: The authors of Ref.~\cite{Regula2019} defined the fidelity of distillation as
\begin{align}\label{eq:fidDist}
	F(\rho,m):= \sup_{\Lambda\in\locc} \tr(\Lambda(\rho) \Phi_m)
\end{align}
and their variant of the $\eps$-single-shot distillable entanglement as
\begin{align}
	E_{D}^{(1),\eps}(\rho):= \log \max \{m\geq 2: F(\rho,m) \geq 1-\eps\}.
\end{align}
Comparing this definition to our definition in Eq.~\eqref{econe}, we see that they differ in the way in which the allowed error $\eps$ is introduced: Whilst we followed the convention of Ref.~\cite{Buscemi2011} and demanded that $1-F^2\le \eps$, in Ref.~\cite{Regula2019}, it was demanded that $1-F\le \eps$. Importantly, for pure states, there exists a closed-form formula for $E_{D}^{(1),\eps}(\psi)$ too \cite[Thm.~15]{Regula2019} which is given in terms of a distillation norm~\cite{Regula2018}. This definition does therefore not coincide with the ones we used. Nevertheless, in App.~\ref{app:Regula}, we discuss in detail how this allows us to efficiently compute $E_{D}^{(1),\eps}(\psi^{\otimes n})$ in a manner that is very similar to how we can compute $\distill^\eps\left(\psi^{\otimes n}\right)$. To conclude the discussion concerning the entanglement manipulation of quantum states, we want to mention that from Refs.~\cite{Vidal2000} and \cite{Kumagai2017}, also analogs of our Thm.~\ref{thm:PureSingleCost}/Cor.~\ref{cor:pureCostHmax} and Thm.~\ref{thm:PureSingleDist} can be extracted, again with definitions of the error that differ from ours. Our choices of the conversion distance lead to particularly compact formulas.

\subsection{Single-shot entanglement cost of channels}\label{subsec:channelCost}
\begin{figure}[t]\centering    \includegraphics[width=0.9\linewidth]{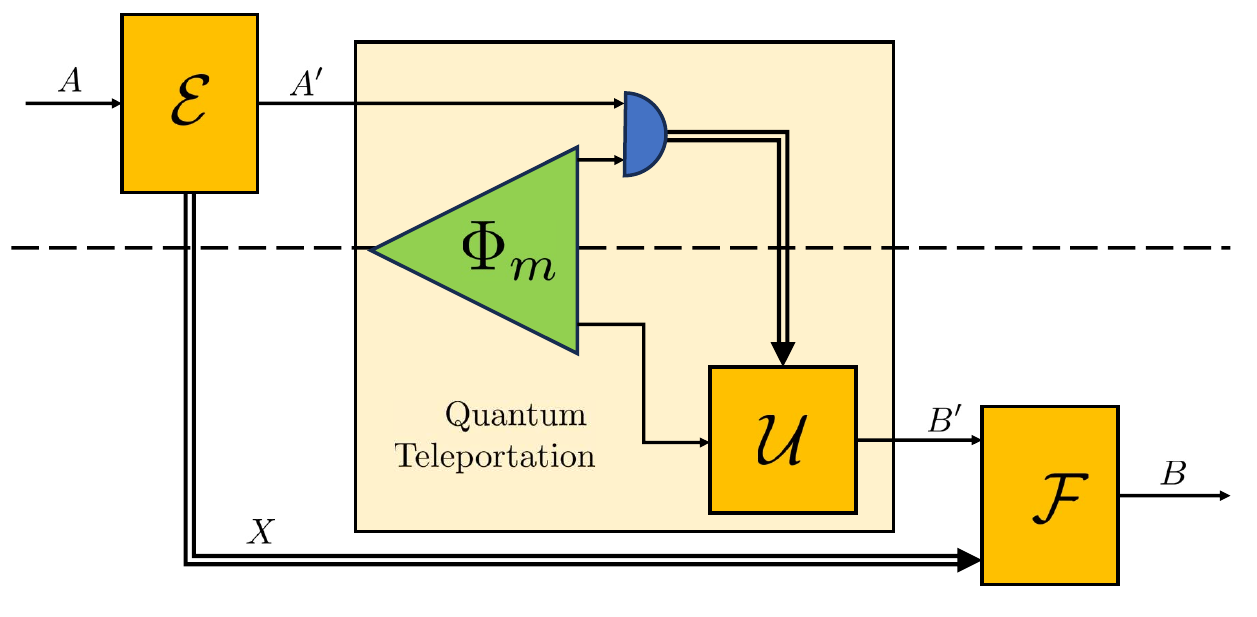}
	\caption{Optimal simulation of a quantum channel $\mN^{A\to B}$ if given access to $\Phi_m$ and LOCC using quantum teleportation. Solid lines represent quantum systems, double lines classical systems, and the dashed line the spacial separation between Alice and Bob.}
	\label{simcha}
\end{figure} 
\begin{figure*}[t]\centering    \includegraphics[width=0.9\textwidth]{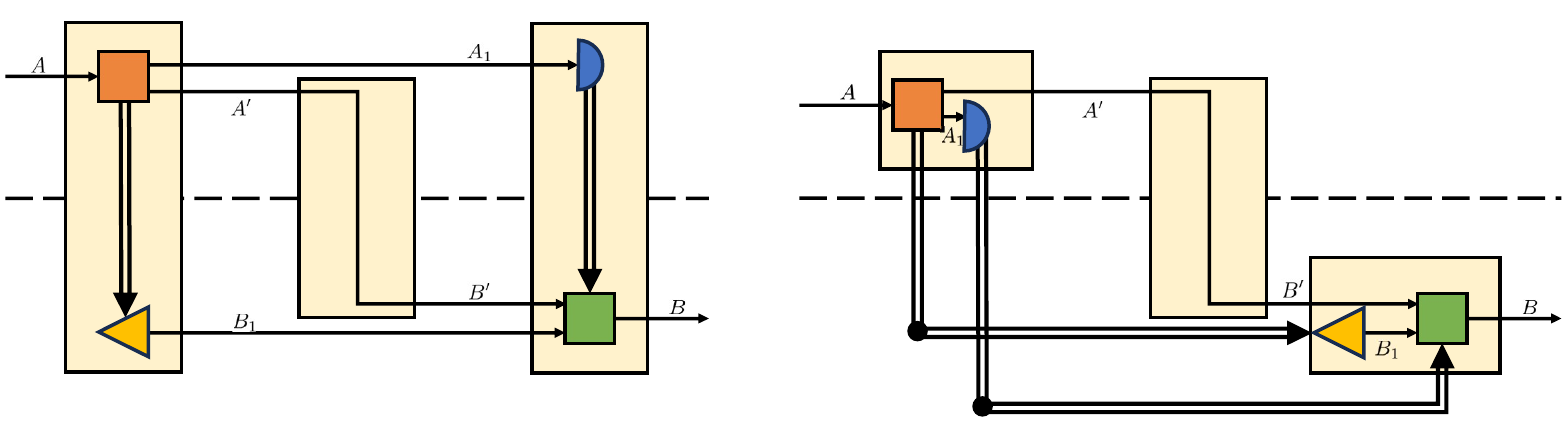}
	\caption{Left: LOCC superchannel converting $\id^{A'\to B'}$ (center) to a channel $\mN\in \cptp(A\to B)$. Right: Any such superchannel can be realized by the following protocol: Alice applies a quantum instrument, sends the quantum output through the identity channel, and Bob applies a channel on its output that is conditioned on the classical outcome of Alice's instrument. Solid lines represent quantum systems, double lines classical systems, and the dashed line the spacial separation between Alice and Bob.}
	\label{fig:Superchannel}
\end{figure*} 

In the following section, we will bound the entanglement that it costs to simulate an arbitrary channel between two parties with a given precision. Since any quantum state can be identified with its corresponding replacement channel, this can be seen as a generalization of the results presented in the previous section. Moreover, quantum teleportation~\cite{Bennett1993} demonstrates that LOCC and one ebit can be used to simulate one identity qubit channel. It will therefore become apparent when we talk about optimal protocols that our results are closely related to teleportation too, see Fig.~\ref{simcha}.

As in the state case, we will start by defining how we quantify the error of a simulation $\mN'$ of a given channel $\mN$: For any channel $\mN\in\cptp(A\to B)$ and a maximally entangled state $\Phi_{m}$, we define the channel conversion fidelity as
\begin{align}\label{eq:ConvFidChan}
	&F\left(\Phi_m\xrightarrow{\locc}\mN^{A\to B}\right)\\
	&\eqdef\sup_{\Theta}\min_{\psi\in\pure(A\tA)}F\left(\Theta\left[\Phi_m\right](\psi^{A\tA}),\mN^{\tA\to B}(\psi^{A\tA})\right),\nonumber
\end{align}
where the supremum is over all LOCC super-channels $\Theta$ that map the state $\Phi_m$ to a channel in $\cptp(\tA\to B)$. Analogously to the $\eps$-single-shot entanglement cost of a bipartite state defined in Eq.~\eqref{econe}, we use the channel conversion fidelity to define the $\eps$-single-shot entanglement cost of the channel $\mN^{A\to B}$ as 
\be\label{defcost}
\cost^\eps(\mN)\eqdef\inf_{m\in\mbb{N}}\Big\{\log m\;:\;P^2\left(\Phi_m\xrightarrow{\locc}\mN\right)\leq\eps\Big\}\;,
\ee
where 
\begin{align}
	P^2\left(\Phi_m\xrightarrow{\locc}\mN\right)\eqdef 1-F^2\left(\Phi_m\xrightarrow{\locc}\mN\right).
\end{align}

To bound $\cost^\eps(\mN)$, we will use that the supremum and minimum in the definition of the channel conversion fidelity can be exchanged: This is the content of the following Lemma, which can be shown with the help of Ref.~\cite[Prop.~8]{wang2019converse} and~\cite[Lem.~II.3]{leditzky2018approaches}.

\begin{lemma}\label{lem:MinMaxFid}
	Let $F\left(\Phi_m\xrightarrow{\locc}\mN^{A\to B}\right)$ be defined as in Eq.~\eqref{eq:ConvFidChan}. It holds that
	\begin{align}\label{eq:ConvFidChanExch}
		&F\left(\Phi_m\xrightarrow{\locc}\mN^{A\to B}\right)  \\
		&= \min_{\psi \in \pure(A\tilde A)} \sup_{\Theta}  F\left(\Theta[\Phi_m](\psi^{A\tilde A}), \mN^{\tilde A \to B}(\psi^{A\tilde A})\right) \nonumber
	\end{align}    
	where the supremum is again over all LOCC super-channels $\Theta$ that map the state $\Phi_m$ to a channel in $\cptp(\tA\to B)$.
\end{lemma}

Having access to $\Phi_m$ and LOCC, Alice and Bob can use quantum teleportation to simulate the identity channel $\id^{A'\to B'}$, where $|A'|=|B'|=m$. Conversely, $\id^{A'\to B'}$ allows us to create $\Phi_m$.  Therefore, $\Phi_m\xleftrightarrow{\locc}\id^{A'\to B'}$, i.e., the state $\Phi_m$ is equivalent to the channel $\id^{A'\to B'}$. In Eqs.~\eqref{eq:ConvFidChan} and~\eqref{eq:ConvFidChanExch}, we can thus replace $\Theta\left[\Phi_m\right]$  with $\Theta\left[\id^{A'\to B'}\right]$ and take the supremums over all LOCC super-channels that map the identity channel $\id^{A'\to B'}$ to a channel in $\cptp(\tA\to B)$. Importantly, every such super-channel can be expanded as 
\ba\label{2p51}
\Theta\left[\id^{A'\to B'}\right]&=\mF^{B'X\to B}\circ\id^{A'\to B'}\circ\mE^{\tA\to A'X}\\
&=\mF^{B'X\to B}\circ\mE^{\tA\to B'X}\\
&=\sum_{x\in[k]}\mF_{(x)}^{B'\to B}\circ\mE_x^{\tA\to B'}\;,
\ea
where $X$ is a classical system (that can be exchanged via LOCC), $\mE^{\tA\to A'X}\in\cptp(\tA\to A'X)$, and $\mF^{B'X\to B}\in\cptp(B'X\to B)$ (see Fig.~\ref{fig:Superchannel}). In the last line, we utilized that this can be seen as the average of an instrument $\{\mE_x^{\tA\to B'}\}_x$ and a channel $\mF_{(x)}^{B'\to B}$ conditioned on its classical outcome $x$. From this follows that the super-channels that we need to consider in Eqs.~\eqref{eq:ConvFidChan} and~\eqref{eq:ConvFidChanExch} are of the form shown in Fig.~\ref{simcha}, highlighting the close relation to quantum teleportation.

The Theorem below shows that the conversion fidelity is closely related to the monotones from Eq.~\eqref{kyfan}, which can be extended to quantum channels in the usual manner, i.e.,
\be
E_{(k)}\left(\mN^{A\to B}\right)\eqdef\max_{\psi\in\pure(A\tA)}E_{(k)}\left(\mN^{\tA\to B}(\psi^{A\tA})\right).
\ee
Using this definition, we get the following result, which is the channel analog of Thm.~\ref{lem1291}.

\begin{lemma}\label{lem:combined}
	Let $\mN\in\cptp(A\to B)$ be a quantum channel. It then holds that
	\begin{align}
		&1-E_{(m)}\left(\mN^{A\to B}\right)\nonumber\\
		&\leq F\left(\Phi_m\xrightarrow{\locc}\mN^{A\to B}\right)\nonumber \\
		&\leq \sqrt{1-E_{(m)}\left(\mN^{A\to B}\right)}\;.
	\end{align}
\end{lemma}

This allows us to provide the promised bounds on $\cost^\eps(\mN^{A\to B})$. 

\begin{theorem}\label{thm231} 
	Let $\mN\in\cptp(A\to B)$ be a quantum channel and $\eps\in[0,1)$. Then
	\begin{align}
		&\max_{\psi\in\pure(A\tA)}\inf_{\sigma^{XAB}} H_{\max}^{\eps}(A|X)_\sigma \nonumber \\
		&\leq \cost^\eps(\mN^{A\to B}) \nonumber \\
		&\leq\max_{\psi\in\pure(A\tA)}\inf_{\sigma^{XAB}} H_{\max}^{\eps/2}(A|X)_\sigma,
	\end{align}
	where the infimums are over all classical systems $X$ and all classical extensions $\sigma^{XAB}$ of $\sigma^{AB}=\mN^{\tA\to B}(\psi^{A\tA})$. Again, the infimums are attained for a regular/classical extension with $|X|=|AB|^2$.
\end{theorem}
In the limit of $\eps$ approaching zero, this Theorem provides the zero-error single-shot entanglement cost of a quantum channel. With the help of Thm.~\ref{cmaint}, we can express the Theorem as 
\begin{align}
	&\max_{\psi\in\pure(A\tA)}\cost^{\eps}\left(\mN^{\tA\to B}\left(\psi^{A \tA}\right)\right) \nonumber \\
	\leq& \cost^\eps(\mN^{A\to B}) \nonumber \\
	\leq&\max_{\psi\in\pure(A\tA)}\cost^{\eps/2}\left(\mN^{\tA\to B}\left(\psi^{A \tA}\right)\right) .
\end{align}
This shows that $\cost^\eps(\mN^{A\to B})$ is lower bounded by the cost of the most expensive state that we can create with its help. This is to be expected for a consistent definition, since otherwise one might be able to build an entanglement perpetuum mobile. It is however not obvious that this should also be an upper bound since we intend to simulate a channel on an unknown input state and potentially do not have access to another correlated system (such as system $A$ in Eq.~\eqref{eq:ConvFidChan}): Simply replacing the input state with (an approximation of) the corresponding output state of the channel we intend to simulate is thus not an option. We conclude this section by noting that with the help of postselection techniques~\cite{Christandl2009}, one can recover the asymptotic entanglement cost of a quantum channel~\cite{Berta2013} from Thm.~\ref{thm231}.

\section{One-way single-shot entanglement manipulation}\label{sec:OneWay}

\subsection{State distillation}
In the following, we will explore a resource measure with respect to one-way LOCC. While this measure may not exhibit monotonicity under arbitrary LOCC, it can still be a valuable tool for providing bounds on the distillable entanglement of mixed bipartite states, as we will see in the following.

The most general one-way LOCC operation that  Alice and Bob can perform is for Alice to apply a quantum instrument $\{\mE_x\}_{x\in[n]}$, with $\mE_x\in\cp(A\to A')$ and $\sum_{x=1}^n\mE_x\in\cptp(A\to A')$, send the outcome $x$ to Bob, who then applies a quantum channel $\mF_{(x)}\in\cptp(B\to B')$ that depends on the outcome $x$ received from Alice~\cite{Chitambar2014}. The overall operation can be described by the quantum channel
\be\label{13111c}
\mN^{AB\to A'B'}\eqdef\sum_{x=1}^n\mE_x^{A\to A'}\otimes\mF_{(x)}^{B\to B'}\;.
\ee

\begin{definition}\label{def:cohInf}
	The \emph{coherent information of entanglement} of a state $\rho\in\md(AB)$ is defined as
	\be\label{13p149}
	E_{\to}\left(\rho^{AB}\right)\eqdef\sup_{\mE\in\cptp(A\to AX)}I\big(A\ra BX\big)_{\mE(\rho)}\;,
	\ee
	where the supremum includes a supremum over all classical systems $X$ with arbitrary dimension and
	\be
	I(A\ra B)_\rho\eqdef -H(A|B)_\rho
	\ee 
	is the coherent information~\cite{Schumacher1996}.
\end{definition}

As promised, we will now show that the coherent information of entanglement is a resource measure with respect to one-way LOCC.

\begin{theorem}
	Let $\rho\in\md(AB)$ with $m=|A|=|B|$, $\sigma\in\md(A'B')$, and $\mN\in\locc_1(AB\to A'B')$. The coherent information of entanglement $E_\to$ is
	\begin{enumerate}
		\item monotonic under one-way LOCC, i.e., $$E_{\to}\left(\mN^{AB\to A'B'}\left(\rho^{AB}\right)\right)\leq E_{\to}\left(\rho^{AB}\right)\;,$$
		\item  non-negative, i.e., $E_\to(\rho^{AB})\geq 0$, with equality if $\rho^{AB}$ is separable,
		\item strongly monotonic under one-way LOCC, i.e., for any ensemble $\{p_y,\sigma_y^{A'B'}\}$ that can be obtained from $\rho^{AB}$ using one-way LOCC and subselection, it holds that
		$$E_{\to}\left(\rho^{AB}\right)
		\ge\sum_y p_y E_{\to}\left(\sigma^{A'B'}_y\right),$$
		\item convex,
		\item bounded by $E_\to\left(\rho^{AB}\right)\le \log(m)=E_\to\left(\Phi_m\right)$,
		\item and superadditive, i.e., $$ \qquad E_\to\left(\rho^{AB}\otimes\sigma^{A'B'}\right)\geq E_\to\left(\rho^{AB}\right)+E_\to\left(\sigma^{A'B'}\right).$$  
	\end{enumerate}
\end{theorem}

The $\eps$-single-shot distillable entanglement under one-way LOCC (and the error bounded by the trace distance) is defined as
\begin{align}\label{f1131}
	&\distill_{\to}^{\eps}\left(\rho^{AB}\right)\\
	&\quad\eqdef\max\Big\{\log m\;:\; T\left(\rho^{AB}\xrightarrow{\locc_1}  \Phi_m\right)\leq\eps\Big\}\;.\nonumber
\end{align}
A simple formula for the above expression is presently not available. However, we can provide an upper bound.
\begin{theorem}\label{epiu}
	Let $\rho\in\md(AB)$ and $\eps\in(0,1/2)$. Then, the one-way $\eps$-single-shot distillable entanglement is bounded by
	\be\label{singleed}
	\distill^\eps_{\to}\left(\rho^{AB}\right)\leq\frac1{1-2\eps}E_\to\left(\rho^{AB}\right)+\frac{1+\eps}{1-2\eps}h\left(\frac\eps{1+\eps}\right)\;,
	\ee 
	where $h(x)\eqdef-x\log x-(1-x)\log(1-x)$ is the binary Shannon entropy.
\end{theorem}
This result should be compared to Ref.~\cite[Lem.~4]{Buscemi2010}, where also an upper bound on a variant of the one-way $\eps$-single-shot distillable entanglement is provided (again with a slightly different definition of the allowed error in terms of the fidelity). The main difference is that in our bound, the optimization over instruments (contained in $E_\to\left(\rho^{AB}\right)$) is independent of $\eps$, whilst the corresponding optimization in their bound is not. The bound provided in Thm.~\ref{epiu} recovers the exact asymptotic solution given in Ref.~\cite[Thm.~13]{Devetak2005}. 

For lower bounds on the one-way $\eps$-single-shot distillable entanglement, see again Ref.~\cite{Buscemi2010} as well as Ref.~\cite[Prop.~21]{Wilde2017}. Even though Ref.~\cite{Wilde2017} defined the conversion distance using the fidelity, their bound holds for our definition too due to the following Lemma. 
\begin{lemma}\label{lem1331}
    Let $\rho\in\md(AB)$, and $\Phi_m\in\md(A'B')$ be the maximally entangled state with $m\eqdef|A'|=|B'|$. Then, 
    \begin{align}\label{13102}
    T\left(\rho\xrightarrow{\locc_1}  \Phi_m\right)&=P^2\left(\rho\xrightarrow{\locc_1}  \Phi_m\right) \nonumber\\
    &=1-\sup_{\mN\in\locc_1}\tr\left[\Phi_m\mN\left(\rho\right)\right]\;,
    \end{align}
    where the supremum is over all $\mN\in\locc_1(AB\to A'B')$, and $P$ is the purified distance as given in Eq.~\eqref{eq:purifieddistance}.
\end{lemma}

\subsection{Channel distillation}

In Def.~\ref{def:cohInf}, we introduced the coherent information of entanglement of a quantum state and subsequently showed that it is a measure of entanglement under one-way LOCC. The following Definition contains the generalization to quantum channels.
\begin{definition}
	Let $\mN\in\cptp(A\to B)$ be a quantum channel. Its coherent information of entanglement is then defined as 
	\be\label{defcic}
	E_{\to}\left(\mN^{A\to B}\right)\eqdef\max_{\psi\in\pure(A\tA)}E_{\to}\left(\mN^{\tA\to B}\left(\psi^{A\tA}\right)\right)\;.
	\ee
\end{definition}
Similarly, one can define the coherent information of a quantum channel $\mN^{A\to B}$ as
\be\label{eq:CohInfChan}
I(A\ra B)_\mN\eqdef \max_{\psi\in\pure(A\tA)}I\big(A\ra B\big)_{\mN^{\tA\to B}\left(\psi^{A\tA}\right)}\;.
\ee
Interestingly, $E_{\to}\left(\mN^{A\to B}\right)$ and $I(A\ra B)_\mN$ coincide.
\begin{theorem}\label{thm:equivCohInf}
	Let $\mN\in\cptp(A\to B)$ be a quantum channel. It then holds that
	\be
	E_{\to}\left(\mN^{A\to B}\right)=I(A\ra B)_\mN\;.
	\ee
\end{theorem}

Consider a channel $\mN\in\cptp(A\to B)$. The most general bipartite state $\rho\in\md(A_3 B_3)$ to which this channel can be converted with the help of $\locc_1$ is given by
\begin{align}
\rho^{A_3B_3}=\mE^{A_2B\to A_3B_3}\circ\mN^{A\to B}\left(\sigma^{AA_2}\right)\;,
\end{align}
where $\mE$ is in $\locc_1$. Since we can purify $\sigma$ by enlarging system $A_2$ and adapting $\mE$ accordingly, w.l.o.g., we will assume that $\sigma$ is pure. For any natural number $m$, let $|A_3|=|B_3|=m$ and we thus define the conversion distance
\begin{align}
	&T\left(\mN^{A\to B}\xrightarrow{\locc_1}\Phi_m\right) \\
	&\eqdef\frac12\inf_{\substack{\mE\in\locc_1\\\psi\in\pure}}\left\|\Phi_m^{A_3B_3}\!-\!\mE^{A_2B\to A_3B_3}\circ\mN^{A\to B}\left(\psi^{AA_2}\right)\right\|_1, \nonumber
\end{align}
where a priori, the infimum also includes an infimum over $|A_2|$.
Since the Schmidt rank of $\psi^{AA_2}$ cannot exceed $|A|$, w.l.o.g., we can fix $|A_2|=|A|$ (and adapt $\mE$ accordingly).
Observe that the equation above implies that (see Eq.~\eqref{eq:convdist})
\begin{align}\label{2.3n}
	&T\left(\mN^{A\to B}\xrightarrow{\locc_1}\Phi_m\right)\\
	&=\min_{\psi,\phi\in\pure}
	T\left(\mN^{A\to B}\left(\psi^{A\tA}\right)\xrightarrow{\locc_1}\Phi_m\right).\nonumber
\end{align}
The one-way $\eps$-single-shot distillable entanglement is then defined as (cf. Eq.~\eqref{f1131})
\be\label{edone}
\distill^{\eps}_\to\left(\mN\right)\eqdef\max\Big\{\log m:\;T\left(\mN\xrightarrow{\locc_1}  \Phi_m\right)\leq\eps\Big\}.
\ee

Therefore, from Eqs.~\eqref{f1131} and~\eqref{2.3n}, we get that 
\begin{align}
	&\distill^{\eps}_\to\left(\mN^{A\to B}\right)=\max_{\psi,\phi\in\pure} \distill^{\eps}\left(\mN^{A\to B}\left(\psi^{A\tA}\right)\right).
\end{align}
Combining Thm.~\ref{epiu} and Thm.~\ref{thm:equivCohInf}, we thus obtain
\begin{theorem}\label{thm:DistChan}
	Let $\mN\in\cptp(A\to B)$ be a quantum channel. It then holds that
	\be
	\distill^\eps_{\to}\left(\mN\right)\leq\frac1{1-2\eps}I(A\ra B)_\mN+\frac{1+\eps}{1-2\eps}h\left(\frac\eps{1+\eps}\right)\;.
	\ee
\end{theorem}

\section{Discussion and outlook}
In this work, we studied entanglement distillation and dilution of states and channels in the single-shot regime~\cite{Bennett1996a, Bennett1996b, Buscemi2011, Das2020, Bauml2019, Wang2023, Kim2021,Kim2021b, fang2019non, Fang2020nogo, Regula2021, Takagi2022, Fang2022}. By restricting the allowed error $\eps$ with the recently introduced star conversion distance~\cite{Zanoni2023}, we determined surprisingly compact closed-form expressions for the $\eps$-single-shot entanglement cost and the $\eps$-single-shot distillable entanglement of pure states, which allowed us to obtain second-order asymptotics~\cite{Kumagai2017}. Since these results are based on (approximate) majorization~\cite{Marshall2011,Horodecki2018}, we expect that similar results can be obtained in other majorization-based resource theories such as coherence~\cite{Baumgratz2014,Streltsov2017}, non-uniformity~\cite{GMNSYH15}, or quantum thermodynamics~\cite{Janzing2000, Egloff2015, Aberg2013, BHORS13,Lostaglio2019}.

For mixed states, we expressed the $\eps$-single-shot entanglement cost introduced in Ref.~\cite{Buscemi2011} in terms of a smoothed version of the conditional max-entropy. As a Corollary, we showed that this equips the smoothed max-entropy with an operational interpretation in terms of the $\eps$-single-shot entanglement cost of pure states. On pure states, it coincides with the $\eps$-single-shot entanglement cost based on the star conversion distance. It is an interesting open question why this is the case. Based on these results, we provided bounds on the entanglement cost of quantum channels that coincide in the zero-error limit. Concerning entanglement distillation, we introduced the coherent information of entanglement and used it to upper bound both the one-way $\eps$-single-shot distillable entanglement of states and channels.

Our work thus contributes to a better understanding of how entanglement can be optimally used to implement a desired quantum channel and how this is related to entropic quantities. This is highly relevant to optimize technological applications in which entanglement plays a role, which will lead to a better understanding of the relevance of entanglement for quantum advantages.

\begin{acknowledgments}
T. T. and G. G. acknowledge support from the Natural
Sciences and Engineering Research Council of Canada
(NSERC). T. T. acknowledges support from the Pacific
Institute for the Mathematical Sciences (PIMS). The research and findings may not reflect those of the Institute.
\end{acknowledgments}

%

\newpage

\onecolumngrid
\newpage
\appendix
\section{Additional remarks}
In the Appendix, for completeness, we provide a few additional comments. 
Technically Eq.~\eqref{eq:equiv} was not shown in Ref.~\cite{Zanoni2023}, but it follows directly from what was shown: As expected for a consistent definition, and because one can always append and remove separable states reversibly,
\begin{align} \label{eq:ConvDistCons}
	&T \left(\rho^{AB} \xrightarrow{\locc} \sigma^{A'B'}\right) \nonumber \\
	&=T \left(\rho^{AB} \otimes \ketbra{11}{11}^{A'B'} \xrightarrow{\locc} \ketbra{11}{11}^{AB}\otimes \sigma^{A'B'}\right).
\end{align}
A technical proof is the following:
\begin{align}
	& T \left(\rho^{AB} \otimes \ketbra{11}{11}^{A'B'} \xrightarrow{\locc} \ketbra{11}{11}^{AB}\otimes \sigma^{A'B'}\right) \nonumber \\
	&= \inf_{\tau\in\md(AA'BB')}\left\{ \frac12 \norm{\tau^{AA'BB'}- \ketbra{11}{11}^{AB}\otimes \sigma^{A'B'}}_1 :\;\rho^{AB} \otimes \ketbra{11}{11}^{A'B'}\xrightarrow{\locc} \tau^{AA'BB'}\right\} \nonumber \\
	&\le \inf_{\tau\in\md(A'B')}\left\{ \frac12 \norm{\ketbra{11}{11}^{AB}\otimes\tau^{A'B'}- \ketbra{11}{11}^{AB}\otimes \sigma^{A'B'}}_1 :\;\rho^{AB} \otimes \ketbra{11}{11}^{A'B'}\xrightarrow{\locc} \ketbra{11}{11}^{AB}\otimes\tau^{A'B'}\right\} \nonumber \\
	&=\inf_{\tau\in\md(A'B')}\left\{ \frac12 \norm{\tau^{A'B'}- \sigma^{A'B'}}_1   :\;\rho^{AB} \xrightarrow{\locc} \tau^{A'B'}\right\} \nonumber \\
	&=T \left(\rho^{AB} \xrightarrow{\locc} \sigma^{A'B'}\right)
\end{align}
and
\begin{align}
	&T \left(\rho^{AB} \otimes \ketbra{11}{11}^{A'B'} \xrightarrow{\locc} \ketbra{11}{11}^{AB}\otimes \sigma^{A'B'}\right) \nonumber \\
	&= \inf_{\tau\in\md(AA'BB')}\left\{ \frac12 \norm{\tau^{AA'BB'}- \ketbra{11}{11}^{AB}\otimes \sigma^{A'B'}}_1 :\;\rho^{AB} \otimes \ketbra{11}{11}^{A'B'}\xrightarrow{\locc} \tau^{AA'BB'}\right\} \nonumber \\
	&\ge \inf_{\tau\in\md(AA'BB')}\left\{ \frac12 \norm{\tr_{AB}\tau^{AA'BB'}- \sigma^{A'B'}}_1 :\;\rho^{AB} \otimes \ketbra{11}{11}^{A'B'}\xrightarrow{\locc} \tau^{AA'BB'}\right\} \nonumber \\
	&\ge \inf_{\tau\in\md(AA'BB')}\left\{ \frac12 \norm{\tr_{AB}\tau^{AA'BB'}- \sigma^{A'B'}}_1 :\;\rho^{AB} \otimes \ketbra{11}{11}^{A'B'}\xrightarrow{\locc} \tr_{AB}\tau^{AA'BB'}\right\} \nonumber \\
	&=\inf_{\tau\in\md(A'B')}\left\{ \frac12 \norm{\tau^{A'B'}- \sigma^{A'B'}}_1   :\;\rho^{AB} \xrightarrow{\locc} \tau^{A'B'}\right\} \nonumber \\
	&=T \left(\rho^{AB} \xrightarrow{\locc} \sigma^{A'B'}\right).
\end{align}
We, therefore, find with the help of Ref.~\cite[Thm.~3]{Zanoni2023} that
\begin{align}
	&\frac{1}{2}T_\star^2 \left(\psi^{AB} \xrightarrow{\locc} \phi^{A'B'}\right)\nonumber\\
	&=\frac{1}{2}T_\star^2 \left(\psi^{AB} \otimes \ketbra{11}{11}^{A'B'} \xrightarrow{\locc} \ketbra{11}{11}^{AB}\otimes \phi^{A'B'}\right) \nonumber \\
	&\le T \left(\psi^{AB} \otimes \ketbra{11}{11}^{A'B'} \xrightarrow{\locc} \ketbra{11}{11}^{AB}\otimes \phi^{A'B'}\right) \nonumber \\
	&= T \left(\psi^{AB} \xrightarrow{\locc} \phi^{A'B'}\right) \nonumber \\
	&= T \left(\psi^{AB} \otimes \ketbra{11}{11}^{A'B'} \xrightarrow{\locc} \ketbra{11}{11}^{AB}\otimes \phi^{A'B'}\right) \nonumber \\
	&\le \sqrt{2T_\star\left(\psi^{AB} \otimes \ketbra{11}{11}^{A'B'} \xrightarrow{\locc} \ketbra{11}{11}^{AB}\otimes \phi^{A'B'}\right)} \nonumber \\
	&= \sqrt{2T_\star\left(\psi^{AB} \xrightarrow{\locc} \phi^{A'B'}\right)}.
\end{align}

Analogously to Eq.~\eqref{eq:ConvDistCons}, and as expected, it also holds that
\begin{align}\label{eq:PurConvDistCons}
	&P\left(\rho^{AB} \xrightarrow{\locc} \sigma^{A'B'}\right) \nonumber\\
	&=P\left(\rho^{AB}\otimes\ketbra{11}{11}^{A'B'} \xrightarrow{\locc} \ketbra{11}{11}^{AB} \otimes\sigma^{A'B'}\right),
\end{align}
which follows from exactly the same arguments.

\section{Proofs of the results in the main text}\label{sec:proofs}
\setcounter{theorems}{0} 
In the following, we provide the proofs of the results presented in the main text, which we repeat for readability. 

\begin{theorem}
Let $\psi\in \pure(AB), \phi\in \pure(A'B')$, and $\p\in\prob^\downarrow(|A|), \q\in\prob^\downarrow(|A'|)$ be their corresponding Schmidt coefficients. Then,
\ba
T_{\star}\left(\psi^{AB}\to\phi^{A'B'}\right)=&\max_{k\in[\sr(\psi^{AB})]}\big\{\|\p\|_{(k)}-\|\q\|_{(k)}\big\}\;.
\ea
\end{theorem}
\begin{proof}
    We notice that the (by convention ordered) Schmidt coefficients of 
    \begin{align}
        \psi^{AB}\otimes\ketbra{11}{11}^{A'B'}
    \end{align}
    are given by
    \begin{align}
         \e_1^{(|A'|)}\otimes\p,
    \end{align}
    where $\e_1^{(|A'|)}=(1,0,...,0)\in \prob^\downarrow(|A'|)$. The Schmidt coefficients of 
    \begin{align}
        \ketbra{11}{11}^{AB}\otimes \phi^{A'B'}
    \end{align}
    on the other hand, are given by
    \begin{align}
        \e_1^{(|A|)}\otimes\q.
    \end{align}
     The claim then follows from Eq.~\eqref{eq:TstarGen} and  Ref.~\cite[Thm.~3]{Zanoni2023}.
\end{proof}

\begin{theorem} 
	Let $\psi\in\pure(AB),\phi\in\pure(A'B')$. Then,
	\be
	P\left(\psi^{AB}\to\phi^{A'B'}\right)=P_{\star}\left(\psi^{AB}\to\phi^{A'B'}\right)\;.
	\ee
\end{theorem}
\begin{proof}
	According to Eq.~\eqref{eq:PurConvDistCons} and Ref.~\cite{Zanoni2023},
	\begin{align}
		&P\left(\psi^{AB} \to \phi^{A'B'}\right)  \nonumber\\
		&=P\left(\psi^{AB}\otimes\ketbra{11}{11}^{A'B'} \to \ketbra{11}{11}^{AB} \otimes\phi^{A'B'}\right) \nonumber \\
		&=P_\star\left(\psi^{AB}\otimes\ketbra{11}{11}^{A'B'} \to \ketbra{11}{11}^{AB} \otimes\phi^{A'B'}\right)\nonumber\\
		&=P_\star\left(\psi^{AB} \to \phi^{A'B'}\right).
	\end{align}
\end{proof}

\begin{theorem}
Let $\eps\in[0,1)$, $\psi\in\pure(AB)$, $d\eqdef\sr(\psi^{AB})$, and $\p\in\prob^\downarrow(|A|)$ be the Schmidt coefficients of $\psi^{AB}$. The $\eps$-single-shot distillable entanglement of $\psi^{AB}$ is then given by
\be
\distill^\eps\left(\psi^{AB}\right)=\min_{k\in\{\ell,\ldots,d\}} \log\left\lfloor\frac{k}{\|\p\|_{(k)}-\eps}\right\rfloor\;,
\ee
where $\ell\in[d]$ is the integer satisfying 
$\|\p\|_{(\ell-1)}\leq\eps<\|\p\|_{(\ell)}$. 
\end{theorem}

\begin{proof}
From Thm.~\ref{thm:StarDist}, we find that for any $m\in\mbb{N}$
\ba\label{12a1}
T_\star\left(\psi^{AB}\to\Phi_m\right)&=\max_{k\in[d]}\left\{\|\p\|_{(k)}-\|\Phi_m\|_{(k)}\right\}.
\ea
For $m\ge d$, this implies that
\be
T_\star\left(\psi^{AB}\to\Phi_m\right)=\max_{k\in[d]}\left\{\|\p\|_{(k)}-\frac km\right\}\;.
\ee

Next, consider the case $m<d$, i.e.,
\ba\label{eq:distMlowerD}
T_\star\left(\psi^{AB}\to\Phi_m\right)=&\max_{k\in[d]}\left\{\|\p\|_{(k)}-\frac{\min\{k,m\}}{m}\right\}.
\ea
Now assume that there exists an optimizer $k^\star$ in the above expression that satisfies $k^\star\ge m$. In this case, it must be equal to $d$, since $\|\p\|_{(k)}$ is strictly increasing for $k\in\{m,...,d\}$, whilst $\frac{\min\{k,m\}}{m}=1$ is constant. This implies that  $T_\star\left(\psi^{AB}\to\Phi_m\right)=0$.
Assume on the contrary that there only exist optimizers $k^\star< m$, and thus $T_\star\left(\psi^{AB}\to\Phi_m\right)>0$ (otherwise $k^\star=d$ would be an optimizer too). In this case, we have that 
\ba
T_\star\left(\psi^{AB}\to\Phi_m\right)=&\max_{k\in[d]}\left\{\|\p\|_{(k)}-\frac{\min\{k,m\}}{m}\right\}\\
=& \max_{k\in[d]}\left\{\|\p\|_{(k)}-\frac{k}{m}\right\},
\ea
since the maximum in the first line is reached for some $k< m$, and 
\ba
    \frac{\min\{k,m\}}{m}\le\frac{k}{m}
\ea
with equality for $k\le m$. 

In summary, we thus find that for all $m\in\mbb{N}$, either $T_\star\left(\psi^{AB}\to\Phi_m\right)=0$
, or
\be\label{eq:TstarDist}
T_\star\left(\psi^{AB}\to\Phi_m\right)=\max_{k\in[d]}\left\{\|\p\|_{(k)}-\frac km\right\}\;.
\ee
Now let $\tilde{m}$ be the largest $m$ such that $T_\star\left(\psi^{AB}\to\Phi_m\right)=0$ (and note that $\tilde{m}\le d$, since $\locc$ cannot increase the Schmidt rank of any state). 
Remembering that by definition
\be
\distill^\eps\left(\psi^{AB}\right)\eqdef\max_{m\in\mbb{N}}\Big\{\log m\;:\;T_\star\left(\psi^{AB}\to\Phi_m\right)\leq\eps\Big\}, \nonumber
\ee
this implies that 
\be\label{eq:NecSufTildeM}
\distill^\eps\left(\psi^{AB}\right)\ge \log{\tilde{m}},
\ee
with equality iff $T_\star\left(\psi^{AB}\to\Phi_{\tilde{m}+1}\right)>\eps$. Assuming that this is not the case, i.e., that $\eps\ge T_\star\left(\psi^{AB}\to\Phi_{\tilde{m}+1}\right)>0$, it follows from Eq.~\eqref{eq:TstarDist} that 
\ba\label{eq:Dist1}
&\distill^\eps\left(\psi^{AB}\right) \\ 
&=\max_{m\in\mbb{N}}\left\{\log m\;:\;\|\p\|_{(k)}-\frac km\leq\eps\;\;\forall\;k\in[d]\right\}\\
&=\max_{m\in\mbb{N}}\left\{\log m\;:\;\|\p\|_{(k)}-\frac km\leq\eps\;\;\forall\;k\in\{\ell,\ldots,d\}\right\} \\
&=\max_{m\in\mbb{N}}\left\{\log m\;:\;m\le\frac{k}{\|\p\|_{(k)}-\eps}\;\;\forall\;k\in\{\ell,\ldots,d\}\right\} \\
&=\max_{m\in\mbb{N}}\left\{\log m\;:\;m\le\min_{k\in\{\ell,...,d\}}\frac{k}{\|\p\|_{(k)}-\eps}\right\} \\
&=\min_{k\in\{\ell,...,d\}} \log \left\lfloor\frac{k}{\|\p\|_{(k)}-\eps}\right\rfloor.
\ea

Moreover, by definition, 
\begin{align}
    \tilde{m}=\max\left\{m\in [d]: \|\p\|_{(k)}-\frac{\min\{k,m\}}{m}\leq 0 \;\forall k\in[d] \right\}. \nonumber
\end{align}
Since $ \|\p\|_{(k)}\le 1$, for $k> m$, $\|\p\|_{(k)}-\frac{\min\{k,m\}}{m}=\|\p\|_{(k)}-1\le0$ and thus
\ba
    \tilde{m}=&\max\left\{m\in [d]: \|\p\|_{(k)}-\frac{k}{m}\leq 0 \;\forall k\in[m] \right\}\\
    =&\max\left\{m\in [d]: \|\p\|_{(k)}-\frac{k}{m}\leq 0 \;\forall k\in[d] \right\},
\ea
from which follows that
\ba
\max_{k\in [d]} \left\{ \|\p\|_{(k)}-\frac{k}{\tilde{m}} \right\} \leq 0 \le \eps.
\ea
To conclude the proof, assume that $\distill^\eps\left(\psi^{AB}\right)=\log{\tilde{m}}$.
According to Eq.~\eqref{eq:NecSufTildeM}, this implies that 
\ba
T_\star\left(\psi^{AB}\to\Phi_{\tilde{m}+1}\right)=\max_{k\in[d]}\left\{\|\p\|_{(k)}-\frac{k}{\tilde{m}+1}\right\}>\eps\;, \nonumber
\ea
and consequently
\ba
    \tilde{m}=\max\left\{m\in\mbb{N}\;:\;\max_{k\in[d]}\left\{\|\p\|_{(k)}-\frac{k}{m}\right\} \le\eps \right\}.
\ea
From this follows again that 
\begin{align}
    &\distill^\eps\left(\psi^{AB}\right)=\log \tilde{m}\\
    &=\max_{m\in\mbb{N}} \left\{\log m:\|\p\|_{(k)}-\frac km\leq\eps\;\forall\;k\in[d] \right\} \nonumber
\end{align}
Continuing with the same steps as in Eq.~\eqref{eq:Dist1} finishes the proof.
\end{proof}

\begin{theorem}
Let $\eps\in[0,1)$, $\psi\in\pure(AB)$, $d\eqdef\sr(\psi^{AB})$, and $\p\in\prob^\downarrow(|A|)$ be the Schmidt coefficients of $\psi^{AB}$.
The $\eps$-single-shot entanglement cost of $\psi^{AB}$ is then given by
\be
\cost^\eps\left(\psi^{AB}\right)=\log m,
\ee
where $m\in[d]$ is the integer satisfying 
$
\|\p\|_{(m-1)}<1-\eps\leq\|\p\|_{(m)}
$.
\end{theorem}

\begin{proof}
Let 
\be
b_k\eqdef\frac{k}{\|\p\|_{(k)}+\eps}\quad\quad\forall\;k\in[m]\;.
\ee
Since for all $k\in[m-1]$, it holds that
\begin{align}
    \|\p\|_{(k)}+\eps \ge \|\p\|_{(k)}\ge p_1\ge p_{k+1},
\end{align}
we find that 
\begin{align}
    \frac{b_{k+1}}{b_k}=&\frac{k+1}{k}\frac{\|\p\|_{(k)}+\eps}{\|\p\|_{(k)}+p_{k+1}+\eps} \nonumber \\
    =&\frac{k(\|\p\|_{(k)}+\eps)+\|\p\|_{(k)}+\eps}{k(\|\p\|_{(k)}+\eps) +p_{k+1}}\ge 1
\end{align}
and thus
\be
b_{1}\leq b_2\leq\cdots\leq b_m\;.
\ee
From Thm.~\ref{thm:StarDist}, it follows that for any $m\in\mbb{N}$ 
\be\label{cost1243}
T_\star\left(\Phi_m\to\psi^{AB}\right)=\max_{k\in[m]}\left\{\frac km-\|\p\|_{(k)}\right\}\;.
\ee
In combination, it, therefore, holds that
\begin{align}
&\cost^\eps\left(\psi^{AB}\right)\\
&=\min_{m\in\mbb{N}}\left\{\log m\;:\;\frac km-\|\p\|_{(k)}\leq\eps\;\;\forall\;k\in[m]\right\}\\
&=\min_{m\in\mbb{N}}\left\{\log m\;:\;m\geq \frac{k}{\|\p\|_{(k)}+\eps}\;\;\forall\;k\in[m]\right\}\\
&=\min_{m\in\mbb{N}}\left\{\log m\;:\;m\geq \frac{m}{\|\p\|_{(m)}+\eps}\right\}\\
&=\min_{m\in\mbb{N}} \big\{\log m\;:\;\|\p\|_{(m)}\geq1-\eps\big\}\;.
\end{align}
Noticing that $\|\p\|_{(d)}=1$ completes the proof.
\end{proof}
The proof of Thm.~\ref{thm:effComp} can be found in App.~\ref{app:computability}.

\setcounter{theorems}{5} 

\begin{lemma}
	For any distribution $\p$ such that $V(\p) > 0$, any natural number $n$, and $\eps \in [0,1)$, let
	\begin{align}
		f_{n,\eps}(\p):=\min \left\{k: \|\p^{\otimes n}\|_{(k)} > \eps\right\}, \nonumber \\
		f'_{n,\eps}(\p):=\min \left\{k: \|\p^{\otimes n}\|_{(k)} \geq \eps\right\}.
	\end{align}
	Then we have that
	\begin{align}
		\lim_{n\to \infty} \frac{\log {f'_{n,\eps}(\p)} - n H(\p)}{\sqrt{n V(\p)}}=\lim_{n\to \infty} \frac{\log {f_{n,\eps}(\p)} - n H(\p)}{\sqrt{n V(\p)}} = \Phi^{-1}(\eps).
	\end{align}    
\end{lemma}
\begin{proof}
	According to Ref.~\cite[Lem.~12]{Kumagai2017} (see also Ref.~\cite[Lem.~16]{chubb2018beyond}), it holds that for any distribution $\p$ such that $V(\p) > 0$, 
	\begin{align}\label{lem: second-order}
		\lim_{n\to \infty} \|\p^{\otimes n}\|_{(k_n(x))} = \lim_{n\to \infty} \sum_{i=1}^{k_n(x)} (\p^{\otimes n})_i^\downarrow = \Phi(x),
	\end{align}
	where 
	\begin{align}\label{eq:kn}
		k_n(x):= \left\lfloor \exp{\left(H(\p^{\otimes n}) + x \sqrt{V(\p^{\otimes n})}\right)} \right\rfloor.
	\end{align}
	In the following, we will use this to prove the Lemma.
	
	For any $\delta > 0$, let $x = \Phi^{-1}(\eps + 2\delta)$. By Eq.~\eqref{lem: second-order}, we have for sufficiently large $n$ that
	\begin{align}
		\|\p^{\otimes n}\|_{(k_n(x))} \geq \Phi(x) - \delta = \eps + \delta > \eps.
	\end{align}
	By the definition of $f_{n,\eps}(\p)$ and $k_n(x)$, this implies that
	\begin{align}
		\log f_{n,\eps}(\p) \leq  \log k_n(x) \leq {H(\p^{\otimes n}) + x \sqrt{V(\p^{\otimes n})}}.
	\end{align}
	Taking $n\to \infty$, we get
	\begin{align}
		\limsup_{n\to \infty} \frac{\log {f_{n,\eps}(\p)} - n H(\p)}{\sqrt{n V(\p)}} \leq x = \Phi^{-1}(\eps + 2\delta).
	\end{align}
	Since the above inequality holds for any $\delta > 0$, we have
	\begin{align}
		\limsup_{n\to \infty} \frac{\log {f_{n,\eps}(\p)} - n H(\p)}{\sqrt{n V(\p)}} \leq \Phi^{-1}(\eps).
	\end{align}

	Next, we prove the other direction by contradiction. Suppose
	\begin{align}
		\liminf_{n\to \infty} \frac{\log f_{n,\eps}(\p) - n H(\p)}{\sqrt{n V(\p)}} < \Phi^{-1}(\eps),    
	\end{align}
	i.e., there exists a value $r$ such that 
	\begin{align}
		\liminf_{n\to \infty} \frac{\log f_{n,\eps}(\p) - n H(\p)}{\sqrt{n V(\p)}} \leq r < \Phi^{-1}(\eps).    
	\end{align}
	For any 
	\begin{align}\label{eq:boundPhi}
		0<\delta<\Phi^{-1}(\eps) - r,
	\end{align}
	there thus exists a subsequence of $n$ (denoted as $n$ as well) such that 
	\begin{align}
		\frac{\log f_{n,\eps}(\p) - n H(\p)}{\sqrt{n V(\p)}} \leq r + \delta
	\end{align} 
	for sufficiently large $n$.
	This is equivalent to
	\begin{align}
		f_{n,\eps}(\p) \leq  \exp\left(n H(\p) + (r + \delta)\sqrt{n V(\p)}\right) \leq \left\lfloor \exp\left({n H(\p) + (r + \delta)\sqrt{n V(\p)}} \right)\right\rfloor + 1.
	\end{align}
	Since $\|\p^{\otimes n}\|_{(k)}$ is non-decreasing in $k$, we have by the definition of $k_n$ in Eq.~\eqref{eq:kn} that
	\begin{align}
		\|\p^{\otimes n}\|_{\left(f_{n,\eps}(\p)\right)} \leq \|\p^{\otimes n}\|_{(k_n(r+\delta)+1)} = \|\p^{\otimes n}\|_{(k_n(r+\delta))} + (\p^{\otimes n})^\downarrow_{(k_n(r+\delta)+1)}.
	\end{align}
	Taking $n\to \infty$ on both sides and using Eq.~\eqref{lem: second-order} as well as Eq.~\eqref{eq:boundPhi}, we have
	\begin{align}\label{eq: tmp1}
		\liminf_{n\to \infty} \|\p^{\otimes n}\|_{({f_{n,\eps}(\p)})} \leq \Phi(r+\delta) < \eps.
	\end{align}
	However, by the definition of $f_{n,\eps}(\p)$, we always have $\|\p^{\otimes n}\|_{({f_{n,\eps}(\p)})}>\eps$, which forms a contradiction to Eq.~\eqref{eq: tmp1}. Since we are working in the asymptotic regime, the proof for $f'_{n,\eps}(\p)$ works exactly analogously.
\end{proof}
 
\begin{proposition}
	For any pure state $\psi\in\pure(AB)$ with Schmidt vector $\p$, $V(\p) > 0$, and $\eps \in [0,1)$, it holds that
	\begin{align}
		\cost^\eps(\psi^{\otimes n}) = n H(\p) - \Phi^{-1}(\eps) \sqrt{n V(\p)} + o(\sqrt{n}).
	\end{align}   
\end{proposition}
\begin{proof}
	By Thm.~\ref{thm:PureSingleCost}, we have that
	\begin{align}
		\cost^\eps(\psi^{\otimes n}) = \log f'_{n,1-\eps}(\p).    
	\end{align}
	The Proposition is thus simply a rewriting of Lem.~\ref{lem: second-order 1} where we used that $\Phi^{-1}(1-\eps) = - \Phi^{-1}(\eps)$. 
\end{proof}

\begin{proposition}
	For any pure state $\psi\in\pure(AB)$ with Schmidt vector $\p$, $V(\p) > 0$, and $\eps \in [0,1)$, it holds that
	\begin{align}
		\distill^\eps(\psi^{\otimes n}) = n H(\p) + \Phi^{-1}(\eps) \sqrt{n V(\p)} + o(\sqrt{n}).
	\end{align}
\end{proposition}

\begin{proof}
	For any $\delta > 0$, let $x = \Phi^{-1}(\eps + 2\delta)$. Using Eq.~\eqref{lem: second-order}, for any sufficiently large $n$, we have $\|\p^{\otimes n}\|_{(k_n(x))} \geq \Phi(x) - \delta = \eps + \delta > \eps$ where $k_n(x)$ is defined in Eq.~\eqref{eq:kn}. Due to Thm.~\ref{thm:PureSingleDist}, we find that
	\begin{align}
		\distill^{\eps}(\psi^{\otimes n}) & \leq \log  \left\lfloor \frac{k_n(x)}{\|\p^{\otimes n}\|_{(k_n(x))} - \eps} \right \rfloor\\
		& \leq \log  \frac{k_n(x)}{\|\p^{\otimes n}\|_{(k_n(x))} - \eps}\\
		& = \log k_n(x) - \log (\|\p^{\otimes n}\|_{(k_n(x))} - \eps)\\
		& \leq n H(\p) + x \sqrt{n V(\p)} - \log \left(\Phi(x) - \delta - \eps\right)\\
		& = n H(\p) + x \sqrt{n V(\p)} - \log \left(\delta\right).
	\end{align}
	Note that $\Phi^{-1}(\cdot)$ is continuously differentiable and thus $x = \Phi^{-1}(\eps + 2\delta) = \Phi^{-1}(\eps) + O(\delta)$. Considering $\delta = 1/n$, we have
	\begin{align}
		\distill^{\eps}(\psi^{\otimes n}) \leq n H(\p) + \Phi^{-1}(\eps) \sqrt{n V(\p)} + o(\sqrt{n}).
	\end{align}
	
	Next, we prove the converse direction. For $n$ large enough, we have
	\begin{align}
		\distill^{\eps}(\psi^{\otimes n}) & = \min_{k\in\{\ell,\ldots,d\}} \log\left\lfloor\frac{k}{\|\p^{\otimes n}\|_{(k)}-\eps}\right\rfloor\\
		& \geq \min_{k\in\{\ell,\ldots,d\}} \left[\log\left(\frac{k}{\|\p^{\otimes n}\|_{(k)}-\eps}\right)\right] -1\\
		& \geq \min_{k\in\{\ell,\ldots,d\}} \left[\log\left(\frac{k}{1-\eps}\right)\right] -1\\
		& = \log f_{n,\eps}(\p) - \log (1-\eps) - 1\\
		& \geq  n H(\p) + \Phi^{-1}(\eps) \sqrt{n V(\p)} + o(\sqrt{n}),
	\end{align}
	where the first inequality follows from $\log \lfloor x \rfloor \geq (\log x) - 1$, the second from $\|\p\|_{(k)} \leq 1$, and the last from Lem.~\ref{lem: second-order 1}. This completes the proof.
\end{proof}

\begin{theorem}
Let $\rho\in\md(AB)$ and $m\in\mbb{N}$. Then,
\be
P^2\left(\Phi_m\xrightarrow{\locc} \rho^{AB} \right)=E_{(m)}\left(\rho^{AB}\right),
\ee
where $E_{(m)}\left(\rho^{AB}\right)$ is defined in Eq.~\eqref{kyfan} with $k=m$.
\end{theorem}

\begin{proof}
We begin with the special case that $\rho^{AB}=\psi^{AB}\in\pure(AB)$: Let $\p\in\prob^\da(n)$, with $n\eqdef|A|$, be the Schmidt coefficients corresponding to $\psi^{AB}$. From  Thm.~\ref{thm:EquivPurifDist}, it follows that
\begin{align}
P&\left(\Phi_m\to\psi^{AB}\right)=P_\star\left(\Phi_m\to\psi^{AB}\right)  \\
&=\min_{\r\in\prob(mn)}\Big\{P(\e_1^{(m)}\otimes\p,\r)\;:\;\r\succ\e_1^{(n)}\otimes\u^{(m)}\Big\}\;. \nonumber
\end{align}
where $\u^{(m)}$ is the uniform probability vector in $\prob(m)$. Now, observe that the condition $\r\succ\e_1^{(n)}\otimes\u^{(m)}$ holds iff $\r$ has at most $m$ non-zero components. Denoting by $\mathrm{nz}(\r)$ the maximal number of non-zero components in $\r$, and using the definition of the purified distance, we get
\begin{align}
	P^2\left(\Phi_m\to\psi^{AB}\right)&=1-\max_{\substack{\r\in\prob(mn),\\ \mathrm{nz}(\r)=m}}\left(\sum_{x=1}^{mn}\sqrt{r_xp_x}\right)^2 \nonumber\\
	&=1-\max_{\r\in\prob(m)}\left(\sum_{x=1}^{m}\sqrt{r_xp_x}\right)^2.
\end{align}
According to the Cauchy-Schwartz inequality, for any $\r\in\prob(m)$,
\begin{align}
	\sum_{x=1}^{m}\sqrt{r_xp_x}\le \sqrt{\sum_{x=1}^{m}r_x} \sqrt{\sum_{x=1}^{m}p_x}=\sqrt{\sum_{x=1}^{m}p_x},
\end{align}
with equality if $r_x=\frac{p_x}{\sum_{x=1}^{m}p_x}$ for all $x\in[m]$. This implies that 
\begin{align}\label{eq:CauchySchwartz}
	\max_{\r\in\prob(m)}\left(\sum_{x=1}^{m}\sqrt{r_xp_x}\right)^2=\sum_{x=1}^{m} p_x =\norm{\p}_{(m)}
\end{align}
and thus 
\ba\label{eq:PSquarePure}
P^2\left(\Phi_m\to\psi^{AB}\right)&=1-\norm{\p}_{(m)}=E_{(m)}\left(\psi^{AB}\right).
\ea

Now we consider the general case in which $\rho^{AB}$ is a mixed state: Let
\begin{align}\label{srank}
	\overline{\sr}\left(\rho^{AB}\right)\eqdef\inf &\left\{\sr\left(\psi^{A'B'}\right)\;:\;\psi^{A'B'}\xrightarrow{\locc} \rho^{AB},\right. \nonumber \\
	&\quad\left. \psi^{A'B'}\in\pure(A'B')\right\}\; 
\end{align}
be the maximal extension of the Schmidt rank to mixed states~\cite{Gour2020}. 
According to Nielson's theorem~\cite{Nielsen1999},
\begin{align}
	\Phi_k\xrightarrow{\locc} \psi^{A'B'}\in\pure(A'B')
\end{align} iff $\sr\left(\psi^{A'B'}\right)\le k$.
This implies that 
\begin{align}\label{eq:SchmidtNumber}
	\overline{\sr}\left(\rho^{AB}\right)=&\min \left\{k\;:\;\Phi_k\xrightarrow{\locc} \rho^{AB}\right\}\le |AB|\;.
\end{align}
Moreover, it also implies that $\overline{\sr}\left(\rho^{AB}\right)\le k$ whenever $\rho^{AB}$ has a pure-state decomposition that contains only states with Schmidt rank at most $k$. But since the Schmidt rank cannot increase under local measurements and classical communication~\cite{Lo2001}, $\overline{\sr}\left(\rho^{AB}\right)= k$ iff the following two conditions hold:
\ben
\item At least one of the states, in any pure-state decomposition of $\rho^{AB}$, has a Schmidt rank no smaller than $k$.
\item There exists a pure-state decomposition of $\rho^{AB}$ with all states having Schmidt rank at most $k$.
\een
In combination with the fact that $\overline{\sr}$ is monotonic under LOCC~\cite{Lo2001,Gour2020,Terhal2000}, this proves that
\ba\label{twosets}
&\Big\{\mE\left(\Phi_m\right)\;:\;\mE\in\locc(A'B'\to AB)\Big\}\\&=\Big\{\omega\in\md(AB)\;:\;\overline{\sr}\left(\omega^{AB}\right)\leq m\Big\}=\vcentcolon \mS_m\;.
\ea
This can also be seen from the fact that according to Eq.~\eqref{eq:SchmidtNumber}, $\overline{\sr}$ characterizes the zero-error entanglement cost, which is equal~\cite{Hayashi2006, Buscemi2011, Yue2019} to the Schmidt number introduced in Ref.~\cite{Terhal2000}. We, therefore, find that
\begin{align}\label{eq:EquivSets}
    &\sup_{\mE\in\locc(A'B'\to AB)}F^2\left(\mE\left(\Phi_m\right),\rho^{AB} \right) \nonumber \\
    &=\max_{\omega\in\md(AB),\;\overline{\sr}(\omega)\leq m}F^2\left(\omega^{AB},\rho^{AB} \right).
\end{align}

Expanding into, e.g., the generalized Gell-Mann matrices, every $\omega\in\md (AB)$ can be identified with a vector $\r\in\mbb{R^d}$, where $d=|AB|^2-1$. The set $\mS_m$ is then identified with the convex hull of the vectors corresponding to all pure states $\{\phi^{AB}\}$ with Schmidt rank no greater than $m$. From Carath\'{e}odory's theorem, it follows that any $\omega^{AB}\in\mS_m$ can be decomposed into at most $|AB|^2$ pure states with Schmidt rank at most $m$. 

Now let $E$ be a purifying system of dimension $n\ge|AB|^2$ with fixed orthonormal basis $\{|x^E\ra\}_{x\in[n]}$. This ensures that for any $\omega^{AB}\in\mS_m$, there exists a purification 
\begin{align}
    \ket{\phi^{ABE}}=\sum_{x=1}^n \sqrt{q_x} \ket{\phi_x^{AB}}\ket{x^E},
\end{align}
with the property that each $|\phi_x^{AB}\ra$ has a Schmidt rank no greater than $m$. Conversely, any pure state $\ket{\phi^{ABE}}$ which can be written in this form is a purification of an $\omega^{AB}\in\mS_m$.
According to Uhlmann's theorem~\cite{Uhlmann1976}, we thus find from Eq.~\eqref{eq:EquivSets} that
\begin{align}\label{13pp13}
    &\sup_{\mE\in\locc(A'B'\to AB)}F^2\left(\mE\left(\Phi_m\right),\rho^{AB} \right) \nonumber \\
    &=\max\left|\sum_{x\in[n]}\sqrt{q_x}\bra{\phi_x^{AB}}\bra{x^E}\ket{\psi^{ABE}}\right|^2,
\end{align}
where the maximum is over all purifications $\ket{\psi^{ABE}}$ of $\rho^{AB}$, all probability vectors $\q\in\prob(n)$, and all pure states $\{\phi_x^{AB}\}$ with Schmidt rank no greater than $m$. Next, we notice that we can expand every purification $\ket{\psi^{ABE}}$ of $\rho^{AB}$ as
\begin{align}
    \ket{\psi^{ABE}}=\sum_{x\in[n]}\sqrt{p_x}\ket{\psi_x^{AB}}\ket{x^E},
\end{align}
where $\{p_x,\psi_x^{AB}\}$ is a pure-state decomposition of $\rho^{AB}$. Conversely, all pure-state decompositions $\{p_x,\psi_x^{AB}\}$  of $\rho^{AB}$ into at most $n$ pure states correspond again to a purification of the above form. From this follows that
\begin{align}\label{eq:optDec}
    &\sup_{\mE\in\locc(A'B'\to AB)}F^2\left(\mE\left(\Phi_m\right),\rho^{AB} \right) \nonumber \\
    &=\max\left|\sum_{x\in[n]}\sqrt{q_x p_x}\bra{\phi_x^{AB}}\ket{\psi_x^{AB}}\right|^2,
\end{align}
where the maximum is over all pure-state decompositions $\{p_x,\psi_x^{AB}\}$  of $\rho^{AB}$ into at most $n$ pure states, all probability vectors $\q\in\prob(n)$, and all pure states $\{\phi_x^{AB}\}$ with Schmidt rank no greater than $m$. For $\psi\in\pure(AB)$, according to the definition in Eq.~\eqref{eq:PurConvDist} and~\eqref{twosets},
\begin{align}
    P^2&\left(\Phi_m \xrightarrow{\locc} \psi^{AB}\right) \nonumber\\
    =&\left(\inf_{\tau\in\md(AB)}\left\{ P\left(\tau, \psi\right) : \Phi_m\xrightarrow{\locc}\tau \right\}\right)^2\nonumber\\
    =&\inf_{\tau\in\md(AB)}\left\{ P^2\left(\tau, \psi\right) : \Phi_m\xrightarrow{\locc}\tau \right\} \nonumber \\
    =&\inf_{\tau\in\mS_m}\left\{ 1-\bra{\psi}\tau\ket{\psi}\right\}\nonumber\\
    =&1-\max_{\substack{\phi\in\pure(AB)\\\sr(\phi)\leq m}}\left|\la\phi|\psi\ra\right|^2,
\end{align}
where in the last line,  we used that the infimum will be attained on a pure state.
Together with Eq.~\eqref{eq:PSquarePure}, this implies that
\begin{align}
    \max_{\substack{\phi\in\pure(AB)\\\sr(\phi)\leq m}}\left|\la\phi|\psi\ra\right|^2=1-E_{(m)}\left(\psi^{AB}\right).
\end{align}
Since the Schmidt rank is independent of a global phase, there exists $\{\phi_x^{AB}\}_{x\in[n]}$ such that
\be
\la\phi_x^{AB}|\psi_x^{AB}\ra=\left|\la\phi_x^{AB}|\psi_x^{AB}\ra\right|=\sqrt{1-E_{(m)}\left(\psi_x^{AB}\right)}\;,
\ee
and that $\left|\la\phi_x^{AB}|\psi_x^{AB}\ra\right|$ cannot exceed this value. Inserting this into Eq.~\eqref{eq:optDec}, we deduce that 
\begin{align}
&\sup_{\mE\in\locc(A'B'\to AB)}F^2\left(\mE\left(\Phi_m\right),\rho^{AB} \right)\nonumber\\
&=\max_{\{p_x,\psi_x\}} \max_{\q\in\prob(n)}\left(\sum_{x\in[n]} \sqrt{q_x p_x\left(1-E_{(m)} \left(\psi_x^{AB}\right) \right)}\right)^2\nonumber\\
&=\max_{\{p_x,\psi_x\}} \sum_{x\in[n]} p_x\left(1-E_{(m)} \left(\psi_x^{AB}\right) \right),
\end{align}
where we used Eq.~\eqref{eq:CauchySchwartz} in the last line and the maximum over $\{p_x,\psi_x\}$ stands for a maximum over all pure state decompositions $\{p_x,\psi_x^{AB}\}_{x\in[n]}$ of $\rho^{AB}=\sum_{x\in[n]}p_x\psi_x^{AB}$. To conclude the proof, we note that 
\begin{align}
 P^2&\left(\Phi_m \xrightarrow{\locc} \rho^{AB}\right) \nonumber\\
&=\inf_{\mE\in\locc(A'B'\to AB)}P^2\left(\mE\left(\Phi_m\right),\rho^{AB} \right)\nonumber\\
&=\inf_{\mE\in\locc(A'B'\to AB)}1-F^2\left(\mE\left(\Phi_m\right),\rho^{AB} \right)\nonumber\\
&=1-\max_{\{p_x,\psi_x\}}\sum_{x\in[n]}p_x\left(1-E_{(m)} \left(\psi_x^{AB}\right) \right)\nonumber\\
&=\min_{\{p_x,\psi_x\}}\sum_{x\in[n]}p_xE_{(m)} \left(\psi_x^{AB}\right)
\end{align}
where the min and max above are still over all pure-state decompositions $\{p_x,\psi_x^{AB}\}_{x\in[n]}$ of $\rho^{AB}$. Remembering that we only required that $n\ge|AB|^2$, we can choose $n$ arbitrarily large, and thus the above equation holds for \textit{any} pure state decomposition. This completes the proof and shows that the infimum in Eq.~\eqref{kyfan} is attained and that it is sufficient to consider decompositions into at most $|AB|^2$ pure states.
\end{proof}

\begin{lemma}
	Let $\rho\in\md(XA)$ be a classical-quantum-state as in Eq.~\eqref{eq:cqState}, and for any $m\in[|A|]$, let $\rho^{(m)}$ be as defined in Eq.~\eqref{eq:mPrune}. Then, for any $\eps\in[0,1]$, 
	\begin{align}
		H_{\max}^{\eps}(A|X)_\rho=&\min_{m\in[|A|]}\left\{\log m: \frac{1}{2}\norm{\rho^{(m)}-\rho^{XA}}_1\le\eps\right\}\nonumber\\
		=&\min_{m\in[|A|]}\left\{\log m: \sum_{x\in[k]}p_x \left\|\rho^A_x\right\|_{(m)}\ge 1-\eps\right\}.
	\end{align}
\end{lemma}
\begin{proof}
	If $\rho^{(m)}\ne\rho^{XA}$, there exists an $x\in{[|A|]}$ such that $p_x\ne0$ and the rank of $\rho_x^{(m)}$ equals $m$, from which follows that (see Eq.~\eqref{1260})
	\begin{align}
		H_{\max}(A|X)_{\rho^{(m)}}=\log m.
	\end{align}
	Using the definition of $H_{\max}^{\eps}(A|X)_\rho$ (see Eq.~\eqref{eq:smoothCondMaxEnt}), we thus obtain
	\begin{align}
		H_{\max}^{\eps}(A|X)_\rho=&\min_{\omega\in\md(XA)}\left\{H_{\max}(A|X)_\omega: \frac{1}{2}\norm{\omega^{XA}-\rho^{XA}}_1\le \eps \right\} \nonumber \\
		\le& \min_{m\in[|A|]}\left\{H_{\max}(A|X)_{\rho^{(m)}}: \frac{1}{2}\norm{\rho^{(m)}-\rho^{XA}}_1\le \eps \right\} \nonumber \\
		=&\min_{m\in[|A|]}\left\{\log m: \frac{1}{2}\norm{\rho^{(m)}-\rho^{XA}}_1\le \eps \right\}.
	\end{align}

	To show that indeed we have an equality in the above equation, we notice that by definition, $\rho^{(m)}_x$ and $\rho_x^A$ commute.
	Denoting the eigenvalues of $\rho_x^A$ by $\{\lambda_{y|x}\}$ we get that the only potentially non-zero eigenvalues of $\rho_x^{(m)}$ are 
	\begin{align}
		\left\{ \lambda_{y|x}^\da/\left\|\rho^A_x\right\|_{(m)} \right\}_{y\in[m]}.
	\end{align}
	From this follows that
	\begin{align}\label{12060}
		\frac12&\left\|\rho_x^A-\rho_x^{(m)}\right\|_1\nonumber\\
		=& \frac12\left[ \sum_{y=1}^m \left| \lambda_{y|x}^\da - \lambda_{y|x}^\da/\left\|\rho^A_x\right\|_{(m)}\right|+\sum_{y=m+1}^{|A|}\lambda_{y|x}^\da \right]\nonumber \\
		=& \frac12 \left[ \left\|\rho^A_x\right\|_{(m)}\left(1/\left\|\rho^A_x\right\|_{(m)}-1\right) +1-\left\|\rho^A_x\right\|_{(m)} \right] \nonumber \\
		=&1-\left\|\rho^A_x\right\|_{(m)}
	\end{align}
	and thus
	\begin{align}
		\frac{1}{2}\norm{\rho^{(m)}-\rho^{XA}}_1=&1-\sum_{x\in[k]}p_x \left\|\rho^A_x\right\|_{(m)}.
	\end{align}
	This proves that the two optimization problems in the Lemma have the same solution, i.e., 
	\begin{align}\label{eq:altExpLemma}
		\min_{m\in[|A|]}\left\{\log m: \frac{1}{2}\norm{\rho^{(m)}-\rho^{XA}}_1\le\eps\right\}
		=&\min_{m\in[|A|]}\left\{\log m: \sum_{x\in[k]}p_x \left\|\rho^A_x\right\|_{(m)}\ge 1-\eps\right\}.
	\end{align}

	Now suppose by contradiction that 
	\begin{align}
		H_{\max}^{\eps}(A|X)_\rho<\min_{m\in[|A|]}\left\{\log m: \sum_{x\in[k]}p_x \left\|\rho^A_x\right\|_{(m)}\ge 1-\eps\right\}
	\end{align}
	and let $m$ be the smallest natural number such that $\sum_xp_x\left\|\rho^A_x\right\|_{(m)}\geq 1-\eps$. Further, let 
	\begin{align}
		\omega^{XA}=\sum_x q_x\ketbra{x}{x}^X\otimes\omega_x^A \in\mb_{\eps}\left(\rho^{XA}\right)
	\end{align} be such that $H_{\max}^\eps(A|X)_\rho=H_{\max}(A|X)_\omega$. 
	Then, according to the assumption,
	\be\label{cgf}
	\log m>H_{\max}(A|X)_\omega=\log\max_{x\in[k]:q_x\ne0}\tr\left[\Pi_{\omega_x}^A\right]
	\ee
	so that $\tr\left[\Pi_{\omega_x}^A\right]<m$ for all $x \in[k]:q_x\ne0$. In particular, this means that for any $x\in[k]:q_x\ne0$, $\left\|\rho^A_x\right\|_{(m-1)}\geq \tr\left[\rho^A_x\Pi_{\omega_x}^A\right]$ since $\Pi_{\omega_x}^A$ has a rank strictly smaller than $m$. Therefore, with 
	\begin{align}\label{eq:ProjOmega}
		\Pi^{XA}_\omega\eqdef\sum_{x\in[k]:q_x\ne0}\ketbra{x}{x}^X\otimes\Pi_{\omega_x}^A,
	\end{align} 
	we have 
	\ba
	\sum_x&p_x\left\|\rho^A_x\right\|_{(m-1)}\\
	&\geq  \tr\left[\rho^{XA}\Pi^{XA}_\omega\right]\\
	&=\tr\left[\omega^{XA}\Pi^{XA}_\omega\right]+\tr\left[\left(\rho^{XA}-\omega^{XA}\right)\Pi^{XA}_\omega\right]\\
	&\geq 1-\tr\left[\left(\rho^{XA}-\omega^{XA}\right)_-\Pi^{XA}_\omega\right]\\
	&\geq 1-\eps\;,
	\ea
	where we used the fact that $\Pi^{XA}_\omega\leq I^{XA}$ and $\tr\left(\rho^{XA}-\omega^{XA}\right)_-=\frac12\|\omega^{XA}-\rho^{XA}\|_1\leq\eps$. This is in contradiction to the definition of $m$ as the smallest natural number such that $\sum_xp_x\left\|\rho^A_x\right\|_{(m)}\geq 1-\eps$, which completes the proof.

\end{proof}

\begin{theorem}
	For $\rho\in\md(AB)$, the $\eps$-single-shot entanglement cost is given by
	\be
	\cost^\eps\left(\rho^{AB}\right)= \inf_{\rho^{XAB}} H_{\max}^{\eps}(A|X)_\rho\;,
	\ee
	where the infimum is over all classical systems $X$ and all classical extensions $\rho^{XAB}$ of $\rho^{AB}$. Moreover, the infimum is attained for a classical extension with $|X|=|AB|^2$ and can also be taken over all regular extensions of $\rho^{AB}$. 
\end{theorem}
\begin{proof}
We begin by showing a useful mathematical identity. For $\eps\in[0,1]$, let
\begin{align}
	I\eqdef&\min_{m\in[|A|]} \left\{\log m: \max_{\rho^{XAB}} \sum_xp_x\left\|\rho^A_x\right\|_{(m)}\ge 1-\eps\right\}, \nonumber\\
	J\eqdef&\min_{\rho^{XAB}} \min_{m\in[|A|]} \left\{\log m:  \sum_xp_x\left\|\rho^A_x\right\|_{(m)}\ge 1-\eps\right\},
\end{align}
where the optimizations over $\rho^{XAB}$ run over all regular extensions  $\rho^{XAB}$ of $\rho^{AB}$ with $|X|=|AB|^2$. Using the same convention, let 
\begin{align}
	\mathfrak{Y}\eqdef&\left\{\log m: m\in[|A|], \max_{\rho^{XAB}} \sum_xp_x\left\|\rho^A_x\right\|_{(m)}\ge 1-\eps\right\}, \nonumber \\
	\mathfrak{Y}_\rho\eqdef&\left\{\log m: m\in[|A|],  \sum_xp_x\left\|\rho^A_x\right\|_{(m)}\ge 1-\eps\right\}. \nonumber \\
\end{align}
Since $ \max_{\rho^{XAB}} \sum_xp_x\left\|\rho^A_x\right\|_{(m)}\ge\sum_xp_x\left\|\rho^A_x\right\|_{(m)}\ge 1-\eps$, it follows that $\mathfrak{Y}_\rho\subset\mathfrak{Y}$  and thus $\min \mathfrak{Y}\le \min \mathfrak{Y}_\rho$ for all regular extensions $\rho^{XAB}$. We therefore find that
\begin{align}
	I=\min \mathfrak{Y}\le \min_{\rho^{XAB}}\min \mathfrak{Y}_\rho=J.
\end{align}
From the definition of $I=\vcentcolon\log m^\star$, it further follows that there exists a regular extension $\tilde{\rho}^{XAB}$ of $\rho^{AB}$ such that
\begin{align}
	\sum_x\tilde{p}_x\left\|\tilde{\rho}^A_x\right\|_{(m^\star)}\ge 1-\eps.
\end{align}
Now assume that $I<J$, i.e.,
\begin{align}
	\log m^\star < \min_{\rho^{XAB}} \min_{m\in[|A|]} \left\{\log m:  \sum_xp_x\left\|\rho^A_x\right\|_{(m)}\ge 1-\eps\right\} \le  \min_{m\in[|A|]} \left\{\log m:  \sum_x\tilde{p}_x\left\|\tilde{\rho}^A_x\right\|_{(m)}\ge 1-\eps\right\}\le \log m^\star, 
\end{align}
which is a contradiction. We thus showed that $I=J$.
	
Next, we notice that from Thm.~\ref{lem1291} and its proof, it follows that 
\be
P^2\left(\Phi_m\xrightarrow{\locc} \rho^{AB} \right)=E_{(m)}\left(\rho^{AB}\right)=\min_{\rho^{XAB}}\left(1-\sum_xp_x\left\|\rho^A_x\right\|_{(m)}\right)=1-\max_{\rho^{XAB}}\sum_xp_x\left\|\rho^A_x\right\|_{(m)},
\ee
where the optimizations are over all classical systems $X$ and all regular extensions $\rho^{XAB}$ of $\rho^{AB}$, with the optimal values attained for $|X|=|AB|^2$. The $\eps$-single-shot entanglement cost as defined in Eq.~\eqref{econe} can thus be expressed as
\begin{align}\label{13p253}
    \cost^\eps\left(\rho^{AB}\right)=&\min_{m\in[|A|]} \Big\{\log m:\;\max_{\rho^{XAB}} \sum_xp_x\left\|\rho^A_x\right\|_{(m)}\geq 1-\eps\Big\}\nonumber \\ 
    =&\min_{\rho^{XAB}} \min_{m\in[|A|]}\Big\{\log m:\;\sum_xp_x\left\|\rho^A_x\right\|_{(m)}\geq 1-\eps\Big\}\nonumber \\
    =&\min_{\rho^{XAB}} H_{\max}^{\eps}(A|X)_\rho,
\end{align}
where we used Lem.~\ref{lem:HMaxCQ} in the last line. 

To conclude the proof, we show that the infimum can also be taken over all classical extensions. Since the set of regular extensions $\rho^{XAB}$ is a subset of the set of classical extensions $\tilde{\rho}^{XAB}$,
	\begin{align}
		\inf_{\rho^{XAB}} H_{\max}^{\eps}(A|X)_\rho \ge \inf_{\tilde{\rho}^{XAB}} H_{\max}^{\eps}(A|X)_{\tilde{\rho}}.
	\end{align}
	The proof of the converse is very similar to the proof of Lem.~\ref{lem:HMaxCQ}: Let $\rho^{XAB}$ be a \textit{regular} extension of $\rho^{AB}$ satisfying $H_{\max}^{\eps}(A|X)_\rho=\inf_{\rho^{XAB}} H_{\max}^{\eps}(A|X)_\rho$ (which exists according to the previous parts of the proof) and let $m$ be the smallest integer such that 
	$\sum_{x\in[|A|]}p_x \left\|\rho^A_x\right\|_{(m)}\ge 1-\eps$.
	Now assume by contradiction that there exists a \textit{classical} extension $\tilde{\rho}^{XAB}$ of $\rho^{AB}$ such that $H_{\max}^{\eps}(A|X)_{\tilde{\rho}}<H_{\max}^{\eps}(A|X)_\rho=\log m$ (see Lem.~\ref{lem:HMaxCQ}).  Denoting by
	\begin{align}
		\omega^{XA}=\sum_x q_x\ketbra{x}{x}^X\otimes\omega_x^A \in\mb_{\eps}\left(\tilde{\rho}^{XA}\right)
	\end{align} a state satisfying $H_{\max}^\eps(A|X)_{\tilde{\rho}}=H_{\max}(A|X)_\omega$, we find that 
	\begin{align}
		\log m > H_{\max}(A|X)_\omega=\log \max_{x\in[|A|]:q_x\ne0}\tr\left[\Pi_{\omega_x}^A\right]
	\end{align}
	and thus $ m > \tr\left[ \Pi_{\omega_x}^A\right] \forall x\in[|A|]:q_x\ne 0$. With $\Pi_{\omega}^{XA}$ defined as in Eq.~\eqref{eq:ProjOmega}, this implies again that 
	\begin{align}
		\sum_x&p_x\left\|\rho^A_x\right\|_{(m-1)} \nonumber\\
		&\geq  \tr\left[\rho^{XA}\Pi^{XA}_\omega\right]\nonumber\\
		&=\tr\left[\omega^{XA}\Pi^{XA}_\omega\right]+\tr\left[\left(\rho^{XA}-\omega^{XA}\right)\Pi^{XA}_\omega\right]\nonumber\\
		&\geq 1-\tr\left[\left(\rho^{XA}-\omega^{XA}\right)_-\Pi^{XA}_\omega\right]\nonumber\\
		&\geq 1-\eps\;,
	\end{align}
	which is a contradiction to the assumption.
\end{proof}

The proof of Cor.~\ref{cor:pureCostHmax} was provided in the main text.
\setcounter{theorems}{12}
\begin{lemma}
	Let $F\left(\Phi_m\xrightarrow{\locc}\mN^{A\to B}\right)$ be defined as in Eq.~\eqref{eq:ConvFidChan}. It holds that
	\begin{align}
		&F\left(\Phi_m\xrightarrow{\locc}\mN^{A\to B}\right) = \min_{\psi \in \pure(A\tilde A)} \sup_{\Theta}  F\left(\Theta[\Phi_m](\psi^{A\tilde A}), \mN^{\tilde A \to B}(\psi^{A\tilde A})\right) \nonumber
	\end{align}    
	where the supremum is again over all LOCC super-channels $\Theta$ that map the state $\Phi_m$ to a channel in $\cptp(\tA\to B)$.
\end{lemma}
\begin{proof}
	We begin by noting that 
	\begin{align}
		\tilde D_{1/2}(\rho\|\sigma)= -2\log F(\rho,\sigma),
	\end{align}
	where $\tilde D_{1/2}$ denotes the sandwiched Rényi relative entropy of order 1/2~\cite{Wilde2014,Mueller2013}. The data-processing inequality of $\tilde D_{1/2}$ follows directly from the analogous property of the fidelity and that $\tilde D_{1/2}$ obeys the direct-sum property, i.e., 
	\begin{align}
		\tilde D_{1/2}\left(\sum_i p_i \ketbra{i}{i}\otimes \rho_i \Big\|\sum_i p_i \ketbra{i}{i}\otimes \sigma_i\right) & = \sum_i p_i \tilde D_{1/2}(\rho_i\|\sigma_i)
	\end{align}
	too. According to (the proofs of) Ref.~\cite[Prop.~8]{wang2019converse} and~\cite[Lem.~II.3]{leditzky2018approaches}, it therefore holds that $$\tilde D_{1/2}(\mN^{A'\to B}(\phi_\rho^{AA'})\|\mM^{A'\to B}(\phi_\rho^{AA'}))$$ is concave in $\rho^{A'}$, where $\phi_\rho^{AA'}$ is a purification of $\rho^{A'}$ and $\mM^{A'\to B}$, $\mN^{A'\to B}$ are quantum channels. In short, their argument is the following: For any convex combination $\rho^{A'} = \sum_i p_i \rho^{A'}_i$, suppose $\rho^{A'}_i$ has a purification $\ket{\phi_i^{AA'}}$. Then
	\begin{align}
		\ket{\psi^{PAA'}} = \sum_i \sqrt{p_i} \ket{i^P} \otimes \ket{\phi_i^{AA'}}
	\end{align}
	is a purification of the average state $\rho^{A'}$. Since all purifications are related by an isometry, there exists an isometric channel $\mW^{A \to PA}$ such that $\mW^{A\to PA}(\phi_\rho^{AA'}) = \psi^{PAA'}$.
	Then we have
	\begin{align}
		\tilde D_{1/2} & \left(\mN^{A'\to B}(\phi_\rho^{AA'})\Big\|\mM^{A'\to B}(\phi_\rho^{AA'})\right)\\
		& = \tilde D_{1/2} \left(\mW^{A\to PA}\mN^{A'\to B}(\phi_\rho^{AA'})\Big\|\mW^{A\to PA}\mM^{A'\to B}(\phi_\rho^{AA'})\right)\\
		& = \tilde D_{1/2} \left(\mN^{A'\to B}\mW^{A\to PA}(\phi_\rho^{AA'})\Big\|\mM^{A'\to B}\mW^{A\to PA}(\phi_\rho^{AA'})\right)\\
		& =    \tilde D_{1/2} \left(\mN^{A'\to B}(\psi^{PAA'})\Big\|\mM^{A'\to B}(\psi^{PAA'})\right)\\
		& \geq \tilde D_{1/2}\left(\sum_i p_i \ketbra{i}{i}^P \otimes \mN^{A'\to B}(\phi_i^{AA'})\Big\|\sum_i p_i \ketbra{i}{i}^P \otimes \mM^{A'\to B}(\phi_i^{AA'})\right)\\
		& = \sum_i p_i \tilde D_{1/2}\left(\mN^{A'\to B}(\phi_i^{AA'})\|\mM^{A'\to B}(\phi_i^{AA'})\right),
	\end{align}
	where the first equality follows from the isometric invariance of the divergence, the second equality follows as $\mW^{A\to PA}$ commutes with $\mN^{A'\to B}$ and $\mM^{A'\to B}$, the inequality follows from the data-processing inequality of $\tilde D_{1/2}$ under the dephasing channel $\sum_i \ketbra{i}{i}^P \cdot \ketbra{i}{i}^P$, and the last equality follows from the direct-sum property of $\tilde D_{1/2}$. 
	
	We thus find that
	\begin{align}
		& \inf_{\Theta \in \locc} \max_{\psi \in \pure(A\tilde A)} \tilde D_{1/2}\left(\Theta[\Phi_m](\psi^{A\tilde A})\Big\| \mN^{\tilde A \to B}(\psi^{A\tilde A})\right) \\
		& = \inf_{\Theta \in \locc} \max_{\rho^A \in \md(A)} \tilde D_{1/2}\left(\Theta[\Phi_m](\phi_\rho^{A\tilde A})\Big\| \mN^{\tilde A \to B}(\phi_\rho^{A\tilde A})\right)\\
		& = \max_{\rho^A \in \md(A)} \inf_{\Theta \in \locc}  \tilde D_{1/2}\left(\Theta[\Phi_m](\phi_\rho^{A\tilde A})\Big\| \mN^{\tilde A \to B}(\phi_\rho^{A\tilde A})\right)\\
		& = \max_{\psi \in \pure(A\tilde A)} \inf_{\Theta \in \locc}  \tilde D_{1/2}\left(\Theta[\Phi_m](\psi^{A\tilde A})\Big\| \mN^{\tilde A \to B}(\psi^{A\tilde A})\right),
	\end{align}
	where the first and last equalities follow because the objective function is invariant with respect to the purification. The second equality follows from the above argument showing that the objective function is concave in $\rho^A$, the fact that the objective function is convex in $\Theta$, and that we are thus allowed to apply Sion's minimax theorem (clearly $\md(A)$ is compact and both sets over which we optimize are convex)~\cite[Cor.~3.3]{Sion1958}.
	Remembering that $\tilde D_{1/2}(\rho\|\sigma) = -2 \log F(\rho,\sigma)$ and thus, e.g.,
	\begin{align}
		& \inf_{\Theta \in \locc} \max_{\psi \in \pure(A\tilde A)} \tilde D_{1/2}\left(\Theta[\Phi_m](\psi^{A\tilde A})\Big\| \mN^{\tilde A \to B}(\psi^{A\tilde A})\right) \\
		&= \inf_{\Theta \in \locc} \max_{\psi \in \pure(A\tilde A)} -2 \log F\left(\Theta[\Phi_m](\psi^{A\tilde A}), \mN^{\tilde A \to B}(\psi^{A\tilde A})\right) \\
		&= \inf_{\Theta \in \locc} \max_{\psi \in \pure(A\tilde A)} -2 \log F\left(\Theta[\Phi_m](\psi^{A\tilde A}), \mN^{\tilde A \to B}(\psi^{A\tilde A})\right) \\
		&=-2 \log \sup_{\Theta \in \locc} \min_{\psi \in \pure(A\tilde A)} F\left(\Theta[\Phi_m](\psi^{A\tilde A}), \mN^{\tilde A \to B}(\psi^{A\tilde A})\right) \\
		&=-2 \log F(\Phi_m \to \mN^{A\to B})
	\end{align}
	completes the proof.
\end{proof}

\begin{lemma}
	Let $\mN\in\cptp(A\to B)$ be a quantum channel. It then holds that
	\begin{align}
		1-E_{(m)}\left(\mN^{A\to B}\right)	\leq F\left(\Phi_m\xrightarrow{\locc}\mN^{A\to B}\right)
		\leq \sqrt{1-E_{(m)}\left(\mN^{A\to B}\right)}\;.
	\end{align}
\end{lemma}
\begin{proof}
	We begin by noticing that from Eq.~\eqref{2p51}, it follows that every channel $\mM\in\cptp(A\to B)$ that is obtained from an LOCC super-channel $\Theta$ via  $\mM^{A\to B}=\Theta\left[\Phi_m^{A'B'}\right]$ has an operator sum representation consisting of Kraus operators whose rank do not exceed $m$. Let 
	\be
	\mM^{A\to B}(\cdot)=\sum_{y\in[k]}M_y(\cdot)M_y^*
	\ee
	be such an operator sum representation, i.e., $\rank(M_y)\leq m$. Further, observe that the pure state
	\be
	|\chi^{ABR}\ra\eqdef \sum_{y\in[k]}M_{y}^{\tA\to B}|\psi^{A\tA}\ra|y^R\ra
	\ee
	is a purification of the density matrix $\mM^{\tA\to B}(\psi^{A\tA})$. Similarly, any Kraus decomposition $\{N_z\}$ of $\mN^{\tA\to B}$ defines a purification of $\mN^{\tA\to B}(\psi^{A\tA})$ via
	\be
	|\varphi^{ABR'}\ra\eqdef\sum_{z\in[n]}N_{z}^{\tA\to B}|\psi^{A\tA}\ra|z^{R'}\ra\;.
	\ee
	By padding the shorter list of Kraus operators with zeros, w.l.o.g, we can choose $R'=R$.
	Hence, from Uhlmann's theorem~\cite{Uhlmann1976}, we get
	\ba
	F\left(\mM^{\tA\to B}(\psi^{A\tA}),\mN^{\tA\to B}(\psi^{A\tA})\right)=\max_U\big|\la\chi|U|\varphi\ra\big|,
	\ea
	where the maximum is over all unitaries $U$ acting on $R$. Now, observe that
	\ba
	U^{R\to R}|\varphi^{ABR}\ra&=\sum_{z\in[n]}N_{z}^{\tA\to B}|\psi^{A\tA}\ra U^{R\to R}|z^{R}\ra\\
	&=\sum_{y\in[k]}L_y^{\tA\to B}|\psi^{A\tA}\ra|y^R\ra\;,
	\ea
	where
	\be
	L_{y}^{\tA\to B}\eqdef\sum_{z\in[n]}\la y|U|z\ra N_z^{\tA\to B}\;.
	\ee
	Since $U$ is unitary, also $\{L_y\}_{y\in[k]}$ forms an operator-sum representation of $\mN$. With this notation
	\be
	\la\chi|U|\varphi\ra=\sum_{y\in[k]}\la\psi^{A\tA}|M_y^*L_y|\psi^{A\tA}\ra.
	\ee
	Since for each $y$ we have $\rank(M_y)\leq m$,  there exists a projection $\Pi_y$ acting on $B$ such that $\tr[\Pi_y]=m$ and $M_y=\Pi_yM_y$. Let 
	\be
	|\phi^{ABR}\ra\eqdef\sum_{y\in[k]}\Pi_y^B L_y^{\tA\to B}|\psi^{A\tA}\ra|y^{R}\ra\;
	\ee
	and observe that $|\phi\ra$ is in general not normalized.
	This implies that 
	\ba
	|\la\chi|U|\varphi\ra|&=|\la\chi|\phi\ra|
	\leq \||\chi\ra\|\||\phi\ra\|=\||\phi\ra\|\; ,
	\ea
	where we used the Cauchy-Schwarz inequality and the fact that $|\chi\ra$ is normalized.
	Now, observe that
	\ba
	\||\phi\ra\|^2&=\sum_{y}\la\psi^{A\tA}|L_{y}^{*}\Pi_{y}L_{y}|\psi^{A\tA}\ra\\
	\Gg{\rho^{\tA}\eqdef\tr_{A}\left[\psi^{A\tA}\right]\to}&\leq\sum_{y}\left\|L_{y}\rho^{\tA}L_{y}^{*}\right\|_{(m)}\;.
	\ea
	Hence,
	\ba\label{eq:BoundChanFid}
	F^2\left(\Phi_m^{A'B'}\xrightarrow{\locc}\mN^{A\to B}\right)&\leq \min_{\rho\in\md(\tA)}\sup_{\{L_y\}}\sum_{y}\left\|L_{y}\rho^{\tA}L_{y}^{*}\right\|_{(m)},
	\ea
    where the supremum is over all operator-sum representations of $\mN$.
	
	Next, we will show that the right-hand side of the above equation can be expressed in terms of $E_{(m)}$. 
	To this end, for a fixed $\psi^{A\tA}\in\pure(A\tA)$, let $$\rho^{AB}=\mN^{\tA\to B} \left(\psi^{A\tA}\right).$$ Now every unnormalized pure state decomposition of $\rho^{AB}=\sum_x \xi_x^{AB}$ corresponds to an operator-sum representation $N_x$ of $\mN$ in the sense that $\xi_x^{AB}=\mN_x^{\tA\to B}\left(\psi^{A\tA}\right)\eqdef N_x\psi^{A\tA}N_x^{*}$: That every operator-sum representation defines a pure state decomposition in this way is obvious. Now let $\{\xi_x^{AB}\}$ be a pure state decomposition corresponding to any operator-sum representation $\{N_x\}$ and let $\{\chi_y^{AB}\}$ be an arbitrary pure state decomposition. According to Ref.~\cite{Hughston1993}, this implies that there exists a unitary $U$ such that $\ket{\chi_y^{AB}}=\sum_x U_{yx}\ket{\xi_x^{AB}}$. Due to the unitary freedom in operator-sum representations, $\{M_y=\sum_x U_{yx} N_x\}$ is an operator-sum representation of $\mN$ too, and $\chi_y^{AB}=\mM_y^{\tA\to B}\left(\psi^{A\tA}\right)$. This implies that 
	\begin{align}
		&\inf_{\{N_x\}}\left(1-\sum_{x\in[k]}\left\|\tr_A\left[\mN_x^{\tA\to B} \left(\psi^{A\tA}\right)\right] \right\|_{(m)}\right)\nonumber \\
		=& \inf \left(1-\sum_{x\in[k]} p_x\left\|\tr_A\left[\xi_x^{AB}\right] \right\|_{(m)}\right),
	\end{align}
	where the second infimum is over all normalized pure state decompositions $\{p_x, \xi_x^{AB}\}$ of $\rho^{AB}=\mN^{\tA\to B} \left(\psi^{A\tA}\right)=\sum_x p_i \xi_x^{AB}$. 
	
	Moreover, remember that according to Eq.~\eqref{kyfan}, 
	\begin{align}
		E_{(m)}\left(\rho^{AB}\right)=& \inf \left(1-\sum_x p_x\left\|\tr_A\left[\psi^{AB}_x\right]\right\|_{(m)}\right),
	\end{align}
	where the infimum is over all pure-state decompositions $\rho^{AB}=\sum_x p_x \psi^{AB}_x$. In combination, this shows that 
	\begin{align}\label{eq:KrausDec}
		E_{(m)}\left(\mN^{A\to B}\right)=&\max_{\psi\in\pure(A\tA)} E_{(m)}\left(\mN^{\tA\to B}\left(\psi^{A\tA}\right)\right)\nonumber\\
		=&\max_{\psi\in\pure(A\tA)}\inf_{\{N_x\}}\left(1-\sum_{x\in[k]}\left\|\tr_A\left[\mN_x^{\tA\to B} \left(\psi^{A\tA}\right)\right] \right\|_{(m)}\right) \nonumber \\
		=&\max_{\rho\in \md(\tA)} \inf_{\{N_x\}}\left(1-\sum_{x\in[k]} \left\|\mN_x^{\tA\to B} \left(\rho^{\tA}\right) \right\|_{(m)}\right) \nonumber \\
		=&1-\min_{\rho\in \md(\tA)} \sup_{\{N_x\}} \sum_{x\in[k]} \left\|\mN_x^{\tA\to B} \left(\rho^{\tA}\right) \right\|_{(m)}.
	\end{align}
	Inserting this into Eq.~\eqref{eq:BoundChanFid}, we find that 
	\begin{align}\label{eq:BoundChanFidUp}
				F^2\left(\Phi_m\xrightarrow{\locc}\mN^{A\to B}\right)&\leq 1-E_{(m)}\left(\mN^{A\to B}\right),
	\end{align}
	which finishes the first part of the proof.

	To complete the proof, we must show that $F\left(\Phi_m\xrightarrow{\locc}\mN^{A\to B}\right)\geq 1-E_{(m)}\left(\mN^{A\to B}\right)$.
	To this end, let \be\mN^{\tA\to B}(\cdot)=\sum_{x\in[n+1]}N_x(\cdot)N_x^*\ee be an operator sum representation of $\mN^{\tA\to B}$, where we assume for $x=n+1$ that $\mN_{n+1}^{A\to B}=0$.  Consider a super-channel $\Theta$ of the form given in Eq.~\eqref{2p51} where we choose $k={n+1}$, $\mF_{(x)}^{B'\to B}(\cdot)=V_x(\cdot)V_x^*$ to be isometries, and $\mE_x^{\tA\to B'}(\cdot)= M_x(\cdot)M_x^*$ with
	\be
	M_x\eqdef \begin{cases} V_x^*N_x & \text{for }x\in[n],\\
		\sqrt{I^B-\sum_{x\in[n]}N_x^*P_xN_x} & \text{for }x=n+1,
	\end{cases}
	\ee
	where $P_x\eqdef V_xV_x^*\in\pos(B)$ is a projection to an $m$-dimensional subspace. Since
	\be
	\sum_{x\in[n]}N_x^*P_xN_x\leq \sum_{x\in[n]}N_x^*N_x= I^B\;,
	\ee
	this ensures that $\{M_x\}_{x\in[n+1]}$ is a valid instrument.  Moreover, let $\psi\in \pure(A\tA)$ be fixed but arbitrary and $\rho^{\tA}=\tr_{A}\left[\psi^{A\tA}\right]$.
	Remembering that $F(P_0+P_1,Q_0+Q_1)\ge F(P_0,Q_0)+F(P_1,Q_1)$ whenever $P_0,P_1,Q_0,Q_1\ge 0$, see, e.g., Ref.~\cite[Thm. 3.25]{Watrous2018}, we get that
	\begin{align} \label{eq:lowBoundChanFid}
		&F\left(\sum_{x\in[n+1]}\mF_{(x)}^{B'\to B}\circ\mE_x^{\tA\to B'}(\psi^{A\tA}),\mN^{\tA\to B}(\psi^{A\tA})\right)\nonumber\\
		&\geq\sum_{x\in[n]}F\left(\mF_{(x)}^{B'\to B}\circ\mE_x^{\tA\to B'}(\psi^{A\tA}),\mN^{\tA\to B}_x(\psi^{A\tA})\right)\nonumber\\
		&=\sum_{x\in[n]}\sqrt{\tr\left[\mN^{\tA\to B}_x(\psi^{A\tA})\mF_{(x)}^{B'\to B}\circ\mE_x^{\tA\to B'}(\psi^{A\tA})\right]}\nonumber\\
		&=\sum_{x\in[n]}\left|\tr\left[\rho^{\tA}N_x^*V_xM_x\right]\right|\nonumber\\
		&=\sum_{x\in[n]}\tr\left[\rho^{\tA}N_x^*P_xN_x\right]\;,
	\end{align} 
	where we used in the third line that $\mN_x^{\tA\to B}(\psi^{A\tA})$ is pure. 
	Next, we recall that for a positive semidefinite operator $A$, $\left\|A\right\|_{(m)}$ denotes the sum of its $m$ largest eigenvalues. Taking in Eq.~\eqref{eq:lowBoundChanFid} $P_x$ to be the projection to the $m$-dimensional eigen-subspace corresponding to the $m$ largest eigenvalues of $N_x\rho^AN_x^*$ and utilizing Lem.~\ref{lem:MinMaxFid} thus gives 
	\ba\label{cfda}
	F\left(\Phi_m\xrightarrow{\locc}\mN^{A\to B}\right)&\geq 
	\min_{\rho\in\md(\tA)}\sup_{\{N_x\}}\sum_{x\in[k]}\left\|N_x\rho^{\tA} N_x^*\right\|_{(m)}\\
	&=1-\max_{\rho\in\md(\tA)}\inf_{\{N_x\}}\left(1-\sum_{x\in[k]}\left\|N_x\rho^{\tA} N_x^*\right\|_{(m)}\right),
	\ea
	where the optimizations involving $\{N_x\}$ are over all operator-sum representations of $\mN$. Invoking again Eq.~\eqref{eq:KrausDec} finishes the proof.

	\end{proof}

\begin{theorem}
	Let $\mN\in\cptp(A\to B)$ be a quantum channel and $\eps\in[0,1)$. Then
	\begin{align}
		&\max_{\psi\in\pure(A\tA)}\inf_{\sigma^{XAB}} H_{\max}^{\eps}(A|X)_\sigma 	\leq \cost^\eps(\mN^{A\to B})\leq\max_{\psi\in\pure(A\tA)}\inf_{\sigma^{XAB}} H_{\max}^{\eps/2}(A|X)_\sigma,
	\end{align}
	where the infimums are over all classical systems $X$ and all classical extensions $\sigma^{XAB}$ of $\sigma^{AB}=\mN^{\tA\to B}(\psi^{A\tA})$. Again, the infimums are attained for a regular/classical extension with $|X|=|AB|^2$.
\end{theorem}

\begin{proof}
	Notice that for every $\psi\in\pure(A\tA)$, it holds that 
	\begin{align}
		\inf_{m\in\mbb{N}}\Big\{\log m\;:\;E_{(m)}\left(\mN^{\tA\to B}\left(\psi^{A\tA}\right)\right)\leq\eps\Big\} \le \inf_{m\in\mbb{N}}\Big\{\log m\;:\;\max_{\psi\in\pure(A\tA)}E_{(m)}\left(\mN^{\tA\to B}\left(\psi^{A\tA}\right)\right)\leq\eps\Big\},
	\end{align}
	which implies that 
	\begin{align}
		\max_{\psi\in\pure(A\tA)} \inf_{m\in\mbb{N}}\Big\{\log m\;:\;E_{(m)}\left(\mN^{\tA\to B}\left(\psi^{A\tA}\right)\right)\leq\eps\Big\} \le \inf_{m\in\mbb{N}}\Big\{\log m\;:\;\max_{\psi\in\pure(A\tA)}E_{(m)}\left(\mN^{\tA\to B}\left(\psi^{A\tA}\right)\right)\leq\eps\Big\}. \nonumber
	\end{align}
	On the contrary, assuming that 
	\begin{align}
		\max_{\psi\in\pure(A\tA)} \inf_{m\in\mbb{N}}\Big\{\log m\;:\;E_{(m)}\left(\mN^{\tA\to B}\left(\psi^{A\tA}\right)\right)\leq\eps\Big\} < \inf_{m\in\mbb{N}}\Big\{\log m\;:\;\max_{\psi\in\pure(A\tA)}E_{(m)}\left(\mN^{\tA\to B}\left(\psi^{A\tA}\right)\right)\leq\eps\Big\} \nonumber
	\end{align}
	leads to a contradiction: Let $m^\star\in\mbb{N}$ be such that 
	\begin{align}
		\log m^\star = \inf_{m\in\mbb{N}}\Big\{\log m\;:\;\max_{\psi\in\pure(A\tA)}E_{(m)}\left(\mN^{\tA\to B}\left(\psi^{A\tA}\right)\right)\leq\eps\Big\}.
	\end{align}
	This implies that there exists an $\chi\in\pure(A\tA)$ such that 
	$E_{(m^\star)}\left(\mN^{\tA\to B}\left(\chi^{A\tA}\right)\right)\leq\eps$ and $E_{(m)}\left(\mN^{\tA\to B}\left(\chi^{A\tA}\right)\right)>\eps$ for $m=m^\star-1$ (and thus all $m<m^\star$, since $E_{(k)}\ge E_{(k+1)}$). Consequently,
	\begin{align}
		\log m^\star&> \max_{\psi\in\pure(A\tA)} \inf_{m\in\mbb{N}}\Big\{\log m\;:\;E_{(m)}\left(\mN^{\tA\to B}\left(\psi^{A\tA}\right)\right)\leq\eps\Big\} \nonumber \\ &\ge \inf_{m\in\mbb{N}}\Big\{\log m\;:\;E_{(m)}\left(\mN^{\tA\to B}\left(\chi^{A\tA}\right)\right)\leq\eps\Big\}\nonumber \\
		&=\log m^\star,
	\end{align}
	resulting in the promised contradiction. We thus showed that 
	\begin{align}\label{eq:maxOutside}
		\max_{\psi\in\pure(A\tA)} \inf_{m\in\mbb{N}}\Big\{\log m\;:\;E_{(m)}\left(\mN^{\tA\to B}\left(\psi^{A\tA}\right)\right)\leq\eps\Big\} = \inf_{m\in\mbb{N}}\Big\{\log m\;:\;\max_{\psi\in\pure(A\tA)}E_{(m)}\left(\mN^{\tA\to B}\left(\psi^{A\tA}\right)\right)\leq\eps\Big\}. 
	\end{align}
	We now turn to the main part of the proof.
	
	According to Lem.~\ref{lem:combined}, we have
	\be
	P^2\left(\Phi_m\xrightarrow{\locc}\mN^{A\to B}\right)\geq E_{(m)}\left(\mN^{A\to B}\right)\;.
	\ee
	Combining this with the definition of the entanglement cost in Eq.~\eqref{defcost} gives 
	\ba\label{steps}
	\cost^\eps(\mN)&\geq\inf_{m\in\mbb{N}}\Big\{\log m\;:\;E_{(m)}\left(\mN^{A\to B}\right)\leq\eps\Big\}\\
	&=\inf_{m\in\mbb{N}}\Big\{\log m\;:\;\max_{\psi\in\pure(A\tA)}E_{(m)}\left(\mN^{\tA\to B}\left(\psi^{A\tA}\right)\right)\leq\eps\Big\}\\
	\GG{Eq.~\eqref{eq:maxOutside}\to}&= \max_{\psi\in\pure(A\tA)}\inf_{m\in\mbb{N}}\Big\{\log m\;:\;E_{(m)}\left(\mN^{\tA\to B}\left(\psi^{A\tA}\right)\right)\leq\eps\Big\}\\
	\GG{Thm.~\ref{lem1291}\to}&=\max_{\psi\in\pure(A\tA)}\inf_{m\in\mbb{N}}\Big\{\log m\;:\;P^2\left(\Phi_m\xrightarrow{\locc}\mN^{\tA\to B}(\psi^{A\tA})\right)\leq\eps\Big\} \\
	&=\max_{\psi\in\pure(A\tA)}\cost^\eps\left(\mN^{\tA\to B}(\psi^{A\tA})\right)\\
	\GG{Thm.~\ref{cmaint}\to}&=\max_{\psi\in\pure(A\tA)}\inf_{\sigma^{XAB}} H_{\max}^{\eps}(A|X)_\sigma\;,
	\ea
	where the infimum in the last line is over all classical systems $X$ and all regular extensions $\sigma^{XAB}$ of $\mN^{\tA\to B}(\psi^{A\tA})$.

	For the converse inequality, expressing the cost in terms of the fidelity gives
	\ba
	\cost^\eps(\mN)
	&=\inf_{m\in\mbb{N}}\Big\{\log m\;:\;F\left(\Phi_m\xrightarrow{\locc}\mN\right)\geq\sqrt{1-\eps}\Big\}\\
	\Gg{\sqrt{1-\eps}\leq1-\eps/2\to}&\leq \inf_{m\in\mbb{N}}\Big\{\log m\;:\;F\left(\Phi_m\xrightarrow{\locc}\mN\right)\geq1-\eps/2\Big\}\\
	\GG{Lem.~\ref{lem:combined}\to}&\leq\inf_{m\in\mbb{N}}\Big\{\log m\;:\;E_{(m)}\left(\mN\right)\leq\eps/2\Big\}\\
	&=\max_{\psi\in\pure(A\tA)}\cost^{\eps/2}\left(\mN^{\tA\to B}(\psi^{A\tA})\right)\\
	&=\max_{\psi\in\pure(A\tA)}\inf_{\sigma^{XAB}} H_{\max}^{\eps/2}(A|X)_\sigma\;,
	\ea
	where the last two lines follow from the same steps as in Eq.~\eqref{steps} and the infimum is again over all classical systems $X$ and all regular extensions $\sigma^{XAB}$ of $\mN^{\tA\to B}(\psi^{A\tA})$. The fact that the infimums are achieved and can be taken over all classical extensions follows from the analogous statements in Thm.~\ref{cmaint}.
\end{proof}

\setcounter{theorems}{16} 
\begin{theorem}
		Let $\rho\in\md(AB)$ with $m=|A|=|B|$, $\sigma\in\md(A'B')$, and $\mN\in\locc_1(AB\to A'B')$. The coherent information of entanglement $E_\to$ is
		\begin{enumerate}
			\item monotonic under one-way LOCC, i.e., $$E_{\to}\left(\mN^{AB\to A'B'}\left(\rho^{AB}\right)\right)\leq E_{\to}\left(\rho^{AB}\right)\;,$$
			\item  non-negative, i.e., $E_\to(\rho^{AB})\geq 0$, with equality if $\rho^{AB}$ is separable,
			\item strongly monotonic under one-way LOCC, i.e., for any ensemble $\{p_y,\sigma_y^{A'B'}\}$ that can be obtained from $\rho^{AB}$ using one-way LOCC and subselection, it holds that
			\begin{align*}
				E_{\to}&\left(\rho^{AB}\right)
				\ge\sum_y p_y E_{\to}\left(\sigma^{A'B'}_y\right),
			\end{align*}
			\item convex,
			\item bounded by $E_\to\left(\rho^{AB}\right)\le E_\to\left(\Phi_m\right)=\log(m)$,
			\item superadditive, i.e., $$ \qquad E_\to\left(\rho^{AB}\otimes\sigma^{A'B'}\right)\geq E_\to\left(\rho^{AB}\right)+E_\to\left(\sigma^{A'B'}\right).$$  
		\end{enumerate}
	\end{theorem}
	
	\begin{proof}
		1. First observe that since $\mN\in\locc_1(AB\to A'B')$, for any $\mL \in \cptp(A'\to A'X)$  we have that $\mL\circ\mN\in\locc_1(AB\to A'B'X)$. Therefore, 
		\begin{align}\label{13p153}
			E_{\to}&\left(\mN^{AB\to A'B'}\left(\rho^{AB}\right)\right)\nonumber\\
			&=\sup_{\mL\in\cptp(A'\to A'X)}I\big(A'\ra B'X\big)_{\mL\circ\mN(\rho)}\nonumber\\
			&\leq \sup_{\mM\in\locc_1(AB\to A'B'X)}I\big(A'\ra B'X\big)_{\mM(\rho)}\;.
		\end{align}
		According to Eq.~\eqref{13111c}, every $\mM\in\locc_1(AB\to A'B'X)$ can be expressed as
		\be
		\mM^{AB\to A'B'X}\eqdef\sum_{y\in[n]}\mE^{A\to A'X}_y\otimes\mF_{(y)}^{B\to B'}\;,
		\ee
		with each $\mE_y^{A\to A'X}$ being a CP map such that $\sum_{y\in[n]}\mE_y\in\cptp(A\to A'X)$, and each $\mF_{(y)}\in\cptp(B\to B')$. Let further be
		\be\label{eq:cohInfProb}
		q_y\eqdef\tr\left[\mE_y^{A\to A'X}\left(\rho^{AB}\right)\right]
		\ee
		and
		\be \label{eq:cohInfStates}
		\sigma_y^{A'BX}\eqdef\frac1{q_y}\mE_y^{A\to A'X}\left(\rho^{AB}\right)\;.
		\ee
		With these notations, we have 
		\be
		\mM^{AB\to A'B'X}\left(\rho^{AB}\right)=\sum_{y\in[n]}q_y\mF^{B\to B'}_{(y)}\left(\sigma_y^{A'BX}\right)\;.
		\ee
		Remembering that
		\begin{align}\label{eq:ConCohInfRelEnt}
			D&\Big(\tau^{CD}\Big\|I^C\otimes \tau^D\Big)\nonumber \\
			&=\tr\left( \tau^{CD} \log \left(\tau^{CD} \right)\right)-\tr\left( \tau^{CD} \log \left(I^C\otimes\tau^D \right)\right) \nonumber \\
			&=\tr\left( \tau^{CD} \log \left(\tau^{CD} \right)\right)-\tr\left( \tau^{CD}  I^C\otimes\log\left(\tau^D \right)\right) \nonumber \\
			&=\tr\left( \tau^{CD} \log \left(\tau^{CD} \right)\right)-\tr\left( \tau^{D} \log \left(\tau^D \right)\right) \nonumber \\
			&=-H\left( \tau^{CD}\right)+ H\left( \tau^{D}\right) \nonumber \\
			&=-H(C|D)_\tau = I(C\ra D)_\tau,
		\end{align}
		it follows from the joint convexity of the relative entropy, the data processing inequality, and the inequality in Eq.~\eqref{13p153} that
		\begin{align}\label{eq:boundCohInfEnt}
			&E_{\to}\left(\mN^{AB\to A'B'}\left(\rho^{AB}\right)\right)\nonumber\\
			&\leq \sup_{\mM\in\locc_1(AB\to A'B'X)}I\big(A'\ra B'X\big)_{\mM(\rho)} \nonumber \\
			&=\sup D\Big(\sum_{y\in[n]}q_y\mF^{B\to B'}_{(y)}\left(\sigma_y^{A'BX}\right)\Big\| I^{A'}\otimes \sum_{y\in[n]}q_y\mF^{B\to B'}_{(y)}\left(\sigma_y^{BX}\right)\Big) \nonumber \\
			&\le \sup\! \sum_{y\in[n]}q_yD\Big(\mF^{B\to B'}_{(y)}\!\left(\sigma_y^{A'BX}\right)\!\Big\|I^{A'}\!\otimes\! \mF^{B\to B'}_{(y)}\left(\sigma_y^{BX}\right)\!\Big) \nonumber \\
			&\le \sup\! \sum_{y\in[n]}q_yD\Big(\sigma_y^{A'BX}\!\Big\|I^{A'}\!\otimes\! \sigma_y^{BX}\!\Big) \nonumber \\
			&=\sup\! \sum_{y\in[n]}q_y I\big(A'\ra BX\big)_{\sigma_y},
		\end{align}
		where the supremums are over all decompositions described above corresponding to an $\mM\in\locc_1(AB\to A'B'X')$. 
		
		Let $\{\tau_i\}\subset \md(C)$ be a set of states that are mutually orthogonal and $\{r_i\}$ a probability distribution. It is well known (and follows from a straightforward calculation) that this implies that 
		\begin{align}\label{eq:vonNeumannOrthog}
			H\left(\sum_i r_i \tau_i\right)=\sum_i r_iH(\tau_i)-\sum_i r_i \log(r_i).
		\end{align}
		If $\rho^{A'BX}\eqdef\sum_{x\in[n]}p_x\rho^{A'B}_x\otimes\ketbra{x}{x}^X$, i.e., is a quantum-classical-state in $\md(A'BX)$, this implies that
		\be\label{eq:FlagCohInf}
		I(A'\ra BX)_\rho=\sum_{x\in[n]}p_xI(A'\ra B)_{\rho_x}:
		\ee
		Let $\rho_x^{A'BX}=\rho_x^{A'B}\otimes\ketbra{x}{x}$. From Eqs.~\eqref{eq:ConCohInfRelEnt} and~\eqref{eq:vonNeumannOrthog}, we thus find that
		\begin{align}
			I(A'\ra BX)_\rho=&H\left(\rho^{BX}\right)-H\left(\rho^{A'BX}\right) \nonumber \\
			=&\sum_{x\in[n]} p_x H\left(\rho_x^{BX}\right)-\sum_{x\in[n]} p_x \log(p_x) \nonumber \\
			&-\sum_{x\in[n]} p_x H\left(\rho_x^{A'BX}\right)+\sum_{x\in[n]} p_x \log(p_x) \nonumber \\
			=& \sum_{x\in[n]}p_xI(A'\ra B)_{\rho_x},
		\end{align}
		which proves the claim.
		
		From Eq.~\eqref{eq:boundCohInfEnt}, and remembering the notation from Eqs.~\eqref{eq:cohInfProb} and~\eqref{eq:cohInfStates}, it then follows that
		\begin{align}
			&E_{\to}\left(\mN^{AB\to A'B'}\left(\rho^{AB}\right)\right)\nonumber\\
			&\le\sup \sum_{y\in[n]}q_y I\big(A'\ra BX\big)_{\sigma_y}\nonumber \\
			&=\sup  I\big(A'\ra BXY\big)_{\sum_{y\in[n]}q_y\sigma_y^{A'BX}\otimes \ketbra{y}{y}^Y} \nonumber \\
			&= \sup I\big(A'\ra BXY\big)_{\sum_{y\in[n]}\mE_y^{A\to A'X}\left(\rho^{AB}\right)\otimes \ketbra{y}{y}^Y},
		\end{align}
		where in the last line, the supremum is over all classical systems $Y$ with arbitrary dimension $n$, all dimensions of $X$, and all instruments $\{\mE_y^{A\to A'X}\}_{y\in[n]}$. 
		
		Now let $Z\eqdef XY$. Since $\mE^{A\to A'Z}(\rho^A)\eqdef \sum_{y\in[n]}\mE_y^{A\to A'X} \left(\rho^{A}\right)\otimes \ketbra{y}{y}^Y\in \cptp(A\to A'Z)$ for any instrument $\{\mE_y^{A\to A'X}\}_{y\in[n]}$, we find that 
		\begin{align}\label{eq:CohInfSupA}
			E_{\to}&\left(\mN^{AB\to A'B'}\left(\rho^{AB}\right)\right)\nonumber\\
			&\leq \sup_{\mE\in\cptp(A\to A'Z)}I\big(A'\ra BZ\big)_{\mE(\rho)},
		\end{align}
		where the supremum is also over all dimensions of the systems $A'$ and $Z$. We will conclude the proof by showing that the second line in the above equation equals $E_{\to}\left(\rho^{AB}\right)$, i.e., that we can restrict the supremum without loss of generality to the case $A'=A$ (compare to Def.~\ref{def:cohInf}): Observe first that the coherent information remains invariant under local isometries. If $|A'|\leq |A|$, this implies that the supremum over $\cptp(A\to AZ)$ cannot be smaller than the supremum over $\cptp(A\to A'Z)$ (since we can always use an isometric embedding). 
		
		On the other hand, suppose that $|A'|> |A|$. In general, we notice that in Eq.~\eqref{eq:CohInfSupA} (and similarly in Eq.~\eqref{13p149}), the supremum can be restricted to quantum channels of the form 
		\begin{align}\label{eq:SingleKraus}
			\mE^{A\to A'Z}=\sum_{x\in[n]}\mE_x^{A\to A'}\otimes \ketbra{x}{x}^Z,
		\end{align} 
		where each $\mE_x^{A\to A'}$ is a CP map with a single Kraus operator. This follows from Eq.~\eqref{eq:ConCohInfRelEnt} and the data processing inequality of the relative entropy: An arbitrary $\mE^{A\to A'Z}\in \cptp(A\to A'Z)$ can be written as 
		\begin{align}
			\mE^{A\to A'Z}=\sum_{x\in[n],y\in[m]}\mE_{x,y}^{A\to A'}\otimes \ketbra{x}{x}^Z,
		\end{align} 
		where each $\mE_{x,y}^{A\to A'}$ is a CP map with a single Kraus operator $M_{x,y}: A\to A'$. Introducing another classical system $\tilde{Z}$ of dimension $m$, 
		\begin{align}
			\tilde{\mE}^{A\to A'Z \tilde{Z}} \eqdef\sum_{x\in[n],y\in[m]}\mE_{x,y}^{A\to A'}\otimes \ketbra{x}{x}^Z\otimes\ketbra{y}{y}^{\tilde{Z}},
		\end{align} 
		is an element of $\cptp(A\to A'Z\tilde{Z})$ and
		\begin{align}
			\mE^{A\to A'Z}=\tr_{Z'}\circ \tilde{\mE}^{A\to A'Z \tilde{Z}}.
		\end{align}
		It thus follows that
		\begin{align}
			I&\big(A'\ra BZ\tilde{Z}\big)_{\tilde{\mE}(\rho)}\nonumber\\
			=  &D\Big(\tilde{\mE}(\rho^{AB})\Big\|I^{A'}\otimes \tr_{A'}\circ\tilde{\mE}(\rho^{AB})\Big) \nonumber \\
			\ge &D\Big(\tr_{Z'}\circ\tilde{\mE}(\rho^{AB})\Big\|I^{A'}\otimes \tr_{Z'}\circ\tr_{A'}\circ\tilde{\mE}(\rho^{AB})\Big) \nonumber \\
			=&D\Big(\mE(\rho^{AB})\Big\|I^{A'}\otimes \tr_{A'}\circ\mE(\rho^{AB})\Big)\nonumber \\
			=&I\big(A'\ra BZ\big)_{\mE(\rho)}.
		\end{align}
		Since the supremums include a supremum over the classical system anyway, this proves our claim. 
		
		Now consider a channel as in Eq.~\eqref{eq:SingleKraus}, i.e., each $\mE_x^{A\to A'}(\cdot)=M_x(\cdot)M_x^*$ is a CP map with a single Kraus operator $M_x: A\to A'$.
		According to the polar decomposition, each $M_x$ can be expressed as 
		$M_x=V_xN_x$, where each $N_x: A\to A$ is an element of a generalized measurement, and each $V_x: A\to A'$ is an isometry. Defining an isometry 
		\begin{align}
			V^{AZ\to A'Z}\eqdef\sum_x V_x^{A\to A'}\otimes \ketbra{x}{x}^Z,
		\end{align}
		and a channel
		\begin{align}
			\mN^{A\to AZ}(\cdot)\eqdef\sum_{x\in[n]}N_x(\cdot)N_x^*\otimes \ketbra{x}{x}^Z,
		\end{align}
		it holds that $\mE^{A\to A'Z}=\mV^{AZ\to A'Z}\circ\mN^{AZ\to AZ}$, where $\mV^{A\to A'Z}(\cdot)=V(\cdot)V^*$. Using again that the coherent information is invariant under local isometries, we find that 
		\begin{align}
			I\big(A'\ra BZ\big)_{\mE(\rho)}= I\big(A\ra BZ\big)_{\mN(\rho)},
		\end{align}
		which completes the proof of 1.
		
		2. 
		Let 
		\begin{align}
			\tilde{\mE}^{A\to AX}(\rho^A)=\tr\left(\rho^A\right)\ \ketbra{0}{0}^A\otimes \ketbra{0}{0}^X.
		\end{align}
		It then follows that
		\begin{align}\label{eq:NonNeg}
			E_{\to}&\left(\rho^{AB}\right)\nonumber\\
			=&\sup_{\mE\in\cptp(A\to AX)}I\big(A\ra BX\big)_{\mE(\rho)}\nonumber \\
			\ge&I\big(A\ra B\big)_{\tilde{\mE}(\rho)} \nonumber \\
			=& H\left(\rho^B\otimes\ketbra{0}{0}^X\right)-H\left(\ketbra{0}{0}^A\otimes \rho^B \otimes \ketbra{0}{0}^X\right) \nonumber \\
			=&0.
		\end{align}
		
		According to Ref.~\cite[Thm.~1]{Lieb1973b} (see also the announcement in Ref.~\cite{Lieb1973a} and Ref.~\cite{Wehrl1978} for a review), 
		\begin{align}
			I\big(A\ra B\big)_{\rho}=H(\rho^B)-H(\rho^{AB})
		\end{align}
		is convex in $\rho^{AB}$. This immediately implies that if $\rho^{AB}$ is separable, i.e.,
		\begin{align}
			\rho^{AB}=\sum_i p_i \phi_i^A\otimes \psi_i^B,
		\end{align}
		then~\cite{Cerf1997}
		\begin{align}
			I\big(A\ra B\big)_{\rho}\le \sum_i p_i \left(H(\psi_i^B)-H(\phi_i^A\otimes \psi_i^B)\right)=0. 
		\end{align}
		For any $\mE\in\cptp(A\to AX)$, $\mE(\rho^{AB})$ is separable if $\rho^{AB}$ was and thus
		\begin{align}
			E_{\to}\left(\rho^{AB}\right)
			=&\sup_{\mE\in\cptp(A\to AX)}I\big(A\ra BX\big)_{\mE(\rho)}\le 0 
		\end{align}
		on separable states. Together with Eq.~\eqref{eq:NonNeg}, this finishes the proof of 2. It also follows immediately from the operational interpretation of the coherent information detailed in Ref.~\cite{Horodecki2005}.
		
		3. Let $\rho^{AB}\in \md(AB)$ be a state. Any ensemble $\{p_y,\sigma_y^{A'B'}\}$ that can be obtained from it by means of one-way LOCC and subselection is of the form 
		\begin{align}
			p_{y}=&\tr\left( \mM_y^{A\to A'} \otimes \mF_{(y)}^{B\to B'} \left(\rho^{AB}\right)\right),\nonumber\\
			\sigma_{y}^{A'B'}=&\frac{1}{p_{y}} \mM_y^{A\to A'} \otimes \mF_{(y)}^{B\to B'} \left(\rho^{AB}\right),
		\end{align} 
		i.e., Alice applies an instrument $\{\mM_y\}_{y\in[n]}$, with $\mM_y\in\cp(A\to A')$ and $\sum_{y=1}^n\mM_y\in\cptp(A\to A')$, sends the outcome $y$ to Bob, who then, conditioned on $y$, applies a channel $\mF_{(y)}\in\cptp(B\to B')$~\cite{Chitambar2014}. 
		
		Now define a channel $\mN\in \cptp(AB\to A'A_sB' B_s)$ as 
		\begin{align}
			\mN=\sum_{y} \mM_y^{A\to A'} \otimes \ketbra{y}{y}^{A_s}\otimes \mF_{(y)}^{B\to B'}\otimes \ketbra{y}{y}^{B_s},
		\end{align}
		where $A_s$ is a classical system \textit{in Alice's possession} and $B_s$ a classical system \textit{in Bob's possession} that are used to store the outcome of our instrument. 
		
		Let $\mathfrak{C}(A'\to A'X|A_s)$ be the set of channels that can be written as
		\begin{align}
			\mE&^{A'A_s\to A' A_s X}(\tau^{A'A_s})\nonumber\\
			&=\sum_y \mE_{(y)}^{A'\to A' X} \tr_{A_s}\left(\ketbra{y}{y}^{A_s} \tau^{A'A_s}\right)  \otimes \ketbra{1}{1}^{A_s}
		\end{align}
		where each $\mE_{(y)}\in \cptp(A\to A'X)$.
		Clearly, $\mN\in\locc_1$, and $\mathfrak{C}(A'\to A'X|A_s)\subset\cptp(A'A_s \to A' A_s X)$. Observing that 
		\begin{align}
			\mN(\rho^{AB})=\sum_y p_y \sigma_y^{A'B'}\otimes \ketbra{y}{y}^{A_s}\otimes\ketbra{y}{y}^{B_s},
		\end{align}
		it thus follows from monotonicity and Eq.~\eqref{eq:FlagCohInf} that
		\begin{align}
			E_{\to}&\left(\rho^{AB}\right)\nonumber\\
			\ge& E_{\to}\left(\mN\left(\rho^{AB}\right)\right)\nonumber\\
			=&\sup_{\mE\in\cptp(A'A_s\to A'A_sX)}I\big(A'A_s\ra B'B_sX\big)_{\mE\circ\mN(\rho)}\nonumber \\
			\ge&\sup_{\mE\in\mathfrak{C}(A'\to A'X|A_s)}I\big(A'A_s\ra B'B_sX\big)_{\mE\circ\mN(\rho)}\nonumber \\
			=&\sup_{\{\mE_{(y)}\in\cptp(A'\to A'X)\}}I\big(A'A_s\ra B'B_sX\big)_{\sum_y p_y \mE_{(y)}\left(\sigma_y^{A'B'}\right)\otimes \ketbra{1}{1}^{A_s}\otimes\ketbra{y}{y}^{B_s}}\nonumber \\
			=&\sup_{\{\mE_{(y)}\in\cptp(A'\to A'X)\}}\sum_y p_y I\big(A'A_s\ra B'X\big)_{ \mE_{(y)}\left(\sigma_y^{A'B'}\right)\otimes \ketbra{1}{1}^{A_s}}\nonumber \\
			=&\sum_y p_y \sup_{\mE\in\cptp(A'\to A'X)\}}I\big(A'\ra B'X\big)_{ \mE\left(\sigma_y^{A'B'}\right)}\nonumber \\
			=&\sum_y p_y E_{\to}\left(\sigma^{A'B'}\right),
		\end{align}
		which finishes the proof.
		
		4. Convexity is inherited from the convexity of $I\big(A\ra B\big)$~\cite[Thm.~1]{Lieb1973b}: For $\rho^{AB}=t\sigma^{AB}+(1-t)\tau^{AB}$ and $t\in[0,1]$, 
		\begin{align}
			E_{\to}&\left(\rho^{AB}\right)\nonumber\\
			=&\sup_{\mE\in\cptp(A\to AX)}I\big(A\ra BX\big)_{\mE(\rho)}\nonumber\\
			=&\sup_{\mE\in\cptp(A\to AX)}I\big(A\ra BX\big)_{t\mE(\sigma)+(1-t)\mE(\tau)} \nonumber \\ 
			\le&\sup_{\mE\in\cptp(A\to AX)} \left(tI\big(A\ra BX\big)_{\mE(\sigma)}+(1-t)I\big(A\ra BX\big)_{\mE(\tau)}\right) \nonumber \\
			\le &t\sup_{\mE\in\cptp(A\to AX)} I\big(A\ra BX\big)_{\mE(\sigma)}\nonumber \\
			&+ (1-t)\sup_{\mE\in\cptp(A\to AX)}I\big(A\ra BX\big)_{\mE(\tau)} \nonumber \\
			=& t E_{\to}\left(\sigma^{AB}\right)+(1-t) E_{\to}\left(\tau^{AB}\right).
		\end{align}
		
		5. Let $\Phi_m\in\md(AB)$ be the maximally entangled state with $m\eqdef|A|=|B|$. This implies that 
		\begin{align}
			E_{\to}&\left(\Phi^{AB}\right)\ge I\big(A\ra BX\big)_{\Phi_m}=H\left(\Phi_m^B\right) - H\left(\Phi_m^{AB}\right)=\log(m).
		\end{align}
		Moreover, for any state $\rho\in\md(AB)$ and any channel $\mE\in\cptp(A\to AX)$  with $m=|A|=|B|$, we have that $\mE(\rho)$ is quantum-classical and we can thus apply Eq.~\eqref{eq:FlagCohInf} to find that
		\begin{align}
			E_{\to}&\left(\rho^{AB}\right)=\sup_{\mE\in\cptp(A\to AX)}I\big(A\ra BX\big)_{\mE(\rho^{AB})} \nonumber \\
			\le&\sup_{\{p_x,\sigma_x^{AB}\}} \sum_x p_x I\big(A\ra B\big)_{\sigma_x^{AB}} \nonumber \\
			=&\sup_{\{p_x,\sigma_x^{AB}\}} \sum_x p_x \left[H\left(\sigma_x^B\right) - H\left(\sigma_x^{AB}\right)\right] \nonumber \\
			\le&\sup_{\{p_x,\sigma_x^{AB}\}} \sum_x p_x \log(m)\nonumber \\
			=&\log(m).
		\end{align}
		In combination, this proves that
		\be
		E_\to\left(\Phi_m^{AB}\right)=\log(m).
		\ee 
		
		6. From the additivity of the von Neumann entropy and Eq.~\eqref{eq:ConCohInfRelEnt}, it follows that 
		\begin{align}
			I\big(AA'\ra BB'\big)_{\rho^{AB}\otimes\sigma^{A'B'}}=I\big(A\ra B\big)_{\rho^{AB}}+I\big(A'\ra B'\big)_{\sigma^{A'B'}}.
		\end{align}
		This implies that
		\begin{align}
			E_\to&\left(\rho^{AB}\otimes\sigma^{A'B'}\right)\nonumber \\
			=& \sup_{\mE\in\cptp(AA'\to AA'X)}I\big(AA'\ra BB'X\big)_{\mE(\rho^{AB}\otimes\sigma^{A'B'})}  \nonumber\\
			=& \sup_{\mE\in\cptp(AA'\to AA'XX')}I\big(AA'\ra BB'XX'\big)_{\mE(\rho^{AB}\otimes \sigma^{A'B'})}  \nonumber\\
			\ge&\sup_{\substack{\mE\in\cptp(A\to AX)\\
					\mE'\in\cptp(A'\to A'X')}}I\big(AA'\ra BB'XX'\big)_{\mE(\rho^{AB})\otimes \mE'(\sigma^{A'B'})}  \nonumber\\
			=&\sup_{\substack{\mE\in\cptp(A\to AX)\\
					\mE'\in\cptp(A'\to A'X')}}\left[I\big(A\ra BX\big)_{\mE(\rho^{AB})} +I\big(A'\ra B'X'\big)_{\mE'(\sigma^{A'B'})} \right]\nonumber\\
			=& E_\to\left(\rho^{AB}\right)+E_\to\left(\sigma^{A'B'}\right).
		\end{align}
	\end{proof}
	
	\begin{theorem}
		Let $\rho\in\md(AB)$ and $\eps\in(0,1/2)$. Then, the one-way $\eps$-single-shot distillable entanglement is bounded by
		\be
		\distill^\eps_{\to}\left(\rho^{AB}\right)\leq\frac1{1-2\eps}E_\to\left(\rho^{AB}\right)+\frac{1+\eps}{1-2\eps}h\left(\frac\eps{1+\eps}\right)\;,
		\ee 
		where $h(x)\eqdef-x\log x-(1-x)\log(1-x)$ is the binary Shannon entropy.
	\end{theorem}
	
	\begin{proof}
		Let $m\in\mbb{N}$ be such that $\distill^\eps_{\to}\left(\rho^{AB}\right)=\log m$, and thus 
		\begin{align}
			T\left(\rho^{AB}\xrightarrow{\locc_1}  \Phi^{A'B'}_m\right)\leq\eps.
		\end{align} 
		This means that $\rho^{AB}\xrightarrow{\locc_1}\sigma^{A'B'}$ for some state $\sigma\in\md(A'B')$ that is $\eps$-close to $\Phi_m^{A'B'}$. Therefore, from the monotonicity of $E_\to$ under one-way LOCC we get that
		\be\label{rsapbp}
		E_\to\left(\rho^{AB}\right)\geq E_\to\left(\sigma^{A'B'}\right)\;.
		\ee
		Next, we use the fact that $\sigma^{A'B'}$ is $\eps$-close to $\Phi_m$ to show that the right-hand side of the equation above cannot be much smaller than $\log(m)$. Indeed, combining
		the continuity of $I(A'\ra B')_\rho\eqdef-H(A'|B')_\rho$ (see  Ref.~\cite[Lem.~2]{Winter2016}), with the fact that $I(A'\ra B')_{\Phi_m}=\log m$ (this follows directly from Eq.~\eqref{eq:ConCohInfRelEnt}), gives
		\begin{align}
			E_\to&\left(\sigma^{A'B'}\right)  \nonumber\\
			&\geq I(A'\ra B')_\sigma \nonumber\\
			&\geq I(A'\ra B')_{\Phi_m}-2\eps\log m-(1+\eps)h\left(\frac\eps{1+\eps}\right) \nonumber\\
			&=(1-2\eps)\log(m)-(1+\eps)h\left(\frac\eps{1+\eps}\right)\;.
		\end{align} 
		The proof is concluded by combining this inequality with Eq.~\eqref{rsapbp}.
	\end{proof}

\begin{lemma}
    Let $\rho\in\md(AB)$, and $\Phi_m\in\md(A'B')$ be the maximally entangled state with $m\eqdef|A'|=|B'|$. Then, 
    \be
    T\left(\rho\xrightarrow{\locc_1}  \Phi_m\right)=P^2\left(\rho\xrightarrow{\locc_1}  \Phi_m\right)
    =1-\sup_{\mN\in\locc_1}\tr\left[\Phi_m\mN\left(\rho\right)\right]\;,
    \ee
    where the supremum is over all $\mN\in\locc_1(AB\to A'B')$, and $P$ is the purified distance as given in Eq.~\eqref{eq:purifieddistance}.
\end{lemma}

\begin{proof}
	Let $\mG\in\locc_1(A'B'\to A'B')$ be the twirling map introduced in Ref.~\cite{Horodecki1999} acting on any $\omega\in\md(A'B')$ as
	\ba\label{formmg}
	\mG\left(\omega\right)&\eqdef\int_{{U}(m)}dU\;(U\otimes\overline{U})\omega(U\otimes\overline{U})^*
	\ea
	where $dU$ denotes the uniform probability distribution on the unitary group proportional to the Haar measure. It was also shown in Ref.~\cite{Horodecki1999} that
	\ba
	\mG\left(\omega\right)&=\left(1-\tr\left[\Phi_m\omega\right]\right)\tau+\tr\left[\Phi_m\omega\right]\Phi_m\;,
	\ea
	where $\tau\in\md(A'B')$ is given by $\tau=(I-\Phi_m)/(m^2-1)$.
	From this follows that for all $\omega\in\md(A'B')$
	\be\label{13104}
	\frac12\left\|\mG\left(\omega\right)-\Phi_m\right\|_1=\big(1-\tr\left[\Phi_{m}\omega\right]\big)\frac12\left\|\tau-\Phi_{m}\right\|_1=1-\tr\left[\Phi_{m}\omega\right]\;,
	\ee
	where the last equality follows from the fact 
	\begin{align}
		\norm{\tau-\Phi_m}_1=\norm{\frac{I}{m^2-1}-\frac{m^2}{m^2-1}\Phi_m}_1=2.
	\end{align}
	From the data processing inequality and the the fact that  $\Phi_{m}$ is invariant under the twirling map $\mG$, it follows that for all $\mN\in\locc_1(AB\to A'B')$ and all $\rho\in\md(AB)$
	\be\label{13105}
	\frac12\left\|\mN\left(\rho\right)-\Phi_m\right\|_1\geq \frac12\left\|\mG\circ\mN\left(\rho\right)-\Phi_m\right\|_1\;.
	\ee
	Since $\mG\circ\mN$ is also an $\locc_1$ channel ($\mG$ can be implemented with shared randomness) it follows from the inequality above that the conversion distance can be expressed as
	\ba 
	T\left(\rho\xrightarrow{\locc_1}  \Phi_m\right)&=\inf_{\mN\in\locc_1}\frac12\left\|\mG\circ\mN\left(\rho\right)-\Phi_m\right\|_1\\
	\GG{\eqref{13104}\to}&=1-\sup_{\mN\in\locc_1}\tr\left[\Phi_{m}\mN\left(\rho\right)\right]\;.
	\ea
	This completes the proof.
\end{proof}

\setcounter{theorems}{20} 

\begin{theorem}
	Let $\mN\in\cptp(A\to B)$ be a quantum channel. It then holds that
	\be
	E_{\to}\left(\mN^{A\to B}\right)=I(A\ra B)_\mN\;.
	\ee
\end{theorem}

\begin{proof}
	From its definition (see Def.~\ref{def:cohInf}) we know that for all $\rho\in\md(AB)$, it holds that $E_\to(\rho^{AB})\geq I(A\ra B)_{\rho}$. Combining this with the definition in Eq.~\eqref{defcic} gives
	\ba
	E_{\to}\left(\mN^{A\to B}\right)&\geq\max_{\psi\in\pure(A\tA)}I(A\ra B)_{\mN^{\tA\to B}\left(\psi^{A\tA}\right)}\\
	&=I(A\ra B)_\mN\;.
	\ea
	
	To obtain the reverse inequality, observe that
	\ba
	\max_{\psi\in\pure(A\tA)}E_\to\left(\mN^{\tA\to B}\left(\psi^{A\tA}\right)\right)&=\sup_{\substack{\mE\in\cptp(A\to AX)\\ \psi\in\pure(A\tA)}}I\big(A\ra BX\big)_{\mE^{A\to AX}\otimes\mN^{\tA\to B}\left(\psi^{A\tA}\right)}\\
	\Gg{\substack{\text{Replacing }\mE^{A\to AX}\left(\psi^{A\tA}\right)\\\text{with arbitrary }\sigma\in\md(A\tA X)}\to}&\leq \sup_{\sigma\in\md(A\tA X)}I\big(A\ra BX\big)_{\mN^{\tA\to B}\left(\sigma^{A\tA X}\right)}\\
	\Gg{\sigma^{A\tA X}=\sum_{x\in[n]}p_x\sigma_x^{A\tA}\otimes \ketbra{x}{x}^X\to}&=\sup_{\sigma\in\md(A\tA X)}\sum_{x\in[n]}p_xI(A\ra B)_{\mN^{\tA\to B}\left(\sigma_x^{A\tA}\right)}\\
	&=\max_{\sigma\in\md(A\tA)}I(A\ra B)_{\mN^{\tA\to B}\left(\sigma^{A\tA}\right)}\\
	&=I(A\ra B)_\mN\;,
	\ea
	where we used that in Eq.~\eqref{eq:CohInfChan}, we can replace the maximum over all pure states in $\pure(A\tA)$ with a maximum over all mixed states in $\md(A\tA)$ because the coherent information is convex.
\end{proof}
Thm.~\ref{thm:DistChan} was proven in the main text.

\section{Computability}\label{app:computability}

In the following, we will provide the proof of Thm.~\ref{thm:effComp} presented in the main text as well as additional related results. To this end, we begin by showing that $\norm{\p^{\otimes n}}_{(k)}$ can be computed efficiently. Due to Thms.~\ref{thm:PureSingleDist} and \ref{thm:PureSingleCost} which express $\distill^\eps\left(\psi^{AB}\right)$ and $\cost^\eps\left(\psi^{AB}\right)$ in terms of the Ky-Fan norm, this will then allow us to derive the promised results. Moreover, we will provide explicit algorithms that can be used to compute all relevant quantities.

\setcounter{theorems}{21} 
\begin{lemma}\label{lem:evalKyFan}
Let $n,k,d  \geq 1$ be integers and $\p \in \prob^\downarrow(d)$. Algorithm~\ref{alg:compute Ky Fan norm} can be used to efficiently compute $\norm{\p^{\otimes n}}_{(k)}$.
\end{lemma}
\begin{proof}
Note that the $n$-fold tensor product $\p^{\otimes n}$ has $r:=\binom{n+d-1}{d-1}$ different terms $p_1^{n_1}p_2^{n_2} \cdots p_d^{n_d}$, where the $n_i$ are non-negative integers such that $n_1+n_2+\cdots+n_d = n$ (see multinomial coefficient), which is a polynomial in $n$ for fixed $d$. Each term $p_1^{n_1}p_2^{n_2} \cdots p_d^{n_d}$ repeats $\binom{n}{n_1,n_2,\cdots,n_d} = \frac{n!}{n_1!n_2!\cdots n_d!}$ times. First, we sort these $r$ terms in non-increasing order to obtain the ordered vector $(s_1, s_2, \cdots, s_{r})$. Let $v_i$ be the number of times that  $s_i$ repeats. From this, we get that
\begin{align}
    (\p^{\otimes n})^\downarrow = (\underbrace{s_1}_{v_1 \,\text{times}},\cdots,\underbrace{s_i}_{v_i \,\text{times}}, \cdots, \underbrace{s_r}_{v_r\, \text{times}}).
\end{align}
Let $N_0 \eqdef 0$, $N_k \eqdef \sum_{i=1}^k v_i$, and $P_k \eqdef \sum_{i=1}^k v_i s_i$. Then we have that $\|\p^{\otimes n}\|_{(N_k)} = P_k$ and for any $m \in  [N_k, N_{k+1}]$,
\begin{align}\label{eq:Alg1KyNorm}
\|\p^{\otimes n}\|_{(m)} = s_{k+1} (m - N_k) + P_{k} = s_{k+1}(m - N_{k+1}) + P_{k+1}.
\end{align}
Therefore, a complete algorithm is given in Algorithm~\ref{alg:compute Ky Fan norm}.
\end{proof}

\begin{algorithm}
\caption{Efficient evaluation of Ky-Fan norms}\label{alg:compute Ky Fan norm}
\LinesNumbered
\KwIn{$\p = (p_1,\cdots,p_d) \in \prob^\downarrow(d)$, integers $n,m \geq 1$}
\KwOut{$\norm{\p^{\otimes n}}_{(m)}$}

\BlankLine

Compute the $r\eqdef\binom{n+d-1}{d-1}$ different terms $p_1^{n_1}p_2^{n_2} \cdots p_d^{n_d}$ where $n_1+n_2+\cdots+n_d = n$\;

Sort the $r$ terms in non-increasing order resulting in the vector $(s_1, s_2, \cdots, s_{r})$. Let $v_i$ be the number of times that $s_i$ is repeated\;

Let $N := 0$, $P :=0$\;

\ForEach{$k \in [r]$}{
Let $N \leftarrow N + v_k$, and $P \leftarrow P + v_k s_k$\;

\If{$m < N$}{\Return $s_{k} (m - N) + P$\;
}
}

\end{algorithm}

At this point, let us mention that in Algorithm~\ref{alg:compute Ky Fan norm}, essentially, we first create an ordered vector and then search it. The search part, i.e., determining which interval $m$ belongs to, can be accomplished with (one of the variations of) binary search, see Algorithm~\ref{alg:compute Ky Fan norm bisection}. This reduces the required search time, but as a trade-off, one needs to compute and store all values of $N_k$ and $P_k$.

\begin{algorithm}[H]
\caption{Efficient evaluation of Ky-Fan norms via binary search}\label{alg:compute Ky Fan norm bisection}
\LinesNumbered
\KwIn{$\p = (p_1,\cdots,p_d) \in \prob^\downarrow(d)$, integers $n,m \geq 1$}
\KwOut{$\|\p^{\otimes n}\|_{(m)}$}

\BlankLine

Compute the $r\eqdef\binom{n+d-1}{d-1}$ different terms $p_1^{n_1}p_2^{n_2} \cdots p_d^{n_d}$ where $n_1+n_2+\cdots+n_d = n$\;

Sort the $r$ terms in non-increasing order resulting in the vector $(s_1, s_2, \cdots, s_{r})$. Let $v_i$ be the number of times that $s_i$ is repeated\;

Let $N_0 := 0$, $P_0 := 0$, and $\forall k \in [r]$, $N_k:= \sum_{i=1}^k v_i$, and $P_k := \sum_{i=1}^k v_i s_i$;

Let $a = 0$, $b = r$;

\While{True}{
Let $c = \lfloor (a+b)/2 \rfloor$;

\uIf{$N_c \leq m \leq N_{c+1}$}{
\Return{$s_{c+1} (m - N_c) + P_c$}\;
}\uElseIf{$m < N_c$}{
    $b \leftarrow c$;
}\Else{
 $a \leftarrow c + 1$\;
}

}

\end{algorithm}

In summary, we showed how $\norm{\p^{\otimes n}}_{(k)}$ can be computed efficiently. To compute $\distill^\eps\left(\psi^{\otimes n}\right)$, we must determine the integer $m$ satisfying  $\|\p^{\otimes n}\|_{(m-1)}\le \eps<\|\p^{\otimes n}\|_{(m)}$ (see Thm.~\ref{thm:PureSingleDist}, also for notation). That this can be done efficiently is the content of the following Proposition.

\begin{proposition}
For any $\p \in \prob^\downarrow(d)$, integer $n \geq 1$, and $\eps \in [0,1)$, the integer $\ell(\p^{\otimes n}, \eps):= \min\{m: \|\p^{\otimes n}\|_{(m)} > \eps\}$ can be  computed efficiently by Algorithms~\ref{alg:compute threshold} and \ref{alg:compute the threshold bisection}. 
\end{proposition}
 \begin{proof}
 The Proposition/Algorithm~\ref{alg:compute threshold}/Algorithm~\ref{alg:compute the threshold bisection} are essentially variations of Lem.~\ref{lem:evalKyFan}/Algorithm~\ref{alg:compute Ky Fan norm}/Algorithm~\ref{alg:compute Ky Fan norm bisection} and we will thus employ the notation used there. The principal idea is again that the elements of $\p^{\otimes n}$ repeat many times. After calculating again the ordered vector $(s_1, s_2, \cdots, s_{r})$ and the number of times  $v_i$ that  $s_i$ repeats, we do a ``rough'' search and determine 
 the integer $j$ such that $P_{j} \leq \varepsilon < P_{j+1}$, that is, determine which interval $\ell\eqdef \ell(\p^{\otimes n}, \eps)$ falls into (remember that $\|\p^{\otimes n}\|_{(N_k)} = P_k$). This can be done efficiently and implies that $N_j<\ell \le N_{j+1}$. Moreover, according to Eq.~\eqref{eq:Alg1KyNorm}, we know that for any $m$ in this interval,
 \begin{align}
 	\|\p^{\otimes n}\|_{(m)} = s_{j+1} (m - N_j) + P_{j} = s_{j+1}(m - N_{j+1}) + P_{j+1},
 \end{align}
and thus 
 \begin{align}
     \ell=\left \lfloor \frac{\eps - P_{j+1}}{s_{j+1}} + N_{j+1} \right \rfloor + 1.
 \end{align}
	This is exactly what Algorithm~\ref{alg:compute threshold} returns. Algorithm~\ref{alg:compute the threshold bisection} is again the variant employing binary search.
 \end{proof}

\begin{algorithm}[H]
\caption{Efficient evaluation of the threshold}\label{alg:compute threshold}
\LinesNumbered
\KwIn{$\p = (p_1,\cdots,p_d) \in \prob^\downarrow(d)$, integer $n \geq 1$, $\eps \in [0,1)$}
\KwOut{$\ell(\p^{\otimes n}, \eps)$}

\BlankLine

Compute the $r\eqdef\binom{n+d-1}{d-1}$ different terms $p_1^{n_1}p_2^{n_2} \cdots p_d^{n_d}$ where $n_1+n_2+\cdots+n_d = n$\;

Sort the $r$ terms in non-increasing order resulting in the vector $(s_1, s_2, \cdots, s_{r})$. Let $v_i$ be the number of times that $s_i$ is repeated\;

Let $N := 0$, $P :=0$\;

\ForEach{$k \in [r]$}{
Let $N \leftarrow N + v_k$, and $P \leftarrow P + v_k s_k$\;

\If{$\eps < P$}{\Return $\left \lfloor \frac{\eps - P}{s_k} + N \right \rfloor + 1$\;
}
}

\end{algorithm}

\begin{algorithm}[H]
\caption{Efficient evaluation of the threshold via binary search}\label{alg:compute the threshold bisection}
\LinesNumbered
\KwIn{$\p = (p_1,\cdots,p_d) \in \prob^\downarrow(d)$, integer $n \geq 1$, $\eps \in [0,1)$}
\KwOut{$\ell(\p^{\otimes n}, \eps)$}

\BlankLine

Compute the $r\eqdef\binom{n+d-1}{d-1}$ different terms $p_1^{n_1}p_2^{n_2} \cdots p_d^{n_d}$ where $n_1+n_2+\cdots+n_d = n$\;

Sort the $r$ terms in non-increasing order resulting in the vector $(s_1, s_2, \cdots, s_{r})$. Let $v_i$ be the number of times that $s_i$ is repeated\;

Let $N_0 := 0$, $P_0 := 0$, and $\forall k \in [r]$, $N_k:= \sum_{i=1}^k v_i$ and $P_k := \sum_{i=1}^k v_i s_i$;

Let $a = 0$, $b = r$;

\While{True}{
Let $c = \lfloor (a+b)/2 \rfloor$;

\uIf{$P_{c} \leq \eps < P_{c+1}$}{
\Return{$\left\lfloor \frac{\eps - P_c}{s_{c+1}} + N_c \right\rfloor + 1$}\;
}\uElseIf{$\eps < P_c$}{
    $b \leftarrow c$;
}\Else{
 $a \leftarrow c + 1$\;
}

}

\end{algorithm}

After we showed how to determine $\ell=\ell(\p^{\otimes n}, \eps)$, according to Thm.~\ref{thm:PureSingleDist}, what remains to do in order to compute $\distill^\eps\left(\psi^{\otimes n}\right)$ is to solve the optimization problem
\begin{align}
	\min_{k\in\{\ell,\ldots,d^n\}} \log\left\lfloor\frac{k}{\|\p^{\otimes n}\|_{(k)}-\eps}\right\rfloor.
\end{align}
That this can be done efficiently is a consequence of the following Proposition.

\begin{proposition}\label{prop:compDist}
For any $\p \in \prob^\downarrow(d)$, integer $n \geq 1$, and $\eps \in [0,1)$, let $\ell = \ell(\p^{\otimes n}, \eps)$. Then
\begin{align}
f_{\min}\eqdef \min_{k\in \{\ell, \cdots, d^n\}} f(k) \quad \text{with}\quad f(k):= \frac{k}{\|\p^{\otimes n}\|_{(k)}-\eps}   
\end{align}
can be efficiently computed by Algorithm~\ref{alg:compute distillable rate}.
\end{proposition}

\begin{proof}
If $\ell = d^n$, the minimum is taken at $\ell$. Otherwise, for any $k\in \{\ell, \cdots, d^n-1\}$,
\begin{align}
f(k+1) - f(k) & = \frac{k+1}{\|\p^{\otimes n}\|_{(k+1)}-\eps} - \frac{k}{\|\p^{\otimes n}\|_{(k)}-\eps}
 = \frac{\|\p^{\otimes n}\|_{(k)}- k (\p^{\otimes n})^\downarrow_{k+1} -\eps}{(\|\p^{\otimes n}\|_{(k+1)}-\eps)(\|\p^{\otimes n}\|_{(k)}-\eps)}.
\end{align}
Let $g(k) := \|\p^{\otimes n}\|_{(k)}- k (\p^{\otimes n})^\downarrow_{k+1} -\eps$ and thus
\begin{align}
    g(k+1) - g(k) = (k+1)((\p^{\otimes n})^\downarrow_{k+1} - (\p^{\otimes n})^\downarrow_{k+2}) \geq 0.
\end{align}
This means that $g(k)$ is non-decreasing in $k$. Now consider three cases: If $g(\ell) \geq 0$, then $g(k) \geq 0$ for all $k\in \{\ell, \cdots, d^n-1\}$. This implies that $f(k)$ is non-decreasing  in $k\in \{\ell, \cdots, d^n-1\}$ (remember that by definition of $\ell$, $\|\p^{\otimes n}\|_{(k)}-\eps>0$). The minimum of $f(k)$ is taken at $\ell$. If $g(d^n - 1) \leq 0$, then $f(k)$ is non-increasing and the minimum of $f(k)$ is taken at $d^n$. 
It remains to consider the case that there is a $k^*\in \{\ell+1, \cdots, d^n-2\}$ such that $g(k^* - 1) < 0 \leq g(k^*)$. In this case, it holds that for any $k \leq k^* - 1$, $g(k) < 0$, and thus $f(k+1) < f(k)$. Moreover, for any $k \geq k^*$, $g(k) \geq 0$ and thus $f(k+1) \geq f(k)$.  That is, $f(k)$ is decreasing in $k$ for $k \in \{\ell,\cdots, k^*\}$ and non-decreasing in $k$ for $k \in \{k^*,\cdots, d^n\}$. So the minimum of $f(k)$ is taken at $k^*$, which can be located via the bisection method. This takes at most $\log (d^n) = n \log d$ steps. 
\end{proof}

\begin{algorithm}
\caption{Efficient evaluation of the distillable entanglement}\label{alg:compute distillable rate}
\LinesNumbered
\KwIn{$\p = (p_1,\cdots,p_d) \in \prob^\downarrow(d)$, integer $n \geq 1$, $\eps \in [0,1)$}
\KwOut{$f_{\min}$}

\BlankLine

Use Algorithm~\ref{alg:compute threshold} to compute $\ell$\;

\If{$\ell = d^n$}{
\Return $f_{\min} = \frac{d^n}{1-\eps}$\;
}

\BlankLine

Use Algorithm~\ref{alg:compute Ky Fan norm} to compute $g(\ell)$\;

\If{$g(\ell) \geq 0$}{
\Return $f_{\min} = f(\ell)$ using Algorithm~\ref{alg:compute Ky Fan norm}\;
}

\BlankLine

Use Algorithm~\ref{alg:compute Ky Fan norm} to compute $g(d^n - 1)$\;

\If{$g(d^n - 1) \leq 0$}{
\Return $f_{\min} = f(d^n) = \frac{d^n}{1-\eps}$\;
}

\BlankLine

Let $a = \ell$ and $b = d^n - 1$\;
\While{$b > a + 1$}{
Let $c = \lfloor (a+b)/2 \rfloor$\;
\uIf{$g(c) \geq 0$}{
$b \leftarrow c$\;
}\Else{$a \leftarrow c$\;}
}

\Return $f_{\min} = f(b)$ using Algorithm~\ref{alg:compute Ky Fan norm}\;

\end{algorithm}

\setcounter{theorems}{4} 
\begin{theorem}
	Let $n\in\mbb{N}$, $\eps\in[0,1)$, and $\psi\in\pure(AB)$. This implies that both $\distill^\eps\left(\psi^{\otimes n}\right)$ and $\cost^\eps\left(\psi^{\otimes n}\right)$ can be computed efficiently.   
\end{theorem} 
\begin{proof}
	That $\distill^\eps\left(\psi^{\otimes n}\right)$ can be computed efficiently follows from Thm.~\ref{thm:PureSingleDist}, Prop.~\ref{prop:compDist}, and the observation that 
	\begin{align}
		\min_{k\in\{\ell,\ldots,d^n\}} \log\left\lfloor\frac{k}{\|\p^{\otimes n}\|_{(k)}-\eps}\right\rfloor=  \log\left\lfloor \min_{k\in\{\ell,\ldots,d^n\}}\frac{k}{\|\p^{\otimes n}\|_{(k)}-\eps}\right\rfloor.
	\end{align}
	Determining $\cost^\eps\left(\psi^{\otimes n}\right)$ is even easier: According to Thm.~\ref{thm:PureSingleCost}, we only need to determine the integer $m\in[d]$ satisfying $	\|\p\|_{(m-1)}<1-\eps\leq\|\p\|_{(m)}$. It is straightforward to see that this can be achieved by running a slight variation of Algorithm~\ref{alg:compute threshold} which is presented in Algorithm~\ref{alg:compute threshold cost}. Of course, a similar adaptation is possible for Algorithm~\ref{alg:compute the threshold bisection}.
\end{proof}

\begin{algorithm}[H]
	\caption{Efficient evaluation of the entanglement cost }\label{alg:compute threshold cost}
	\LinesNumbered
	\KwIn{$\p = (p_1,\cdots,p_d) \in \prob^\downarrow(d)$, integer $n \geq 1$, $\eps \in [0,1)$}
	\KwOut{$m$}

	\BlankLine
	
	Compute the $r\eqdef\binom{n+d-1}{d-1}$ different terms $p_1^{n_1}p_2^{n_2} \cdots p_d^{n_d}$ where $n_1+n_2+\cdots+n_d = n$\;
	
	Sort the $r$ terms in non-increasing order resulting in the vector $(s_1, s_2, \cdots, s_{r})$. Let $v_i$ be the number of times that $s_i$ is repeated\;
	
	Let $N := 0$, $P :=0$\;
	
	\ForEach{$k \in [r]$}{
		Let $N \leftarrow N + v_k$, and $P \leftarrow P + v_k s_k$\;
		
		\If{$1-\eps \le P$}{\Return $\log \left \lceil \frac{1- \eps - P}{s_k} + N \right \rceil$\;
		}
	}
\end{algorithm}

\section{Closed-form formula for Ref.~\cite{Regula2019}}\label{app:Regula}
In this Appendix, we provide an efficient method to evaluate the fidelity of distillation defined in Eq.~\eqref{eq:fidDist} as well as $E_{D}^{(1),\eps}(\psi^{\otimes n})$. For an arbitrary pure state $\ket{\psi}\in \pure(AB)$, let $\p = (p_1,\cdots,p_{|A|})$ be its Schmidt vector (ordered non-increasingly according to our convention). Moreover, let $\sqrt{\p} = (\sqrt{p_1},\cdots,\sqrt{p_{|A|}})$. The fidelity of distillation is then given by~\cite{Regula2019}
\begin{align}
    F(\psi^{\otimes n}, m) = \frac{1}{m} \left\|\sqrt{\p^{\otimes n}}\right\|_{[m]}^2,
\end{align}
where the distillation norm can be expressed as~\cite{Regula2018} 
\begin{align}
 \|\sqrt{\p^{\otimes n}}\|_{[m]} & := \|(\sqrt{\p})^{\otimes n}\|_{(m-k^*)} +\sqrt{k^* (1- \|\p^{\otimes n}\|_{(m-k^*)})}
\end{align}
with 
\begin{align}
    k^* = \argmin_{1\leq k\leq m} \frac{1- \|\p^{\otimes n}\|_{(m-k)}}{k}
\end{align}
and $\|\p^{\otimes n}\|_{(0)} := 0$. To determine the fidelity of distillation, we thus need to find $k^*$ first.

\begin{algorithm}
	\caption{Efficient determination of $k^*$}\label{alg:compute distillable norm}
	\LinesNumbered
	\KwIn{$\p = (p_1,\cdots,p_d) \in \prob^\downarrow(d)$, integer $n \geq 1,m \geq 2$}
	\KwOut{$k^*$ necessary to determine $F(\psi^{\otimes n}, m)$}

	\BlankLine
	Use Algorithm~\ref{alg:compute Ky Fan norm} to compute $t(1)$\;
	
	\If{$t(1) \geq 0$}{
		\Return $k^* = 1$\;
	}
	
	\BlankLine
	Use Algorithm~\ref{alg:compute Ky Fan norm} to compute $t(m-1)$\;
	\If{$t(m-1) \leq 0$}{
		\Return $k^* = m$\;
	}
	
	\BlankLine
	
	Let $a = 1$ and $b = m-1$\;
	
	\While{$b > a + 1$}{
		Let $c = \lfloor (a+b)/2 \rfloor$\;
		\uIf{$t(c) \geq 0$}{
			$b \leftarrow c$\;
		}\Else{$a \leftarrow c$\;}
	}
	
	\Return{$k^* = b$}\;
	
\end{algorithm}

Let 
\begin{align}
h(k) = \frac{1- \|\p^{\otimes n}\|_{(m-k)}}{k}.
\end{align}
Then for any $k \in [m-1]$,
\begin{align}
h(k+1) - h(k) = \frac{\|\p^{\otimes n}\|_{(m-k)} + k (\p^{\otimes n})^\downarrow_{m-k} - 1}{k(k+1)}.
\end{align}
Defining
\begin{align}
    t(k) = \|\p^{\otimes n}\|_{(m-k)} + k (\p^{\otimes n})^\downarrow_{m-k} - 1 = (k+1)\|\p^{\otimes n}\|_{(m-k)} - k \|\p^{\otimes n}\|_{(m-k-1)} - 1,
\end{align}
we have
\begin{align}
    t(k+1) - t(k) = (k+1)((\p^{\otimes n})_{m-k-1}^\downarrow - (\p^{\otimes n})^\downarrow_{m-k}) \geq 0.
\end{align}
So $t(k)$ is non-decreasing in $k \in [m-1]$. Analogously to the proof of Prop.~\ref{prop:compDist}, we can now consider three cases: 
If $t(1) \geq 0$, this implies that $h(k)$ is non-decreasing in $\{1,\cdots,m\}$ and $k^* = 1$. If $t(m-1) \leq 0$, $h(k)$ is non-increasing and we get $k^* = m$. Otherwise, there exists a $k'\in\{2,\cdots,m-2\}$ such that $t(k'-1)<0\le t(k')$ and $h(k)$ is again decreasing in $k$ for $k\in\{1,\cdots,k'\}$ and non-increasing for $k\in\{k',\cdots,m\}$. This implies that $k^* = k'$ which can be located via the bisection method. This takes at most $\log (d^n) = n \log d$ steps. The corresponding Algorithm is provided as Algorithm~\ref{alg:compute distillable norm}. Once $k^*$ is determined, we can use Algorithm~\ref{alg:compute Ky Fan norm} to compute
\begin{align*}
	F(\psi^{\otimes n}, m) = \frac{1}{m} \norm{\sqrt{\p^{\otimes n}}}_{[m]}^2 & = \frac{1}{m} \left(\|(\sqrt{\p})^{\otimes n}\|_{(m-k^*)} +\sqrt{k^* \left(1- \|\p^{\otimes n}\|_{(m-k^*)}\right)}\right)^2. 
\end{align*}

According to Ref.~\cite[Lem.~1]{Vidal2000}, the maximal achievable fidelity between a state obtained by applying LOCC to any initial state $\rho\in \md(AB)$ with a maximally entangled state of dimension $m>|A|$ is given by $\sqrt{|A|/m}$. This implies that 
\begin{align}
	\sqrt{\frac{|A|^n}{m}} \ge F(\psi^{\otimes n}, m),
\end{align} 
and thus $E_{D}^{(1),\eps}(\psi^{\otimes n})\le \log \left\lfloor\frac{|A|^n}{(1-\eps)^2}\right\rfloor$. From this follows that Algorithm~\ref{alg:compute distillable rate 1} efficiently computes $E_{D}^{(1),\eps}(\psi^{\otimes n})$.

\begin{algorithm}[t]
\caption{Efficient evaluation of the distillable entanglement of Ref.~\cite{Regula2019}}\label{alg:compute distillable rate 1}
\LinesNumbered
\KwIn{$\p = (p_1,\cdots,p_d) \in \prob^\downarrow(d)$, integer $n \geq 1$, $\eps \in [0,1)$}
\KwOut{$E_{D}^{(1),\eps}(\psi^{\otimes n})$}

\BlankLine

\If{$F(\psi^{\otimes n}, 2) < 1-\eps$}{\Return{$0$};}

\If{$F(\psi^{\otimes n}, \log \left\lfloor\frac{|A|^n}{(1-\eps)^2}\right\rfloor) \geq 1-\eps$}{\Return{$\log \log \left\lfloor\frac{|A|^n}{(1-\eps)^2}\right\rfloor$};}

\BlankLine

Let $a = 2$ and $b = \log \left\lfloor\frac{|A|^n}{(1-\eps)^2}\right\rfloor$\;

\While{$b > a + 1$}{
Let $c = \lfloor (a+b)/2 \rfloor$\;
\uIf{$F(\psi^{\otimes n},c) < 1-\eps$}{
$b \leftarrow c$\;
}\Else{$a \leftarrow c$\;}
}

\Return{$\log a$}

\end{algorithm}

\end{document}